\documentclass{article}
\usepackage{adjustbox}
\usepackage{microtype}
\usepackage{graphicx}
\usepackage{subfigure}
\usepackage{enumitem}
\usepackage{booktabs} % for professional tables
\usepackage{bbm}
\usepackage{hyperref}
\usepackage{algorithm}
\usepackage[noend]{algorithmic}

\usepackage{subcaption}

\usepackage[preprint]{neurips_2026}
\usepackage{amsmath}
\usepackage{amssymb}
\usepackage{mathtools}
\usepackage{amsthm}
\usepackage[capitalize,noabbrev]{cleveref}

\usepackage{tikz}
\usepackage{tikz-cd}
\usetikzlibrary{automata}
\usetikzlibrary{arrows.meta,positioning,decorations.pathreplacing}

\theoremstyle{plain}
\newtheorem{theorem}{Theorem}[section]
\newtheorem{proposition}[theorem]{Proposition}
\newtheorem{example}[theorem]{Example}

\newtheorem{lemma}[theorem]{Lemma}
\newtheorem{corollary}[theorem]{Corollary}
\theoremstyle{definition}
\newtheorem{definition}[theorem]{Definition}
\newtheorem{assumption}[theorem]{Assumption}
\newtheorem{remark}[theorem]{Remark}

\def\est{\mathrm{est}}
\def\E{\mathbb{E}}
\def\N{\mathbb{N}}
\DeclareMathOperator*{\argmax}{arg\,max} 

\usepackage{pgffor}
\foreach \x in {A,...,Z}{
	\expandafter\xdef\csname m\x\endcsname{\noexpand\mathbf{\x}}
}
\foreach \x in {A,...,Z}{
	\expandafter\xdef\csname om\x\endcsname{\noexpand\overline{\noexpand\mathbf{\x}}}
}
\foreach \x in {A,...,Z}{
	\expandafter\xdef\csname c\x\endcsname{\noexpand\mathcal{\x}}
}

\makeatletter
\DeclareRobustCommand\widecheck[1]{{\mathpalette\@widecheck{#1}}}
\def\@widecheck#1#2{%
    \setbox\z@\hbox{\m@th$#1#2$}%
    \setbox\tw@\hbox{\m@th$#1%
       \widehat{%
          \vrule\@width\z@\@height\ht\z@
          \vrule\@height\z@\@width\wd\z@}$}%
    \dp\tw@-\ht\z@
    \@tempdima\ht\z@ \advance\@tempdima2\ht\tw@ \divide\@tempdima\thr@@
    \setbox\tw@\hbox{%
       \raise\@tempdima\hbox{\scalebox{1}[-1]{\lower\@tempdima\box
\tw@}}}%
    {\ooalign{\box\tw@ \cr \box\z@}}}
\makeatother

\usepackage[textsize=tiny]{todonotes}

\newcommand{\nocontentsline}[3]{}
\let\origtoc=\addcontentsline
\let\addcontentsline=\nocontentsline

\title{Learning Approximate Nash Equilibria in Cooperative Multi-Agent RL via Mean-Field Subsampling}

\author{
Emile Anand\\ 
Georgia Institute of Technology \\
Atlanta, GA \\
\texttt{emile@gatech.edu} 
\And Ishani Karmarkar \\ 
Stanford University\\
Palo Alto, CA \\ 
\texttt{ishanik@stanford.edu}
}

\begin{document}

\maketitle

\begin{abstract} 
Many large-scale platforms and networked control systems have a centralized decision maker interacting with a massive population of agents under strict observability constraints. Motivated by such applications, we study a cooperative Markov game with a global agent and $n$ homogeneous local agents in a communication-constrained regime, where the global agent only observes a subset of $k$ local agent states per time step. We propose an alternating learning framework (\texttt{ALTERNATING-MARL}), where the global agent performs subsampled mean-field $Q$-learning against a fixed local policy, and local agents update by optimizing in an induced MDP. We prove that these approximate best-response dynamics converge to an $\tilde{O}(1/\sqrt{k})$-approximate Nash Equilibrium, while separating the sample complexities between the joint state and action spaces. Finally, we validate our results in numerical simulations for multi-robot control.
\end{abstract}

\section{Introduction}

Large-scale platforms and networked control systems often feature a \emph{global} decision maker interacting with a massive population of distributed \emph{local} agents under strict communication and observability constraints. Examples include online marketplaces where a platform sets pricing or matching rules for many users \cite{yang2025agentexchangeshapingfuture, 10.1145/3308558.3313433, 10.1145/3269206.3272021, pmlr-v125-jin20a, chaudhari2025peer}, networked control systems where a central coordinator intermittently communicates with distributed subsystems \cite{qu2021scalable,DBLP:journals/corr/abs-2006-06555,NEURIPS2023_a7a7180f,pmlr-v247-lin24a,9351818, shalevshwartz2016safemultiagentreinforcementlearning, deweese2024locally}, and robotic swarms that coordinate via low-bandwidth signals \citep{hespanha2007survey,oliehoek2016concise,7989376, lv2024localinformationaggregationbased, aina2025deepreinforcementlearningmultiagent, chiun2025marvelmultiagentreinforcementlearning}.

Drawing on the success of reinforcement learning (RL), multi-agent reinforcement learning (MARL) has emerged as a powerful framework for modeling such complex, large-scale networked systems, where the objective is to learn optimal policies that maximize the collective reward for the system \cite{bowling2000analysis}. However, in such settings, a fully centralized multi-agent RL approach is typically infeasible as the space of joint policies grows exponentially with the population size. Moreover, the global agent typically cannot observe the full joint state of all local agents at each decision epoch due to communication constraints. Likewise, local agents may only observe their own local state and a limited global context due to communication or privacy constraints \cite{Yang_2023}. These informational and computational limitations violate the assumptions underlying  centralized MARL, where policies and value functions are defined over the joint state space of all agents.\looseness=-1

Motivated by this setting, we study a cooperative setting with one \emph{global} agent and $n$ \emph{homogeneous}
local agents.  
Note that naively, even in centralized MARL, the size of policy search space is exponential in $n$, which is intractable for moderate $n$ \citep{daskalakis2007computing,Blondel_Tsitsiklis_2000, anand2024efficientreinforcementlearningglobal}. Motivated by the practical importance of decentralized MARL, this work explores an even more challenging \emph{communication-constrained} setting. Concretely, we assume that the global agent can only observe (and hence, can only condition its policy on) its own state-action pair and the states of subsets of $k \ll n$ local agents, and that each local agent can only observe their own state and action, and the global agent's state. In this communication-constrained setting, the agents  do not have sufficient information to learn or deploy a \emph{globally} optimal behavior policy, as one aims to do in centralized MARL. This is because, the globally optimal behavior policy would, in general, be a mapping over the full joint state space of all agents. In contrast, in the communication-constrained scenario, agents are unable to observe the full joint state space, and hence such policies are neither learnable nor deployable. \looseness=-1

Consequently, as in prior work \citep{jin2021v,10.5555/3540261.3542236, xie2021learning, 10.5555/3618408.3618858, lanctot2017unified, balsells2025scaling,delafuente2024gametheorymultiagentreinforcement,yang2025gametheoreticmultiagentreinforcementlearning}, we must settle for a \emph{locally} optimal behavior policy, or a Nash equilibrium (Definition~\ref{def: a nash equilibrium}). That is, we seek to find a behavior policy for the global agent $\pi_g(s_g, s_1, ..., s_k)$ (parameterized by its own state and the states of $k \ll n$ local agents) and a behavior policy for the local agent $\pi_l(s_i, s_g)$ (parameterized by its own state and the global state). Such a pair of policies $(\pi_g, \pi_l)$ is said to form a Nash equilibrium under the reward model if neither the global agent nor the local agent has any incentive to deviate from its policy. In this paper, we focus on computing an \emph{approximate} Nash equilibrium for this system (Definition~\ref{def:approx-nash}).\looseness=-1

\textbf{Contributions.} We propose \texttt{ALTERNATING-MARL}, a framework that couples a \emph{subsampled mean-field} Q-learning procedure for the global agent \cite{anand2024efficientreinforcementlearningglobal,anand2025meanfield}
with a generic local-agent learning procedure that acts as an approximate best-response oracle in the induced local MDP.
The framework alternates between updating the global agent's policy (global update) and updating the local agents' policy (local update) according to the best-response dynamics, where a 
\emph{global update} fixes the local policy and learns a near-best-response global policy using only $k$ sampled local agents, and a 
\emph{local update} fixes the global policy and updates the shared local policy using any RL routine with a provable performance
guarantee. We show this converges to an $\smash{\tilde{O}(1/\sqrt{k})}$-approximate   Nash equilibrium with high probability. We also bound the sample complexity of this process, removing the exponential dependence on the size of the action space from prior works. Moreover, at $\smash{k=O(\log n)}$, \texttt{ALTERNATING-MARL} achieves a polylogarithmic sample complexity in $n$ and decouples the prohibitive dependence on the size of the action space. While our primary focus is on providing a theoretical analysis of the method, we provide numerical simulations and implementation details of \texttt{ALTERNATING-MARL} for a multi-robot control setting in \cref{sec: numerical simulations}.

\textbf{Markov games framework.} Under our cooperative multi-agent framework, we show that the problem of finding a Nash equilibrium of the $n$-agent system reduces to finding the Nash equilibrium of a Markov potential
game, which is a special class of two-player games where the players represent the global agent $g$ and a ``representative'' local agent $l$  \cite{doi:10.1073/pnas.39.10.1095, littman}. For this Markov potential game, we show that approximate best-response dynamics, in which we alternately optimize $g$ and $l$'s strategy, holding the other player's strategy fixed, monotonically improves a common potential \cite{chen2022convergencepriceanarchyguarantees}. We leverage this structure to show that this dynamic reaches an $\epsilon$-approximate Nash equilibrium, where the error $\epsilon$ scales with the iterative approximation errors of each agent's learning routine.  However, efficiently implementing these best-response dynamics remains non-trivial in our communication-constrained setting; instead, our approach leverages that, when the policy for $l$ is held fixed, $g$'s best response corresponds to optimizing a particular subsampled induced MDP (and vice versa). For the global agent, we extend prior work of \cite{anand2025meanfield} in mean-field sampling and value iteration to efficiently compute approximately optimal best-responses. For the local agents, we design novel chained-episodic MDP reductions and apply techniques from \cite{dann2015sample}. We discuss these details further in Section~\ref{sec:approach-algorithms}.\looseness=-1

\subsection{Related Work}

\textbf{Markov games and equilibrium learning.}
Markov games generalize MDPs to multi-agent settings and admit equilibrium notions such as Nash equilibria
and related refinements \cite{NIPS1999_464d828b}. Its foundations trace back to stochastic games \citep{doi:10.1073/pnas.39.10.1095} and Markov games in MARL
\citep{littman}. A long line of work studies sample-efficient learning in Markov games, including algorithmic and
statistical guarantees under various structural assumptions (e.g., tabular or function approximation) \cite{jin2021v,10.5555/3540261.3542236,10.5555/3600270.3600377, bichuch2024stackelberg,10.5555/3586589.3586718}. Aligning with
this tradition, we focus on a large-population cooperative structure with communication constraints.\looseness=-1
   
\textbf{Mean-field methods for large-population MARL.} 
Mean-field MARL and mean-field games address the curse of dimensionality by exploiting \emph{exchangeability} and \emph{aggregation}:
agents interact with population-level statistics rather than full joint configurations \citep{lasry2018mean, yang2019sample, gu2022meanfieldmultiagentreinforcementlearning, gu2021meanfieldcontrolsqlearningcooperative, anand2025meanfield}. In cooperative MARL, mean-field style
approximations replace dependence on all agents' states/actions with dependence on an empirical distribution or mean action \cite{10.1145/3308558.3313433,cui2022learning,carmona2013probabilistic,subramanian2022decentralized}.
Our work uses a complementary perspective: rather than requiring access to the full population statistic, we analyze
\emph{subsampled} mean-field statistics and quantify the approximation error when sampling $k$ agents.\looseness=-1

\textbf{Stackelberg and leader--follower learning.}
Many applications motivating our model can also be viewed through a leader-follower lens, where a global entity commits to a
policy and local agents respond \cite{bacchiocchi2024samplecomplexitystackelberggames, 10310098, 4597505, bruckner2011stackelberg, yu2022learningcorrelatedstackelbergequilibrium}. Recent works study sample-efficient learning of Stackelberg equilibria in general-sum games and Deep-RL \citep{10.5555/3540261.3542236, 10.5555/3618408.3618858, balcan2025learningstructuredstackelberggames, gan2025robust}.
While our focus is similar, we instead analyze a cooperative objective with  partial observability and use the potential-game structure
to obtain convergence guarantees for alternating approximate best-response updates. \looseness=-1

\textbf{Potential games and best-response dynamics.}
Potential games \citep{Monderer1996, guo2025markovalphapotentialgames, arefizadeh2024characterizationspotentialordinalpotential} provide a powerful framework where unilateral improvements align with ascent in a
shared potential, implying convergence of best-response dynamics under suitable conditions and motivating a large literature on dynamics and rates  \citep{doi:10.1137/17M1139461,leonardos2025globalconvergencemultiagentpolicy,doi:10.1137/1.9781611976465.84,overman2025oversightgamelearningcooperatively,chen2022almost,yang2019sample,niu2025findingmultiplefollowerstackelberg, giannou2021rate, 10.1007/s00182-023-00837-4}. We leverage a Markov potential structure induced by
the additive reward decomposition in our cooperative Markov game to establish finite-time convergence to approximate Nash
equilibria under probabilistic approximate best responses. \looseness=-1

\section{Preliminaries}
\label{sec: preliminaries}

For $z\!\in\!\mathbb R^d$, $\|z\|_p$ is the $\ell_p$ norm of $z$. For $s_1,\dots, s_n$ and $\Delta\subseteq [n]$, let $s_\Delta\!$ be the multiset $\!\{s_i:i\!\in\!\Delta\}$. For a discrete measurable space $(\mathcal{S}, \mathcal{F})$, the TV-distance between measures $\mu$ and $\mu'$ is $\mathrm{TV}(\mu, \mu') := \frac{1}{2}\sum_{s\in\mathcal{S}}|\mu(s)\!-\!\mu'(s)|$. We use $x\sim P(\cdot)$ to denote that $x$ is a random sample from distribution $P$, and we let $\mathcal{U}(\Omega)$ denote the uniform distribution over a set $\Omega$. 
For {$k,n\in\mathbb N$, we let ${[n]\choose k} := \{S \subset [n] : |S| = k\}$} where $[n]\!=\!\{1,\dots,n\}$. For a finite set $\mathcal{X}$, $\Delta^{\cX}$ denotes the space of discrete distributions over $\cX$. Finally, $\tilde{O}(\cdot)$ hides polylogarithmic factors in any problem parameters except $n$.
\looseness=-1

Following \cite{anand2025meanfield}, we study a discounted infinite-horizon cooperative multi-agent decision problem with one \emph{global} agent $g$
and $n$ \emph{local} agents indexed by $[n]$. At time $t$, the joint state is
$s(t) = (s_g(t), s_1(t),\ldots,s_n(t)) \in \mathcal S := \mathcal S_g \times \mathcal S_l^n$,
where $s_g(t)$ is the global state and $s_i(t)$ is the state of local agent $i$. The global agent selects action $a_g(t)\in\cA_g$ while each $i$-th local agent selects $a_i(t)\in\cA_l$, and we denote the joint action as $a(t)=(a_g(t),a_1(t),\ldots,a_n(t)) \in \mathcal A := \mathcal A_g \times \mathcal A_l^n$. Next, the system collects an immediate stage reward\looseness=-1
\begin{equation}\label{equation:rewards}
    r(s, a) := \underbrace{r_g(s_g, a_g)}_{{\text{global component}}} + 1/n \cdot \sum_{i\in [n]} \underbrace{r_l(s_i, s_g, a_i)}_{\text{local component}}
\end{equation}
and then transitions to $s(t+1)$, where $s_g(t+1)\!\sim\!P_g(\cdot|s_g(t), a_g(t))$ and $s_i(t+1)\!\sim\!P_l(\cdot|s_i(t), s_g(t), a_i(t))$ for all $i\in[n]$. In this formulation, note that the local agents are \emph{homogeneous} in the sense that the system is invariant to permutations of the local agents.
 
\textbf{Information constraints and policy classes.} 
We focus on a communication-constrained regime in which the global agent cannot condition on all local agent states.
At each decision step it may only communicate with a subset of $k$ local agents (e.g., due to bandwidth or sensing limits). Likewise, local agents are unable to communicate (directly) with one another and may only communicate with the global agent. Correspondingly, we constrain the global agent's policy $\pi_g: \cS_g \times \cS_l^k$ to depend on the global agent's state and the observed states of at most $k \leq n$ other agents (we are interested in the case of $k \ll n$). Since the local agents are homogeneous, we may assume (without loss) that all local agents share a common policy $\pi_l: \cS_l \times \cS_g\to\cA_l$ that depends on the local agent's own state and the global agent's state. In the remainder of the paper, we use the notation \smash{$\Pi_g\coloneqq \{\pi_g: \mathcal{S}_g \times \mathcal{S}_l^k \to \Delta^{\mathcal{A}_g}\}$} and \smash{$\Pi_\ell \coloneqq \{\pi_l: \mathcal{S}_l \times \mathcal{S}_g \to \Delta^{\mathcal{A}_l}\}$} to denote the set of feasible policies for the global and local agents, under the communication-constrained setting. 
For a fixed discount factor $\gamma\in(0,1)$, we denote the value of the joint policy $\pi = (\pi_g, \pi_l)$ by 
\begin{equation}
V^\pi(s)=\mathbb{E}_{a(t) \sim \pi(\cdot | s)}  \Big[\sum_{t \geq 0} \gamma^t r(s(t), a(t)) \Big\lvert s(0) = s\Big].
\end{equation}
\textbf{Challenges in learning the optimal policy.} The cardinality of the discrete state/action space for the optimal policy is $|\mathcal{S}_g||\mathcal{S}_l|^n|\mathcal{A}_g||\mathcal{A}_l|^n$, which is exponential in the number of agents. Since the local agents are homogeneous, and therefore permutation-invariant with respect to the rewards of the system (the order of the other agents does not matter to any single decision-making agent), techniques from mean-field MARL restrict the cardinality of the search space for the optimal policy to $|\mathcal{S}_g||\mathcal{A}_g||\mathcal{S}_l||\mathcal{A}_l|n^{|\mathcal{S}_l||\mathcal{A}_l|}$. That is, mean-field MARL reduces the exponential dependence on $n$ to a polynomial dependence on $n$ \cite{yang2020meanfieldmultiagentreinforcement,gu2021meanfieldcontrolsqlearningcooperative,anand2025meanfield}. However, in practical systems, when $n$ is large, the $\mathrm{poly}(n)$ sample-complexity may still be computationally infeasible, as noted in \cite{anand2025meanfield}. \looseness=-1

\textbf{Solution concept under communication constraints.}
In a centralized setting without communication constraints, one could directly optimize the joint policy to maximize the discounted return $V^{\pi}(s)$. This setting was previously studied in \cite{anand2025meanfield}. 
However, under the communication constraints encoded by the policy classes $(\Pi_g,\Pi_\ell)$, agents choose policies from different restricted classes and the natural target of our alternating procedure is a fixed point of unilateral improvements.
Accordingly, we study approximate \emph{Nash equilibria} of the restricted Markov game induced by $(\Pi_g,\Pi_\ell)$, where all the local agents must learn the same policy as they are homogeneous, which corresponds to a local optimum of the objective function: intuitively, a policy pair is an $\epsilon$-Nash equilibrium if neither policy can improve the value by more than $\epsilon$ via a unilateral deviation within its respective class. We show that best-response dynamics in \texttt{ALTERNATING-MARL} converges to an $\epsilon$-Nash Equilibrium.\looseness=-1

We introduce the following standard assumptions and definitions:
\begin{assumption}[Finite state/action spaces]\label{assumption: finite cardinality}
    \emph{We assume that the state and action spaces of all the agents in the MARL game are finite:  $|\mathcal{S}_l|, |\mathcal{S}_g|, |\mathcal{A}_g|, |\mathcal{A}_l| < \infty$.  }
\end{assumption}
\begin{assumption}[Positive Bounded rewards]\label{assumption: bounded rewards}
    \emph{Each component of the reward function is positive and bounded. Specifically, $\|r_g(\cdot,\cdot)\|_\infty \leq \tilde{r}_g$, and
    $\|r_l(\cdot,\cdot,\cdot)\|_\infty \leq \tilde{r}_l$. Together, $\|r(\cdot,\cdot)\|_\infty \leq \tilde{r}_g + \tilde{r}_l := \tilde{r}$.}
\end{assumption}

\begin{definition}[Nash Equilibrium] \label{def: a nash equilibrium} \emph{A Nash Equilibrium policy tuple $\pi^\star = (\pi_g^\star,\pi_1^\star,\dots,\pi_n^\star)$ is a solution to a state of the game where all the players perform their best-response strategies against the other players so that no player unilaterally deviates from its current policy \cite{doi:10.1073/pnas.36.1.48}. Specifically, the policy tuple $(\pi_g^\star,\pi_1^\star,\dots,\pi_n^\star)$ satisfies that for all $\pi_g, \pi_1, \dots, \pi_n \in \Delta^{(g)}\times\Delta^{(1)}\times\dots\times\Delta^{(n)}$, we have that $V^{\pi^\star}\geq V^{\pi_g^\star,\pi_1^\star,\dots,\pi_{i-1}^\star,\pi_i,\pi_{i+1}^\star,\dots,\pi_n^\star}$ for all $i\in[n]$ and $V^{\pi^\star}\geq V^{\pi_g,\pi_1^\star,\dots,\pi_n^\star}$.}
\end{definition}

\begin{definition}[$\epsilon$-approximate Nash Equilibrium policy]\label{def:approx-nash}\emph{ A policy tuple $(\hat{\pi}_g, \hat{\pi}_{[n]})$ is said to be an $\epsilon$-approximate Nash Equilibrium if it corresponds to a solution to a state of the system  where all the players perform their best-responses against the other players so that no player unilaterally deviates from its current policy. Formally, it satisfies, that for all $(\pi_g, \pi_{[n]})$, that  $V^{\hat{\pi}_g, \hat{\pi}_{[n]}} \geq V^{\pi_g, \hat{\pi}_{[n]}} - \epsilon$, and moreover, for all $i\in[n]$, $V^{\hat{\pi}_g, \hat{\pi}_{[n]}} \geq V^{\hat{\pi}_g, \hat{\pi}_1, ..., \hat{\pi}_{i-1}, \pi_i, \hat{\pi}_{i+1}, ..., \hat{\pi}_n} - \epsilon$.}
\end{definition} 

\textbf{Motivating examples.} We give examples of two cooperative systems which are naturally modeled by our setting. We provide experimental results and code in \cref{sec: numerical simulations}.
\begin{figure}[t]
    \centering
\includegraphics[width=0.49\linewidth]{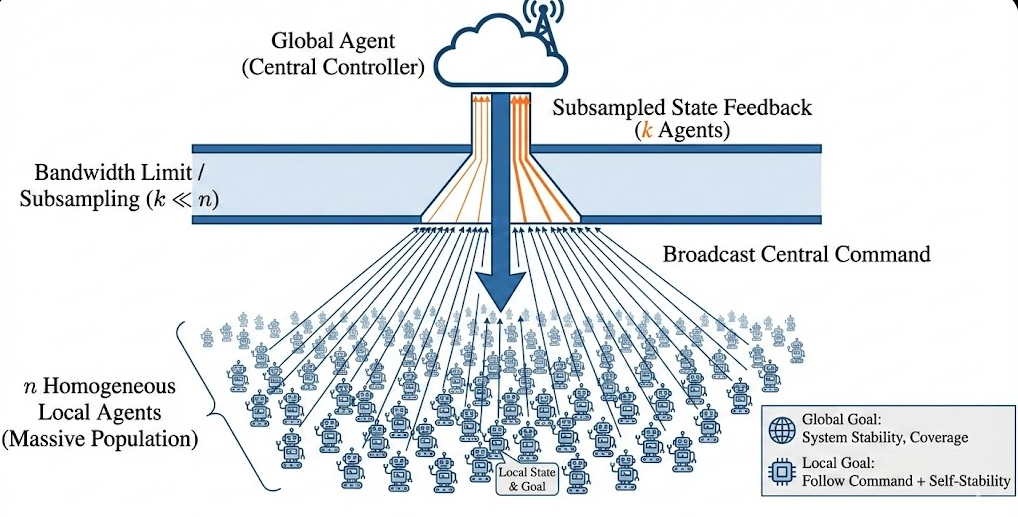} 
        \includegraphics[width=0.49\linewidth]{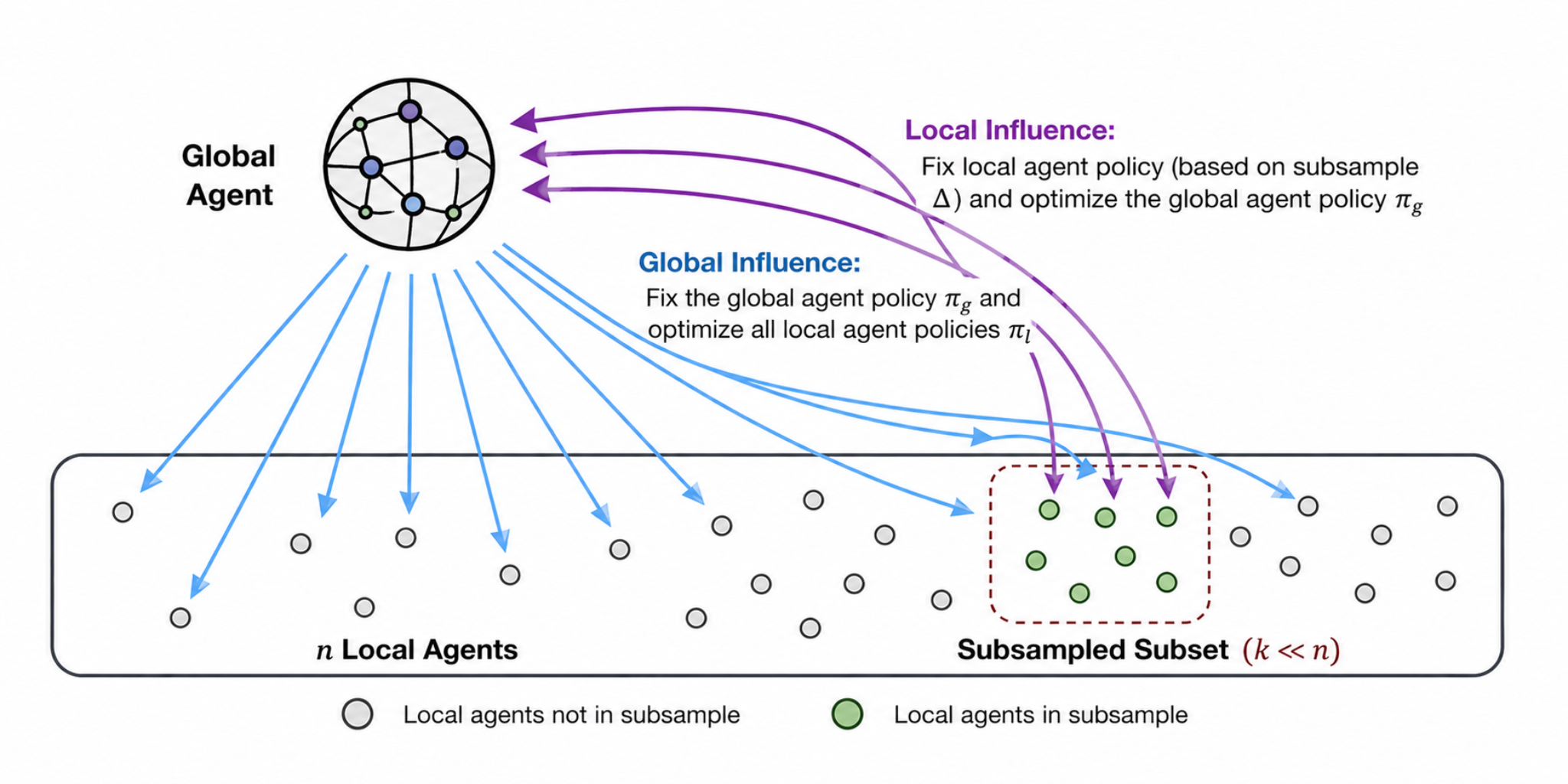} 
        \caption{(Left) Robot Coordination. This figure illustrates how our communication-constrained framework can be applied for decentralized coordination of large teams of robots. (Right) Subsampling with alternating best-response dynamics. AI was used to refine the aesthetics of the figures.}
    \label{fig:placeholder}
\end{figure}
 \begin{example}[Communication-constrained control] 
\emph{Consider a networked system (smart grid,  multi-robot team, etc.) with a  global agent who can only communicate with $k$ local agents  at each timestep due to bandwidth limitation. Each local agent (robot, generator, etc.) can sense its state and the global agent's state. The global agent's task is  to maintain system stability or optimize a global performance metric (voltage stability in a grid, coverage in a surveillance task, etc.), and the local agent’s objective is to follow the central command but also a local goal (or maintaining its own stability).}\looseness=-1
\end{example}
 \begin{example}[Federated optimization with partial participation] \emph{A central server (global agent) coordinates $n$ clients (local agents) to minimize a shared objective (e.g., validation loss). Due to bandwidth constraints, at each round the server can only query a subset of $k\ll n$ clients for their local states (e.g., gradient statistics) and then broadcasts a global update, where each client performs a local action (e.g., local SGD steps or a regularization choice) based on its own state and the broadcast.} \looseness=-1
\end{example}

\section{Algorithmic Approach: Alternating MARL}\label{sec:approach-algorithms}
Our approach builds on that of \cite{anand2025meanfield}. However, to extend \cite{anand2025meanfield} to the communication-constrained setting, we require several technical adaptations to generalize their methods to a Markov game framework. While some of our approach follows the convention of \cite{anand2025meanfield}, the analytic techniques and algorithms differ substantially: while Algorithm~\ref{algorithm: g learn} is perhaps similar to \cite{anand2025meanfield}, our work introduces technical novelties in the analysis of Algorithms~\ref{algorithm: l learn} and~\ref{algorithm: alternating marl}, and in its  theoretical guarantees.
\subsection{Overview of Approach}
Fix the local agent's policy $\pi_\ell \in \Pi_\ell$ and optimize the global agent's policy $\pi_g \in \Pi_g$. Then, letting $a_g \sim \pi_g(\cdot|s_g,s_\Delta)$, the transition dynamics are constrained to $s_g' \!\sim\! P_g(\cdot|s_g, a_g)$ and $s_i'\!\sim\!\bar{P}_l^{\pi_\ell}(\cdot|s_i,s_g)$, where $\bar{P}_l^{\pi_\ell}(\cdot|s_i, s_g)= P_l(\cdot| s_i, s_g, \pi_\ell( s_i, s_g))$. Optimizing $\pi_g$ in these constrained dynamics corresponds to approximately computing a \emph{best-response} to the fixed policy $\pi_\ell$. We then repeat this procedure, holding $\pi_g$ fixed, and optimizing $\pi_\ell$ in the induced MDP. This alternating procedure forms a sequence of \emph{best-response dynamics}, whose convergence implies a Nash equilibrium.\looseness=-1

\textbf{Offline Learning:} 
First, in \cref{algorithm: g learn} (\texttt{G-LEARN}), the global agent fixes the local agent's policy $\pi_\ell$. Next, following the techniques of \cite{anand2025meanfield}, the algorithm randomly samples a subset of local agents $\Delta\subseteq [n]$ such that $|\Delta|=k$: it runs value iteration (with $m$ samples) to approximate the $Q$-function \smash{$\hat{Q}_{k,m}^{\text{est}}$} and policy \smash{$\hat{\pi}_{k,m}^{\text{est}}$} for its best-response to $\pi_\ell$ in the surrogate subsystem of $k$ local agents. The surrogate reward of this system, under the fixed local policy $\pi_\ell$, is \smash{$\bar{r}^{\pi_\ell}_l(s_i, s_g)=r_l(s_i, s_g, \pi_\ell(s_i, s_g))$}, and the stage reward received by the global agent is \smash{$\bar{r}^{\pi_\ell}_\Delta:\cS\times\cA_g\to\mathbb{R}$}, where
\begin{align}
\bar{r}^{\pi_\ell}_\Delta(s,a_g) = r_g(s_g,a_g)+1/{|\Delta|} \cdot \sum_{i\in\Delta}\bar{r}^{\pi_\ell}_l( s_i, s_g).
\end{align}
 \cref{theorem: Q-lipschitz of Fsdelta and Fsn} proves an analog of Theorem E.3 of \cite{anand2025meanfield}, which shows that $\|\hat{Q}_{k,m}^{\text{est}} - Q^*\|_\infty$ is Lipschitz continuous with respect to the TV-distance between the states of the subsampled agents and the full set of agents, thereby decaying on the order of $\tilde{O}(1/\sqrt{k})$.

Next, in \cref{algorithm: l learn} (\texttt{L-LEARN}), a representative local agent fixes the global agent's policy $\pi_g$, and learns an approximate best-response in the policy class $\Pi_\ell$ consisting of policies of the form $\smash{\pi_\ell:\cS_g\times\cS_l\to\cA_l}$. However, since the global agent's action depends on a sample of $k$ local agents, the representative local agent's environment is not Markovian in $(s_g, s_i)$. We resolve this by constructing an episodic chained-MDP in \cref{algorithm: k chained induced MDP} and \cref{algorithm: S chained induced MDP} and then applying a standard PAC episodic RL method (UCFH, \cite{dann2015sample}) to extract the local agent's policy.
Finally, in \cref{algorithm: alternating marl}, \texttt{ALTERNATING-MARL} deploys algorithms \texttt{G-LEARN} and \texttt{L-LEARN} to alternately learn a better policy through the induced best-response dynamics process until it converges.

\textbf{Online Execution:} To convert the optimality of the global agent's action in the $k$ local-agent subsystem to an approximate optimality guarantee in the full system, \cref{algorithm: online execution} adapts the techniques of \cite{anand2025meanfield} to form a randomized policy $\pi_{k,m}^{\text{est}}$. In it, the global agent subsamples $\smash{\Delta\in{[n]\choose k}}$ at each time-step to derive an action $\smash{a_g\gets \hat{\pi}_{k,m}^{\text{est}}(s_g, s_\Delta)}$, while each local agent $i\in[n]$ evaluates the shared $\pi_\ell$ to get action $a_i \sim \pi_\ell(\cdot|s_i,s_g)$. \cref{theorem: main result} shows that the resultant policy is a $\tilde{O}(\frac{1}{\sqrt{k}})$-approximate NE.   
\subsection{Algorithm Description}
We first introduce the empirical distribution function.

\begin{definition}[Empirical distribution function $F_{s_\Delta}$] \emph{Let $\mu_k(\cS_l) = \{\nu\in\Delta(\cS_l): \nu(x)\in\{0,\frac1k, \frac2k,\dots, 1\}$. Then, for $\Delta\in {[n]\choose k}$, the empirical distribution function $F_{s_\Delta} \in \mu_k(\cS_l)$ is the proportion of sampled agents at each state, i.e., for $s\in \cS_l, F_{s_\Delta}(s) = \frac1k \sum_{i\in\Delta}\mathbbm{1}\{s_i=s\}$.}\looseness=-1
\end{definition}

\textbf{Offline Learning (I).} Let $m\in\mathbb N$ denote the sample size in \cref{algorithm: g learn} (\texttt{G-LEARN}) with parameter $k\leq n$. When $|\cS_l|^k \leq |\cS_l|k^{|\cS_l|}$, the algorithm uses traditional value-iteration, and otherwise it uses mean-field value iteration. We formally describe the procedures for each regime below.

\begin{algorithm}[t]
    \caption{\texttt{G-LEARN:} Global-agent $Q$-learning}
    \begin{algorithmic}
        \REQUIRE Global agent transition function $P_g$, local agent policy-guided transition function $\bar{P_l}^{\pi_\ell}$, subsampling parameters $k$ and $m$, and number of Bellman updates $T$.\\
        \STATE Let $\mu_{k}(\cS_l) \coloneqq \{v=\{0, \frac1k, \dots, 1\}^{|\cS_l|}: \sum_i v_i = 1\}$.
        \STATE \textbf{if} {$|\cS_l|^k \leq |\cS_l|k^{|\cS_l|}$} \quad \textcolor{blue}{\texttt{/$\star$ subsampled standard $Q$-learning}}
        \STATE \quad Set $\hat{Q}_{k,m}^0(s_g, s_\Delta,a_g)=0, \forall (s_g, s_\Delta, a_g)\in \mathcal{S}_g\times \mathcal{S}_l^k\times  \mathcal{A}_g$.\looseness=-1
        \STATE \quad \textbf{for} {$t=0,\dots,T-1$} \textbf{do}
        \STATE \quad\quad \textbf{for} {$s_g,s_\Delta,a_g\in\mathcal{S}_g\times\mathcal{S}_l^k\times\mathcal{A}_g$} \textbf{do}
        \STATE \quad\quad\quad $\hat{Q}_{k,m}^{t+1}(s_g,s_\Delta,a_g) = \tilde{T}_{k,m} \hat{Q}_{k,m}^t(s_g,s_\Delta,a_g)$
        \STATE \quad\textbf{else} \quad\quad\quad\quad\quad\quad \textcolor{blue}{\texttt{/$\star$ subsampled mean-field $Q$-learning}}
        \STATE \quad Set $\hat{Q}_{k,m}^0(s_g,F_{ s_\Delta},a_g)=0$
        \STATE \quad \textbf{for} {$t=0,\dots,T-1$} \textbf{do}
        \STATE \quad\quad \textbf{for} {$s_g,F_{s_\Delta},a_g\in\mathcal{S}_g\times\mu_k(\cS_l)\times\mathcal{A}_g$}
        \STATE \quad\quad\quad $\hat{Q}_{k,m}^{t+1}(s_g, F_{s_\Delta},a_g) = \hat{T}_{k,m} \hat{Q}_{k,m}^t(s_g, F_{s_\Delta},a_g)$
        \STATE Let
 $\hat{\pi}_{k,m}^{T}(\cdot,\cdot)=\arg\max_{a_g\in\cA_g} \hat{Q}_{k,m}^{T}(\cdot,\cdot,a_g)$.
 \STATE Let $V_{\text{new}} = \texttt{APPROX-VALUE}(\hat{\pi}_{k,m}^{T}, \hat{Q}_{k,m}^{T})$.
     \STATE \textbf{return} $\hat{\pi}_{k,m}^{T}$ and $V_{\text{new}}$.
\smallskip
\STATE \textbf{Function} \texttt{APPROX-VALUE}$(\pi_g, \hat{Q}_{k,m})$
\STATE \textbf{if} { $|\cS_l|^k \leq |\cS_l| k^{|\cS_l|}$} 
        \STATE \quad \textbf{for} $s \in \mathcal{S}_g\times \mathcal{S}_l^k$ \textbf{do}
         \STATE \quad\quad  $V_k^{\pi_g}(s) = \sum_{a_g \in \cA_g} \pi_g(a_g|s) \hat{Q}_{k,m}(s, a_g)$
    \STATE \textbf{else}
   \STATE \quad \textbf{for} $s\in \mathcal{S}_g\times\mu_k(\mathcal{S}_l)$ \textbf{do}
   \STATE \quad\quad $V_k^{\pi_g}(s) = \sum_{a_g \in \cA_g} \pi_g(a_g|s) \hat{Q}_{k,m}(s, a_g)$
\STATE \textbf{Return $V_k^{\pi_g}$}.
         \label{algorithm: g learn}
    \end{algorithmic}
\end{algorithm}
When $\smash{|\cS_l|^k\leq|\cS_l| k^{|\cS_l|}}$, we fix the policy $\pi_\ell$ and learn the best-response $Q$-function for a subsystem with $k$ local agents, denoted by $\hat{Q}_{k,m}^t:\cS_g\times\cS_l^k\times\cA_g\to\mathbb{R}$, which we initialize to $0$. At time $t$, we update $\hat{Q}_{k,m}^{t+1}(s_g, s_\Delta, a_g) = \tilde{T}_{k,m}\hat{Q}_{k,m}^t(s_g,s_\Delta,a_g)$, where $\tilde{T}_{k,m}$ is the empirical adapted Bellman operator, which uses $m$ random samples $s_g^j\sim P_g(\cdot|s_g,a_g)$ and $s_i^j\sim \bar{P_l}^{\pi_\ell}(\cdot|s_i,s_g)$ for each $j\in[m]$.
\begin{align*}\tilde{T}_{k,m}\hat{Q}_{k,m}^t(s_g,s_\Delta,a_g) &= \bar{r}^{\pi_\ell}_\Delta(s,a_g) + {\gamma}/{m} \cdot \sum_{j\in[m]}\max_{a_g'\in\cA_g}\hat{Q}_{k,m}^t(s_g^j, s_\Delta^j,a_g').
\end{align*}
\texttt{G-LEARN} applies value iteration with $\tilde{\cT}$ until $\hat{Q}_{k,m}$ converges to a fixed point satisfying $\smash{\tilde{T}_{k,m}\hat{Q}_{k,m}^{\text{est}} = \hat{Q}_{k,m}^{\text{est}}}$ ($\smash{\tilde{\cT}_{k,m}}$ is $\gamma$-contractive), converging to a deterministic policy $\hat{\pi}_{k,m}^{\text{est}}$,
\begin{equation}\hat{\pi}_{k,m}^{\text{est}}(s_g,s_\Delta) = \arg\max_{a_g\in\mathcal{A}_g} \hat{Q}_{k,m}^{\text{est}}(s_g, s_\Delta, a_g).\end{equation}
When $\smash{|\cS_l|^k>|\cS_l|k^{|\cS_l|}}$, the mean-field MARL transformation is helpful in reducing complexity, and we learn the mean-field best-response $Q$-function, $\hat{Q}_{k,m}^t:\cS_g\times\mu_{k}(\cS_l)\times\cA_g\to\mathbb{R}$. Initializing at $0$, we iteratively perform the update in the empirical distribution parameterization $\hat{Q}_{k,m}^{t+1}(s_g, F_{s_\Delta}, a_g) = \hat{T}_{k,m}\hat{Q}_{k,m}^t(s_g, F_{s_\Delta}, a_g)$,
where $\hat{T}_{k,m}$ is the empirical adapted mean-field Bellman operator: 
\begin{align*}
    \hat{\cT}_{k,m}\hat{Q}_{k,m}^t & (s_g, F_{s_\Delta},a_g) = \bar{r}^{\pi_\ell}_\Delta(s,a_g) + {\gamma}/{m} \cdot  \sum_{j\in[m]}\max_{a_g'\in\cA_g}\hat{Q}_{k,m}^t(s_g^j, F_{s_\Delta^j},a_g'),
\end{align*}
 which admits the same fixed point $\hat{Q}_{k,m}^{\text{est}}$ and policy $\hat{\pi}_{k,m}^{\text{est}}$.

\textbf{Offline Learning (II).} For \cref{algorithm: l learn} (\texttt{L-LEARN}), we fix the global agent's policy $\pi_g$. Then, a representative local agent faces an induced MDP with state space
$\mathcal{S}_l \times \mathcal{S}_g$, action space $\mathcal{A}_l$, transition kernel inherited from the game,
and reward of the original game.  Since the global agent's action depends on a sample of $k$ local agents, the representative local agent's environment is not Markovian in $(s_g, s_i)$. We resolve this by constructing an episodic chained-MDP in Algorithms \ref{algorithm: k chained induced MDP} and \ref{algorithm: S chained induced MDP}. However, the global agent's policy $\pi_g$ requires a $k$-tuple of local states whereas a single local agent only sees $s_{g}$ and $s_{i}$. Our reduction resolves this by delicately maintaining a $k$-tuple of local states and serializing the $k$ local decisions within each macro step where each agent uses the same policy which maps from $\mathcal{S}_l\times\mathcal{S}_g\to\mathcal{A}_l$. \looseness=-1

\begin{algorithm}[t]
\caption{\texttt{L-LEARN}: Local-agent $Q$-learning}
    \begin{algorithmic}
        \REQUIRE Transition kernels $P_l, P_g$, fixed global-agent policy $\pi_g$, parameters $k$ and $\gamma$,  and failure probability $\delta_\ell$. Set $H\coloneqq \frac{1}{1-\gamma}\log \frac{\|r\|_\infty\sqrt{k}}{1-\gamma}$ as the effective horizon.
        \STATE \textcolor{blue}{\texttt{/$\star$ Form a proxy episodic MDP ${\widetilde{M}_k(\pi_g)}$}}
        \STATE \textbf{if} { $|\cS_l|^k \leq |\cS_l| k^{|\cS_l|}$} \quad \textcolor{blue}{\texttt{/$\star$ standard parameterization}}
        \STATE \quad  \cref{algorithm: k chained induced MDP}  constructs a $k$-chained MDP $\widetilde M_k$  with horizon $\widetilde H = H\cdot k$ 
        \STATE \textbf{else}  \qquad\qquad\qquad\quad \textcolor{blue}{\texttt{/$\star$ mean-field parameterization}}
        \STATE \quad \cref{algorithm: S chained induced MDP} constructs a $|\mathcal S_l|$-chained mean-field MDP $\widetilde M_{k}$ with horizon $\widetilde H = H(|\mathcal S_l|+1)$.
        \smallskip
        \STATE Run an episodic PAC RL solver  on $\widetilde{{M}}_k(\pi_g)$ with parameters $\epsilon_\ell=\frac{1}{n\sqrt{k}}, \delta_\ell$ to get induced policy $\pi_\ell$.
        \STATE \textbf{Return} policy $\pi_\ell$ and value function estimate $\widehat{V}^{\pi_\ell, \pi_g}$.
    \end{algorithmic}
    \label{algorithm: l learn}
\end{algorithm}

\cref{algorithm: k chained induced MDP} replaces each step $t$ by a chain of $k$ micro steps $\tau = tk,  \dots, tk+(k-1)$, indexed by a pointer $j\in[k]$. The state of the unfolded system at micro step $\tau$ is $\smash{(j, s_g, \bar{s}_g, s_{1:k})}$, where $s_{1:k}$ are the local ``replica'' states, $s_g$ is the current global state, and $\bar{s}_g$ is a pending next global state. At the first micro node of each chain at stage $(1, s_g, \perp, s_{1:k})$, we let $\smash{a_g\sim \pi_g(\cdot|s_g, s_{1:k})}$, and sample a pending global next state $\smash{\bar{s}_g\sim P_g(\cdot|s_g, a_g)}$. This transition is computed before updating any local replica, making it consistent with the original time-step. Along the chain of micro nodes, only one local replica is ``active'' at any time. At stage $j$, the learner selects an action $a_\ell$ for replica $j$, and the environment updates the replica via ${s_j' \sim P_l(\cdot|s_j, s_g, a_\ell)}$, leaving the other local replica states unchanged, and increments $j$. At stage $j=k$, it sets $s_g\!\gets\! \bar{s}_g$ and resets $\bar{S}_g \gets \perp$, completing a step. \footnote{We can recover a corresponding near-optimal policy in the original MDP from one in the chained MDP (as in Fig.~\ref{fig:k-chained-mdp}).}
Similarly, in the mean-field regime, \cref{algorithm: S chained induced MDP} maintains a distribution on the number of agents at each state: if $h(u)$ agents occupy state $u\in\mathcal{S}_l$, we apply the local action $a_\ell$ for state $u$ to move the mass $h(u)$, to next states via the multinomial induced by $P_l(\cdot|u,s_g,a_\ell)$, in a new histogram $\bar{F}_\Delta$. 

Finally, we reduce $\widetilde{M}_k(\pi_g)$ to an episodic MDP with state-only rewards in \cref{subsection: reduction episodic}, and run an episodic PAC RL solver (UCFH \citep{dann2015sample} in \cref{algorithm from brunskill and dann}) to get an $\epsilon_\ell$ best-response induced policy $\pi_\ell$. 

\textbf{Offline Learning (III).} \texttt{ALTERNATING-MARL} in \cref{algorithm: alternating marl} orchestrates alternating best-responses of \texttt{G-LEARN} and \texttt{L-LEARN} for $N_{\text{steps}}$-many iterations: each algorithm proposes its best-response policy, and \cref{algorithm: alternating marl} accepts the policy if it estimates that it has improved the joint value function, up to a tolerance radius of $\smash{\tilde{O}(1/\sqrt{k})}$. To ensure that the potential drifts upward on every accepted update, we introduce function \texttt{UPDATE} in \cref{algorithm: alternating marl}, which rejects worse joint policies and terminates the algorithm earlier if a learned best-response policy is within the tolerance radius.\looseness=-1

\begin{algorithm}[t]
\begin{algorithmic}
\caption{\texttt{ALTERNATING-MARL}} 
\REQUIRE Transition functions $P_g$ and $P_\ell$, initial uniform policies $(\pi_g^0,\pi_\ell^0)$, steps $N_{\text{steps}}$,   failure probability $\delta$, algorithms \texttt{G-LEARN} and \texttt{L-LEARN}, and sampling parameters $k$ and $m$.
\STATE Set $\delta_{\text{learn}} = \frac{\delta}{4N_{\text{steps}}}, T=\frac{2}{1-\gamma}\log\frac{\tilde{r}\sqrt{k}}{(1-\gamma)^2}$, and  $\hat{V}_{\text{old}} = 0$.\\
\FOR{$t=1,\dots,N_{\text{steps}}$}
\STATE $ \tilde\pi_g, \hat{V}_{\text{new}} \leftarrow \texttt{G-LEARN}
(P_g,\bar{P}_l^{\pi_\ell^{t-1}},k,m,T)$ 
\STATE $\smash{\pi_g^t, \textsc{Nash}, \hat{V}_{\text{old}}} \gets 
\texttt{UPDATE}(\tilde{\pi}_g, \pi_g^{t-1}, \hat{V}_{\text{new}}, \hat{V}_{\text{old}})$ \quad \textcolor{blue}{\texttt{/$\star$ Propose a policy update to $\pi_g$}}
\STATE \textbf{if} \textsc{Nash} = \textsc{True} \textbf{then} \textbf{return} 
$(\pi_g^t, \pi_\ell^{t-1})$.\smallskip
  \STATE $\tilde\pi_\ell, \hat{V}_{\text{new}} \leftarrow \texttt{L-LEARN}(P_g,P_\ell,\pi_g^t,k, \delta_{\text{learn}})$.
  \STATE $\pi_\ell^t, \textsc{Nash}, \hat{V}_{\text{old}} \gets \texttt{UPDATE}(\tilde{\pi}_\ell, \pi_\ell^{t-1},  \hat{V}_{\text{new}}, \hat{V}_{\text{old}})$  \quad \textcolor{blue}{\texttt{/$\star$ Propose a policy update to $\pi_\ell$}}

  \STATE \textbf{if} \textsc{Nash} = \textsc{True} \textbf{ then }  \textbf{return} $(\pi^t_g, \pi^t_\ell)$.
  \ENDFOR
\STATE \textbf{return} $(\pi_g^{N_{\text{steps}}},\pi_\ell^{N_{\text{steps}}})$.
\smallskip
\STATE \textbf{Function} \texttt{UPDATE}$(\pi_{\text{new}}, \pi_{\text{old}}, \hat{V}_{\text{new}}, \hat{V}_{\text{old}}):$
\STATE \quad Let ${\eta = \frac{2\tilde{r}}{(1-\gamma)^2}(\sqrt{\frac{1}{2k} \ln(2|\mathcal{S}_l||\mathcal{A}_g|\sqrt{k})} + \frac{5}{\sqrt{k}})}$.

\STATE \quad \textbf{if} {$\hat{V}_{\text{new}} < \hat{V}_{\text{old}} - 2\eta$} on some state \textbf{ then }  \textbf{return} $(\pi_{\text{old}}, \textsc{False}, \hat{V}_\text{old})$\quad  \textcolor{blue}{\texttt{/$\star$ reject}}
\STATE \quad \textbf{if} {$\hat{V}_{\text{new}} > \hat{V}_{\text{old}} + 2\eta$} on some state \textbf{then}  \textbf{return} $(\pi_{\text{new}}, \textsc{False}, \hat{V}_{\text{new}})$  \textcolor{blue}{\texttt{/$\star$ accept}}
\STATE \quad \textbf{else  }    \textbf{return} $(\pi_{\text{old}}, \textsc{True}, \hat{V}_\text{old})$  \textcolor{blue}{\texttt{/$\star$ terminate}}
\label{algorithm: alternating marl}
\end{algorithmic}
\end{algorithm} 
    
\textbf{Online Execution:} \cref{algorithm: online execution}  (Online execution). To convert the optimality of the global agent's action in the $k$ local-agent subsystem to an approximate optimality guarantee on the full $n$-agent subsystem, we propose a randomized policy $\smash{\pi_{k,m}^{\text{est}}}$ where the global agent samples $\smash{\Delta\in{[n]\choose k}}$ at each time-step to derive an action $\smash{a_g\gets \hat{\pi}_{k,m}^{\text{est}}(s_g, s_\Delta)}$ and each local agent $\smash{i\in[n]}$ evaluates the shared $\smash{\pi_\ell}$ to get an action ${a_i \sim \pi_\ell(\cdot|s_i,s_g)}$. \cref{theorem: main result} shows that the resultant policy is a $\tilde{O}(1/\sqrt{k})$-approximate NE.\looseness=-1

\begin{algorithm}
    \caption{Online execution stage}
    \begin{algorithmic}
        \REQUIRE Policies $\pi_g$ and $\pi_\ell$, and states $s_g(0), s_{[n]}(0)\!\sim\!s_0$.
        \STATE Let $R \gets 0$
        \FOR{$t=0,\dots,T-1$}
        \STATE Sample $\Delta\subseteq {[n]\choose k}$ uniformly at random and let $a_g(t) \sim \hat{\pi}_{g,k}(\cdot|s_g(t), s_\Delta(t))$
        \STATE For all $i\in[n]$, let $a_i(t) \sim \pi_\ell(\cdot|s_i(t),s_g(t))$
        \STATE Collect a stage reward $r_t= r_g(s_g,a_g)+\frac1n \sum_{i=1}^n r_l(s_i,s_g,a_i)$
        \ENDFOR
        \label{algorithm: online execution}
    \end{algorithmic}
\end{algorithm}
\section{Theoretical Guarantees and Analysis}
 We begin by giving theoretical guarantees for \texttt{G-LEARN.} For this, we introduce the Bellman noise.

\begin{definition}[Bellman noise] \emph{Note that $\hat{\cT}_{k,m}$ is an unbiased estimator of the adapted Bellman operator $\hat{\cT}_k$, where $\hat{\mathcal T}_k \hat{Q}_k(s_g, s_\Delta, a_g)
\coloneqq
\bar{r}^{\pi_\ell}_\Delta(s,a)
+\gamma\mathbb E_{(s_g',s_\Delta')\sim \mathcal P_k} \max_{a_g'\in\cA_g} \hat{Q}_k(s_g',s_\Delta',a_g')]$,
where $\mathcal P_k\!\coloneqq\!P_g\!\otimes\!(\bar{P}^{\pi_\ell}_l)^{\otimes k}$. Here, $\hat Q_k$ is initialized to $0$ and at time $t$ we update $\hat{Q}_k^{t+1}\!=\!\hat{\cT}_k \hat{Q}_k^t$. Similarly, $\hat{\cT}_k$ is a $\gamma$-contraction and admits a fixed point $\hat{Q}_k^*$. By the law of large numbers, as $m\!\to\!\infty$, $\hat{\cT}_{k,m}\!\to\!\hat{\cT}_k$ and $\hat{Q}^*_{k,m}\!\to\!\hat{Q}^*_k$. For finite $m$, $\epsilon_{k,m}\coloneqq \|\hat{Q}_{k,m}^{\text{est}} - \hat{Q}_k^*\|_\infty$ is the Bellman noise.}\looseness=-1
\end{definition}
 \cref{theorem: g-learn preliminary} provides a bound on the approximate best-response of the learned policy.

\begin{theorem}\label{theorem: g-learn preliminary} For all $s\in\cS_g\times\cS_l^n$, if $T\geq\frac{2}{1-\gamma}\log\frac{\tilde{r}\sqrt{k}}{(1-\gamma)}$ in \emph{\texttt{G-LEARN}}, then
\begin{align*}V^{\pi^*}(s)&- V^{{{\pi}}^\est_{k,m}}(s) \leq \frac{2\epsilon_{k,m}}{1-\gamma} + \frac{2\tilde{r} k^{-1/2}}{(1-\gamma)^2} + \frac{2\tilde{r}}{(1-\gamma)^2}\sqrt{\frac{1}{2k} \ln 2k |\mathcal{S}_l||\cA_g|}.
\end{align*}
\end{theorem}
We prove \cref{theorem: g-learn preliminary} in \cref{sec: optimality of global agent policy}, and generalize the result to stochastic rewards in \cref{stochastic generalization}. To bound $\epsilon_{k,m}$, we specify the number of samples $m$ for controlling the Bellman noise in \cref{lemma: controlling bellman noise}.

\begin{lemma}
    [Controlling the Bellman noise]\label{lemma: controlling bellman noise} For $k\in[n]$, let \smash{$m_1 \geq  \tilde{O}(\frac{k^2 \gamma^2 \tilde{r}^2 |\cS_l|}{(1-\gamma)^4})$} be the number of samples in the mean-field parameterization and {$m_2 \geq \tilde{O}(\frac{\tilde{r}^2 k^2 \gamma^2}{(1-\gamma)^4})$} be the number of samples in the standard parameterization. Then, with probability at least $1-\frac{1}{100e^k}$ we have $\epsilon_{k,m}\leq \frac{2}{\sqrt{k}}$.\looseness=-1
\end{lemma}
    We prove \cref{lemma: controlling bellman noise} in Appendix \ref{lemma: epsilon_km_is_k}. Combining \cref{theorem: g-learn preliminary} and \cref{lemma: controlling bellman noise},  we have

\begin{theorem}
    [Global agent subsampling result] \label{theorem: g-learn}
    Let $m^*$ be chosen as in \cref{lemma: controlling bellman noise}, and set $\pi_{k}^{\text{est}} \coloneqq \pi^{\text{est}}_{k,m^*}$. Then, for all $s\in\mathcal{S}$, we have that with probability at least $1-\frac{1}{100e^k}$: \label{theorem: application of PDL}
\begin{align*}V^{\pi^*}(s) &- V^{{{\pi}}^\text{est}_{k}}(s)  \leq {O}\left(\frac{\tilde{r}}{\sqrt{k}(1-\gamma)^2}\sqrt{ \ln(2|\mathcal{S}_l||\mathcal{A}_g|\sqrt{k})} + \frac{2}{\sqrt{k}}\right) = \tilde{O}\left(\frac{1}{\sqrt{k}}\right).
\end{align*}
\end{theorem}
Moreover, the sample complexity of \cref{algorithm: g learn} (\texttt{G-LEARN}) in the standard parameterization is $\tilde{O}(\frac{k^4\gamma^2 \tilde{r}^2 |\cS_g| |\cS_l|^k |\cA_g|}{(1-\gamma)^5})$, while the mean-field sample complexity is $\tilde{O}(\frac{k^{3+|\cS_l|}\gamma^2\|r_l\|^2_\infty |\cS_g| |\cS_l|^2 |\cA_g|}{(1-\gamma)^5})$.

\begin{remark} We extend the formulation of \texttt{G-LEARN} to off-policy $Q$-learning \cite{chen2022sample}, which replaces the generative oracle with a stochastic approximation scheme to learn a policy using historical data. \cref{Appendix/off-policy} gives theoretical guarantees with a similar decaying optimality gap as in \cref{theorem: performance_difference_lemma_applied}. \end{remark}

\subsection{Convergence Guarantees on \texttt{L-LEARN}}

\begin{theorem}[Local agent best-response optimality and sample-complexity guarantees] \label{theorem: l-learn}With probability at least $1-\delta$, the local-agent's learned policy is an $\epsilon$ best-response to the global agent's fixed policy $\pi_g$ with sample complexity
\[\tilde{O}\!\left(\!\min\!\left\{\!\frac{|\cS_l|^5|\cS_g|^2 k^{2|\cS_l|}|\cA_l|}{\epsilon^2(1-\gamma)^2} \!\ln\!\frac1\delta, \frac{k^5  |\cS_g|^2 |\cS_l|^k |\cA_l|}{\epsilon^2(1-\gamma)^3}\!\ln\!\frac1\delta\!\right\}\!\right).\]   
\end{theorem} 

We prove \cref{theorem: l-learn} in Appendix \ref{appendix theorem: local agent br optimality}, and use these results to analyze \texttt{ALTERNATING-MARL}.

 \subsection{Convergence Guarantees on {\texttt{ALTERNATING-MARL}}}

The dynamics of the meta-algorithm induces a Markov game, where the play proceeds by steps from position to position, according to transition probabilities controlled jointly by the agents.

\begin{definition}[Markov Potential games (MPG)] \emph{Denote the global agent by index $0$, the policy of agent $i\in\{0\}\cup[n]$ by $\pi_i$, and the joint policies of all the other agents by $\pi_{-i}$. Then, a stochastic game $\mathcal{M}$ is a Markov Potential Game if there is a potential function $\Phi:\mathcal{S}\times\mathcal{A}\to\mathbb{R}$ so that for any agent $i\in\{0\}\cup[n]$, policy pairs $(\pi_i, \pi_{-i}), (\pi_i', \pi_{-i})$, and state $s_0$, we have that
$\mathbb{E}_{\pi'_i,\pi_{-i}}[\sum_{t=0}^\infty \gamma^t r_i(s_t,a_t) | s_0]\!-\! \E_{\pi}[\sum_{t=0}^\infty \gamma^t r_i(s_t,a_t) | s_0] 
\!=\! \mathbb{E}_{\pi'_i,\pi_{-i}}[\sum_{t=0}^\infty \gamma^t \Phi(s_t,a_t)| s_0]\!-\! \E_{\pi}[\sum_{t=0}^\infty \gamma^t \Phi(s_t,a_t)|s_0]$,
i.e., if agent $i$ switches from policy $\pi_i$ to policy $\pi'_{i}$, $\Delta\Phi$ is the change in the agent's utility.}
\end{definition}

\cref{lemma: M is a Markov potential game} 
 shows that our setting is an MPG and uses a theorem of \cite{doi:10.1073/pnas.39.10.1095} that best-response dynamics converge to pure Nash Equilibria in MPGs. However, unlike the continuous time dynamics which converges rapidly \cite{doi:10.1137/17M1139461}, the discrete-time dynamics can converge exponentially slowly in $n$ \cite{DurandGaujal:2016:ACBRA}. However, under homogeneity, the $n$ local-agent policies are symmetric, and the $n$ Nash inequalities are isomorphic around a symmetric profile in $\Pi_\ell$, making the problem equivalent to a $2$-player game between the global agent and a representative local agent, yielding a dimension-reduced dynamic that achieves polynomial convergence towards a Nash Equilibrium in \cref{theorem: main result}.\looseness=-1

\begin{theorem}[Sample complexity of learning a $2\eta$-approximate Nash Equilibrium] \label{theorem: main result}
Fix $\delta\in(0,1)$, and let $\eta=\max\{\epsilon_g, \epsilon_l\}$ be the largest error incurred by Algorithms \emph{\texttt{G-LEARN}} and \emph{\texttt{L-LEARN}}, where \begin{align*}\eta\! &=\!\frac{2\tilde{r}}{(1\!-\!\gamma)^2}\left(\sqrt{\frac{1}{2k} \ln(2|\mathcal{S}_l||\mathcal{A}_g|\sqrt{k})}\!+\!\frac{5}{\sqrt{k}}\right)\!\leq\!\tilde{O}\left(\frac1{\sqrt{k}}\right).
\end{align*}
Setting $N_{\text{steps}} = \frac{\tilde{r}\sqrt{k}}{1-\gamma}$, \emph{\texttt{ALTERNATING-MARL}} learns a  $2\eta$-approximate NE with probability at least $1-\delta$,  and sample complexity (which we bound in \ref{thm: bounding sample complexity})
\[\tilde{O}\left(\min\left\{\frac{  \|r_l\|^3_\infty |\mathcal S_g|^2 |\mathcal A_l| |\mathcal A_g| |\mathcal S_l|^5  k^{2|\mathcal S_l|+4.5}}{(1-\gamma)^6} \log\frac1\delta, \frac{k^{6.5} \|r_l\|^3_\infty |\cS_g|^2 |\mathcal A_l| |\cA_g| |\cS_l|^k }{(1-\gamma)^6}\log \frac1\delta\right\}\right),\]
\end{theorem}

\begin{remark}[Tradeoff in $k$ and the Rewards of Subsampling]\label{remark: tradeoff in k} In the mean-field regime,  \cref{theorem: main result} dismantles the exponential dependence on $|\mathcal{A}_l|$, while in the standard parameterization, it removes the dependence on $|\mathcal{A}_l|^n$ that had existed in prior works \cite{yang2020meanfieldmultiagentreinforcement,li2021permutationinvariantpolicyoptimization, gu2021meanfieldcontrolsqlearningcooperative,gu2022meanfieldmultiagentreinforcementlearning,anand2025meanfield}. As $k\to n$, the approximate NE approaches a NE, revealing a fundamental trade-off in the choice of $k$: increasing $k$ improves the performance of the policy, but increases the size of the $\hat{Q}_k$-function.  
Finally, setting $k=O(\log n)$  gives a sample complexity of $\tilde{O}(|\cS_g|^4 |\cA_g|^2 |\cA_l|^2 \cdot \min\{|\cS_l|^8 (\log n)^{|\cS_l|}, |\cS_l|^{\log n}\})$ with a decaying optimality gap of  $1/\sqrt{\log n}$. Notably, this breaks prior heavy dependencies on $|\mathcal{A}_l|$ and only has a polylogarithmic dependence on the number of agents $n$, which is a doubly-exponential improvement.  \looseness=-1
\end{remark}

\section{Conclusion}
We introduce a mean-field framework for communication-constrained MARL by formulating this problem as best-response dynamics under a Markov game. We propose \texttt{ALTERNATING-MARL}, which learns each agent’s best response to the mean effect of the local agents policies from $k$ samples of the local agents, enabling an exponential reduction in the sample complexity of approximating a solution to the MDP. We prove that the resultant policy is a $\smash{\tilde{O}(\frac{1}{\sqrt{k}})}$-approximate Nash equilibrium. We extend our result to the off-policy learning setting with stochastic rewards and validate our theory with numerical simulations in communication-constrained control settings.

\section*{Impact Statement}
This paper contributes to the theoretical foundations of MARL, with the goal of developing mean-field tools for controlling networked systems. The work can potentially lead to algorithms to adaptively control of cyber-physical systems. Furthermore, any applications of the proposed algorithm in its current form should be considered cautiously since the analysis here focuses on efficiency and stationarity, but does not consider issues of fairness in the vector-reward framework \cite{pmlr-v19-abernethy11b}.

\section*{Acknowledgements} We are deeply grateful to Guannan Qu, Zaiwei Chen, Sarah Liaw, Jan van den Brand, and Jacob Abernethy for sharing their helpful ideas and insightful discussions. Emile's research is supported by NSF Grant CCF 2338816, and Ishani's research is supported by NSF Grant CCF-1955039, and a PayPal research award.

\bibliography{main}
\bibliographystyle{unsrt}

\newpage

\appendix
\onecolumn

\textbf{Outline of the Appendices}.
\begin{itemize}
    \item \cref{sec: mathematical background and additional remarks} provides mathematical background and some additional remarks,
    \item \cref{sec: optimality of global agent policy} gives the proof of the approximate optimality of the learned global agent policy,
    \item \cref{sec: optimality of the local agent policy} presents the proof of the approximate optimality of the local agents' policies,
    \item \cref{sec: proof of convergence of BR dynamics} gives convergence rates for Nash Equilibria and Approximate Nash Equilibria,
    \item \cref{Appendix/off-policy} provides an extension of the result to off-policy learning,
    \item \cref{sec: generalize to stochastic rewards} provides an extension of the result for stochastic rewards, and
    \item \cref{sec: numerical simulations} presents numerical simulations to validate our theoretical framework.
    \item \cref{section: limitations and future work} discusses the limitations of this work and possible directions for future works.
\end{itemize}

% \tableofcontents

\section*{Impact Statement}
This paper contributes to the theoretical foundations of MARL, with the goal of developing mean-field tools for controlling networked systems. The work can potentially lead to algorithms to adaptively control cyber-physical systems. Furthermore, any applications of the proposed algorithm in its current form should be considered cautiously since the analysis here focuses on efficiency and stationarity, but does not consider issues of fairness, for instance via a vector-reward framework \cite{pmlr-v19-abernethy11b}. 

\section{Mathematical Background and Additional Remarks}
\label{sec: mathematical background and additional remarks}
  
\let\addcontentsline=\origtoc

\begin{table}[hbt!]
  \centering
  \caption{Important notations in this paper.}\label{table:notations}
  \begin{tabular}{c|l}
  \hline
    \textbf{Notation} & \textbf{Meaning} \\
 \hline
    $\|\cdot\|_1$ & $\ell_1$ (Manhattan) norm; \\
    $\|\cdot\|_\infty$ & $\ell_\infty$ norm; \\
    $\mathbb{Z}_{+}$ & The set of strictly positive integers; \\
    $\mathbb{R}^d$ & The set of $d$-dimensional reals;\\
    $[m]$ & The set $\{1,\dots,m\}$, where $m\in\mathbb{Z}_+$;  \\
    ${[m]\choose k}$ & The set of $k$-sized subsets of $\{1,\dots,m\}$;\\
    $a_g$ & $a_g\in\mathcal{A}_g$ is the action of the global agent;\\
    $s_g$ & $s_g\in\mathcal{S}_g$ is the state of the global agent;\\
    $a_1,\dots,a_n$ & $a_{1},\dots,a_n\in\mathcal{A}_l^n$ are the actions of the local agents $1,\dots,n$;\\
    $s_1,\dots,s_n$ & $s_{1},\dots,s_n\in\mathcal{S}_l^n$ are the states of the local agents $1,\dots,n$;\\
    $a$ & $a=(a_g,a_1,\dots,a_n)\in\mathcal{A}_g\times\mathcal{A}_l^n$ is the tuple of actions of all agents;\\
    $s$ & $s=(s_g,s_1,\dots,s_n)\in\mathcal{S}_g\times\mathcal{S}_l^n$ is the tuple of states of all agents;\\
    $z_i$ & $z_i = (s_i,a_i) \in \cZ_l$, for $i\in[n]$;\\
    $\mu_k(\cS_l)$ &  $\mu_k(\cS_l) = \{v \in \{0,1/k,2/k,\dots,1\}^{|\cS_l|}: \sum_x v_x = 1\}$; \\ 
    $\mu(\cS_l)$ &  $\mu(\cS_l) \coloneq \{v\in 0,1/n,2/n,\dots,1\}^{|\cS_l|}: \sum_x v_x = 1\}$; \\ 
    $s_\Delta$ & For $\Delta\subseteq [n]$, and a collection of variables $\{s_1,\dots,s_n\}$, $s_\Delta\coloneq\{s_i: i\in\Delta\}$;\\
    $\sigma(z_{\Delta}, z'_{\Delta})$ & Product sigma-algebra generated by sequences $z_{\Delta}$ and $z'_{\Delta}$; \\
    $\pi^*$ & $\pi^*$ is the optimal deterministic policy function such that $a = \pi^*(s)$; \\
    $\hat{\pi}_k^*$ & $\hat{\pi}_k^*$ is the optimal deterministic policy function on a constrained system of \\ & \quad\quad $|\Delta|=k$ local agents; \\
    $\tilde{\pi}^\est_k$ & $\tilde{\pi}^\est_k$ is the stochastic policy map learned with parameter $k$ such that $a\sim\tilde{\pi}_k^\est(s)$;\\ 
    $P_g(\cdot|s_g,a_g)$ & $P_g(\cdot|s_g,a_g)$ is the stochastic transition kernel for the state of the global agent; \\ 
    $P_l(\cdot|s_i,s_g, a_i)$ & $P_l(\cdot|s_i,s_g, a_i)$ is the stochastic transition kernel for the state of local agent $i\in[n];$\\  
    $r_g(s_g,a_g)$ & $r_g$ is the global agent's component of the reward;\\
    $r_l(s_i,s_g,a_i)$ & $r_l$ is the component of the reward for local agent $i\in[n]$;\\
    $r(s,a)$ & $r(s,a) = r_g(s_g,a_g)+\frac{1}{n}\sum_{i\in[n]}r_l(s_i,s_g,a_i)$ is the reward of the system; \\
    $r_\Delta(s,a)$ & $r_\Delta(s,a) = r_g(s_g,a_g)+\frac{1}{|\Delta|}\sum_{i\in\Delta}r_l(s_i,s_g,a_i)$ is the constrained system's reward\\ &\quad\quad with $|\Delta|=k$ local agents;\\
    $\mathcal{T}$ & $\mathcal{T}$ is the centralized Bellman operator;\\
    $\hat{\mathcal{T}}_k$ & $\hat{\mathcal{T}}_k$ is the Bellman operator on a constrained system of $|\Delta|=k$ local agents;\\    \hline
  \end{tabular}
  \medskip
\end{table}

\begin{definition}[Lipschitz continuity] Given metric spaces $(\mathcal{X}, d_\mathcal{X})$ and $(\mathcal{Y}, d_\mathcal{Y})$ and a constant $L>0$, a map $f:\mathcal{X}\to \mathcal{Y}$ is $L$-Lipschitz continuous if for all $x,y\in \mathcal{X}$, $
    d_\mathcal{Y}(f(x), f(y)) \leq L \cdot  d_\mathcal{X}(x,y)$.
\end{definition} 

\begin{theorem}[Banach-Cacciopoli fixed point theorem \citep{Banach1922}]
Consider the metric space $(\mathcal{X}, d_\mathcal{X})$, and $T: \mathcal{X}\to \mathcal{X}$ such that $T$ is a $\gamma$-Lipschitz continuous mapping for $\gamma \in (0,1)$. Then, by the Banach-Cacciopoli fixed-point theorem, there exists a unique fixed point $x^* \in \mathcal{X}$ for which $T(x^*) = x^*$. Additionally, $x^* = \lim_{s\to\infty} T^s( x_0)$ for any $x_0 \in \mathcal{X}$.
\end{theorem}

\begin{lemma}[Freedman's inequality, Lemma D.2 in \cite{liu2023maximize}]
Let $\{X_t\}_{t\leq T}$ be a real-valued martingale difference sequence adapted to filtration $\{\cF_t\}_{t\leq T}$. If $|X_t|\leq R$ almost surely, then for any $\eta\in (0, 1/R)$, it holds that
\begin{equation}
    \Pr\left[\sum_{t=1}^T X_t \leq \mathcal{O}\left(\eta\sum_{t=1}^T \E[X_t^2|\cF_{t-1}] + \frac{\log(1/\delta)}{\eta}\right)\right]\geq 1-\delta.\\
\end{equation}
\end{lemma}

\begin{lemma}[Lemma 11 in \cite{abbasi2011improved}] Let $\{x_s\}_{s\in[T]}$ be a sequence of vectors with $x_s \in \cV$ for some Hilbert space $\cV$. Let $\Lambda_0$ be a positive definite matrix and define $\Lambda_t = \Lambda_0 + \sum_{s=1}^t x_s x_s^\top$. Then, it holds that
\begin{equation}\sum_{s=1}^T \min\{1,\|x_s\|_{\Lambda_{s-1}^{-1}}\} \leq 2\log\left (\frac{\det(\Lambda_T)}{\det(\Lambda_0)}\right).\\
\end{equation}
\end{lemma}

\begin{lemma}[Multiplayer advantage decomposition from \cite{zhong2024heterogeneous}] In any cooperative Markov game, given a joint policy $\pi$, for any state $s$ and any agent subset $i_{1:m}$, we have that 
\begin{equation}A_\pi^{i_{1:m}}(s,a^{i_{1:m}}) = \sum_{j=1}^m A_\pi^{i_j} (s, a^{i_{1:j-1}}, a^{i_j}).\\
\end{equation}
    
\end{lemma}

\begin{figure}[hbt!]
    \centering
    \includegraphics[width=\linewidth]{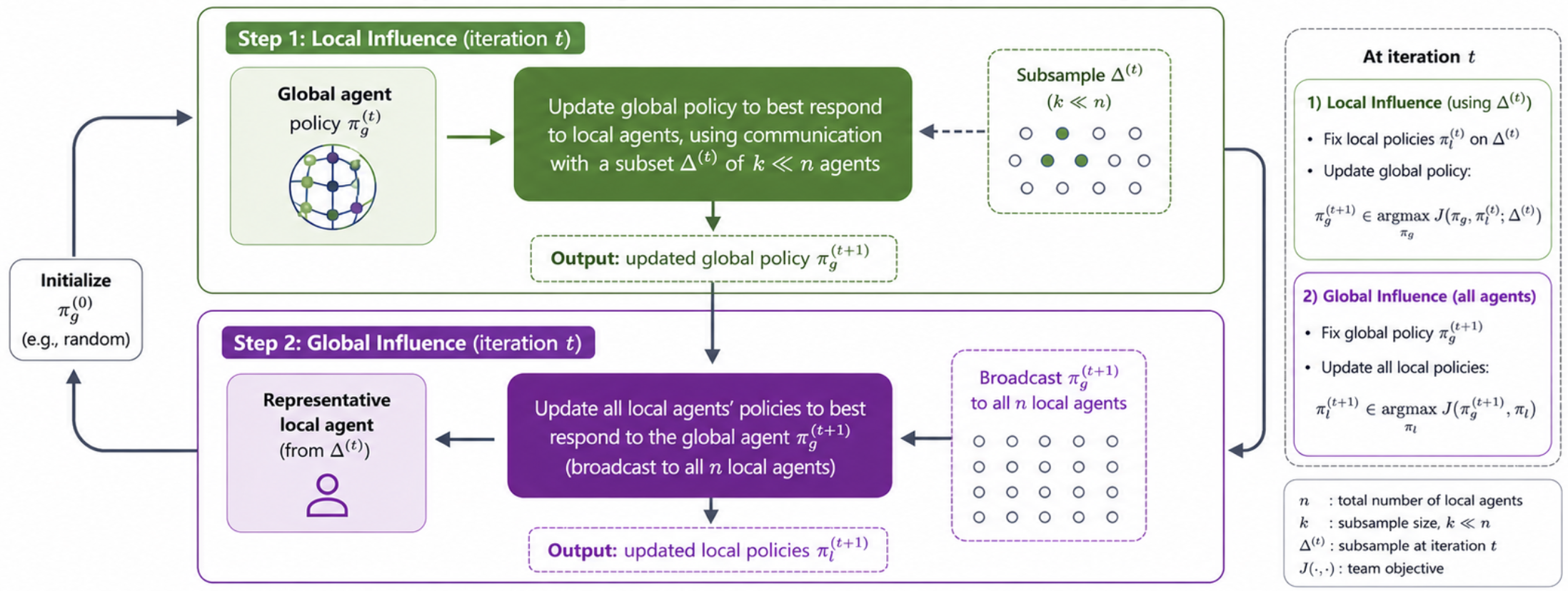}
    \caption{The alternating best-response dynamics, which repeat until convergence. In it, the global agent does a best-response update, treating the policy representative $\pi_\ell$ fixed, and then the local agent representative does a best-response update, treating the global agent's policy $\pi_g$ as fixed. We note that generative AI was used to refine the aesthetics of this figure.}
    \label{fig:alternating-marl-pipeline}
\end{figure}

\section{Optimality of the Global Agent Policy}
\label{sec: optimality of global agent policy}

In this section, we show that the policy learned by the global agent is a $\tilde{O}(1/\sqrt{k})$ best-response to the local agent's fixed policy $\pi_\ell$.  The structure of the proofs in this section closely mirror those of \cite{anand2025meanfield}. For convenience, we begin by restating the various Bellman operators under consideration below.

\begin{definition}[Bellman Operator $\mathcal{T}$]\label{defn:bellman}
\begin{equation}
\mathcal{T}Q^t(s,a_g) \coloneqq r^{\pi_\ell}(s,a) + \gamma\E_{\substack{s_g'\sim P_g(\cdot|s_g,a_g),\\ s_i'\sim \bar{P}^{\pi_\ell}_l(\cdot|s_i,s_g),\forall i\in[n]}} \max_{a_g'\in\mathcal{A}_g} Q^t(s',a_g').
\end{equation}
\end{definition}

\begin{definition}[Adapted Bellman Operator $\hat{\mathcal{T}}_k$]\label{defn:adapted bellman} \emph{The adapted Bellman operator updates a smaller $Q$ function (which we denote by $\hat{Q}_k$), for a surrogate system with the global agent and $k\in[n]$ local agents denoted by $\Delta$, using mean-field value iteration and $j\in\Delta$ such that}
\begin{equation}
\hat{\mathcal{T}}_k\hat{Q}_k^t(s_g,  F_{s_{\Delta}}, a_g) \coloneqq r^{\pi_\ell}_\Delta(s,a_g) + \gamma \E_{\substack{s_g'\sim P_g(\cdot|s_g,a_g), \\ s_i'\sim \bar{P}^{\pi_\ell}_l(\cdot|s_i,s_g),\forall i\in\Delta}} \max_{a_g'\in\mathcal{A}_g} \hat{Q}_k^t(s_g',  F_{s'_{\Delta}},a'_g).
\end{equation}
\end{definition}
\begin{definition}[Empirical Adapted Bellman Operator $\hat{\mathcal{T}}_{k,m}$] \label{defn:empirical adapted bellman} \emph{The empirical adapted Bellman operator $\hat{\mathcal{T}}_{k,m}$ empirically estimates the adapted Bellman operator update using mean-field value iteration by drawing $m$ random samples of $s_g\sim P_g(\cdot|s_g,a_g)$ and $s_i\sim \bar{P}^{\pi_\ell}_l(\cdot|s_i,s_g)$ for $i\in\Delta$, where for $\ell\in[m]$, the $\ell$'th random sample is given by $s_g^\ell$ and $s_\Delta^\ell$, and $j\in\Delta$:}
\begin{equation}
\hat{\mathcal{T}}_{k,m}\hat{Q}_{k,m}^t(s_g, F_{s_{\Delta}}, a_g)\coloneqq r^{\pi_\ell}_\Delta(s,a_g) + \frac{\gamma}{m} \sum_{\ell \in [m]}\max_{a_g'\in\mathcal{A}_g} \hat{Q}_{k,m}^t(s_g^\ell,F_{s_{\Delta}^\ell},a_g').
\end{equation}
\end{definition} 
 
\begin{lemma}
\label{lemma: Q-bound}
{For any $\Delta\subseteq[n]$ such that $|\Delta|=k$, suppose $0\leq \bar{r}^{\pi_\ell}_\Delta(s,a_g)\leq \tilde{r}$. Then, $\hat{Q}_k^t \leq \frac{\tilde{r}}{1-\gamma}$.}
\end{lemma}
\begin{proof}
    We prove this by induction on $t\in\mathbb{N}$. The base case is satisfied as $\hat{Q}_k^0 = 0$. Assume that $\|\hat{Q}_k^{t-1}\|_\infty \leq \frac{\tilde{r}}{1-\gamma}$. We bound $\hat{Q}_k^{t+1}$ from the Bellman update at each time step as follows, for all $s_g\in\mathcal{S}_g,F_{s_\Delta}\in \mu_k({\mathcal{S}_l}), a_g\in\mathcal{A}_g$:
    \begin{align*}
        \hat{Q}_k^{t+1}(s_g,F_{s_\Delta},a_g) &= \bar{r}^{\pi_\ell}_\Delta(s,a_g) + \gamma\mathbb{E}_{\substack{s_g'\sim P_g(\cdot|s_g, a_g), \\ s_i' \sim \bar{P}^{\pi_\ell}_l(\cdot|s_i,s_g), \forall i\in\Delta}}\max_{a_g'\in\mathcal{A}_g}\hat{Q}_k^t(s_g',F_{s_\Delta'},a_g') \\
        &\leq \tilde{r} + \gamma \max_{\substack{a_g'\in\mathcal{A}_g, s_g' \in \mathcal{S}_g, F_{s_\Delta'} \in \mu_k{(\mathcal{S}_l)}}}\hat{Q}_k^t(s_g',F_{s_\Delta'},a_g') \leq \frac{\tilde{r}}{1-\gamma}.
    \end{align*}
Here, the first inequality follows by noting that the maximum value of a random variable is at least as large as its expectation. The second inequality follows from the inductive hypothesis.\qedhere
\end{proof}

\begin{remark}\cref{lemma: Q-bound} is independent of the choice of $k$. Therefore, for $k=n$, this implies an identical bound on $Q^t$. A similar argument as \cref{lemma: Q-bound} implies an identical bound on $\hat{Q}_{k,m}^t$.
\end{remark}

Recall that the original Bellman operator $\mathcal{T}$ satisfies a $\gamma$-contractive property under the infinity norm. We similarly show that $\hat{\mathcal{T}}_k$ and $\hat{\mathcal{T}}_{k,m}$ satisfy a $\gamma$-contractive property under infinity norm in \cref{lemma: gamma-contraction of adapted Bellman operator} and \cref{lemma: gamma-contraction of empirical adapted Bellman operator}.

\begin{lemma}\label{lemma: gamma-contraction of adapted Bellman operator}
{$\hat{\mathcal{T}}_k$ satisfies the $\gamma$-contractive property under infinity norm:}
    \[\|\hat{\mathcal{T}}_k\hat{Q}_k' - \hat{\mathcal{T}}_k\hat{Q}_k\|_\infty \leq \gamma \|\hat{Q}_k' - \hat{Q}_k\|_\infty\]
\end{lemma}
\begin{proof} Suppose we apply $\hat{\mathcal{T}}_k$ to $\hat{Q}_k(s_g,F_{s_\Delta}, a_g)$ and $\hat{Q}'_k(s_g, F_{s_\Delta}, a_g)$ for $|\Delta|=k$. Then:
\begin{align*}
&\|\hat{\mathcal{T}}_k\hat{Q}_k'- \hat{\mathcal{T}}_k\hat{Q}_k\|_\infty \\
&= \gamma \max_{\substack{s_g\in \mathcal{S}_g,\\ a_g\in\mathcal{A}_g,\\ F_{s_\Delta} \in \mu_k{(\mathcal{S}_l)}}}\!\left| \mathbb{E}_{\substack{s_g'\sim P_g(\cdot|s_g,a_g),\\ s_i'\sim \bar{P}^{\pi_\ell}_l(\cdot|s_i, s_g),\\ \forall s_i'\in s_\Delta',\\ }}\max_{a_g'\in\mathcal{A}_g}\hat{Q}_k'(s_g', F_{s_\Delta'}, a_g')- \mathbb{E}_{\substack{s_g'\sim P_g(\cdot|s_g,a_g),\\ s_i'\sim \bar{P}^{\pi_\ell}_l(\cdot|s_i, s_g),\\\forall s_i'\in s_\Delta'
}}\max_{a_g'\in\mathcal{A}_g}\hat{Q}_k(s_g', F_{s_\Delta'}, a_g')\right|\\
    &\leq \gamma  \max_{\substack{s_g' \in \mathcal{S}_g, F_{s_\Delta'} \in \mu_k{(\mathcal{S}_l)}, a_g'\in\mathcal{A}_g
        }}\left| \hat{Q}_k'(s_g', F_{s_\Delta'}, a_g') -  \hat{Q}_k(s_g', F_{s_\Delta'}, a_g')\right| \\
        &= \gamma \|\hat{Q}_k' - \hat{Q}_k\|_\infty.
    \end{align*}
The equality implicitly cancels the common $r_\Delta(s, a_g)$ terms from each application of the adapted-Bellman operator. The inequality follows by triangle inequality and Jensen's inequality (by the convexity of $|\cdot|$), and maximizing over the random variables. The last line recovers the definition of infinity norm. \qedhere
\end{proof}
 
\begin{lemma}
{$\hat{\mathcal{T}}_{k,m}$ satisfies the $\gamma$-contractive property under infinity norm.\label{lemma: gamma-contraction of empirical adapted Bellman operator}}
\end{lemma}
\begin{proof}
Similarly to \cref{lemma: gamma-contraction of adapted Bellman operator}, suppose we apply $\hat{\mathcal{T}}_{k,m}$ to $\hat{Q}_{k,m}(s_g,F_{s_\Delta},a_g)$ and $\hat{Q}_{k,m}'(s_g,F_{s_\Delta},a_g)$. Then:
\begin{align*}    \|\hat{\mathcal{T}}_{k,m}\hat{Q}_{k,m} - \hat{\mathcal{T}}_{k,m}\hat{Q}'_{k,m}\|_\infty &= \frac{\gamma}{m}\left\|\sum_{j\in [m]} (\max_{a_g'\in\mathcal{A}_g} \hat{Q}_{k,m}(s_g^j,F_{s_\Delta^j},a_g') -  \max_{a_g'\in\mathcal{A}_g} \hat{Q}'_{k,m}(s_g^j,F_{s_\Delta^j},a_g'))\right\|_\infty \\
    &\leq \gamma \max_{\substack{a_g'\in\mathcal{A}_g, s_g' \in\mathcal{S}_g, s_\Delta\in\mathcal{S}_l^k}}|\hat{Q}_{k,m}(s_g', F_{s_\Delta'}, a_g') - \hat{Q}'_{k,m}(s_g', F_{s_\Delta'}, a_g')| \\
    &\leq \gamma \|\hat{Q}_{k,m} - \hat{Q}'_{k,m}\|_\infty
\end{align*} 
 The first inequality uses the triangle inequality and the general property $|\max_{a\in A}f(a) - \max_{a\in A}g(a)| \leq \max_{a\in A}|f(a) - g(a)|$. In the last line, we recover the definition of infinity norm.\qedhere
\end{proof}

\begin{remark}\label{remark: gamma-contractive of ad_T}The $\gamma$-contractivity of $\hat{\mathcal{T}}_k$ and $\hat{\mathcal{T}}_{k,m}$ attracts the trajectory between two $\hat{Q}_k$ and $\hat{Q}_{k,m}$ functions on the same state-action tuple by $\gamma$ at each step. Repeatedly applying the Bellman operators produces a unique fixed-point from the Banach fixed-point theorem which we introduce in \cref{defn:qkstar,defn:qkmest}.
\end{remark}

\begin{definition}[$\hat{Q}_k^*$-function]\label{defn:qkstar}\emph{Suppose $\hat{Q}_k^0:=0$ and let $\hat{Q}_k^{t+1}(s_g,F_{s_\Delta},a_g) = \hat{\mathcal{T}}_k \hat{Q}_k^t(s_g,F_{s_\Delta},a_g)$
    for $t\in\N$. Denote the fixed-point of $\hat{\mathcal{T}}_k$ by $\hat{Q}^*_k$ such that $\hat{\mathcal{T}}_k \hat{Q}^*_k(s_g,F_{s_\Delta},a_g) = \hat{Q}^*_k(s_g,F_{s_\Delta},a_g)$.}
\end{definition}

\begin{definition}[$\hat{Q}_{k,m}^\est$-function]\label{defn:qkmest}\emph{Suppose $\hat{Q}_{k,m}^0:=0$ and let $\hat{Q}_{k,m}^{t+1}(s_g,F_{s_\Delta},a_g) = \hat{\mathcal{T}}_{k,m} \hat{Q}_{k,m}^t(s_g,F_{s_\Delta},a_g)$
    for $t\in\N$. Denote the fixed-point of $\hat{\mathcal{T}}_{k,m}$ by $\hat{Q}^\est_{k,m}$ such that $\hat{\mathcal{T}}_{k,m} \hat{Q}^\est_{k,m}(s_g,F_{s_\Delta},a_g) = \hat{Q}^\est_{k,m}(s_g,F_{s_\Delta},a_g)$.}
\end{definition}

\begin{corollary}\label{corollary:backprop}
    \emph{By backpropagating results of the $\gamma$-contractive property for $T$  time steps, we have }\begin{equation}\|\hat{{Q}}_k^* - \hat{{Q}}_k^T\|_\infty \leq \gamma^T \cdot \|\hat{Q}_k^* - \hat{Q}_k^0\|_\infty\end{equation}
    and 
    \begin{equation}
    \|\hat{{Q}}_{k,m}^\est - \hat{{Q}}_{k,m}^T\|_\infty \leq \gamma^T \cdot \|\hat{Q}_{k,m}^\est - \hat{Q}^0_{k,m}\|_\infty.
    \end{equation}
    
Further, noting that $\hat{Q}_k^0 = \hat{Q}_{k,m}^0 := 0$, $\|\hat{Q}_k^*\|_\infty \leq \frac{\tilde{r}}{1-\gamma}$, and $\|\hat{Q}_{k,m}^\est\|_\infty \leq \frac{\tilde{r}}{1-\gamma}$ from \cref{lemma: Q-bound}, we have:
\begin{equation}\label{eqn: qstar,qt decay}\|\hat{Q}_k^* - \hat{Q}_k^T\|_\infty \leq \gamma^T \frac{\tilde{r}}{1-\gamma},\end{equation}
and
\begin{equation}\|\hat{Q}_{k,m}^\est - \hat{Q}_{k,m}^T\|_\infty \leq \gamma^T \frac{\tilde{r}}{1-\gamma}.
    \end{equation}
\end{corollary}

\begin{remark}
\cref{corollary:backprop} characterizes the error decay between $\hat{Q}_k^T$ and $\hat{Q}_k^*$ and shows that it decays exponentially in the number of Bellman iterations by a $\gamma^T$ multiplicative factor.
\end{remark}

Furthermore, we characterize the maximal greedy policies obtained from $Q^*, \hat{Q}_k^*$, and $\hat{Q}_{k,m}^\est$.

\begin{definition}[Optimal policy $\pi^*$] \emph{The greedy policy derived from $Q^*$ is given by }\begin{equation}\pi^*(s) \coloneqq \arg\max_{a_g \in\mathcal{A}_g} Q^*(s,a_g).\end{equation}
\end{definition}

\begin{definition}[Optimal subsampled policy $\hat{\pi}_k^*$]
    \emph{The greedy policy from $\hat{Q}_k^*$ is}
\begin{equation}\hat{\pi}_k^*(s_g, F_{s_{\Delta }}) \coloneq \mathop{\argmax}_{a_g \in\mathcal{A}_g} \hat{Q}_k^*(s_g, F_{s_{\Delta}}, a_g).\end{equation}
\end{definition}
\begin{definition}[Optimal empirically subsampled policy $\hat{\pi}_{k,m}^\est$]
    \emph{The greedy policy from $\hat{Q}_{k,m}^\est$ is given by }\begin{equation}\hat{\pi}_{k,m}^\est(s_g,F_{s_\Delta})\coloneqq \mathop{\argmax}_{a_g\in \mathcal{A}_g}\hat{Q}_{k,m}^\est(s_g, F_{s_{\Delta}}, a_g).\end{equation}
\end{definition}

\begin{theorem} {For any state $s\in\mathcal{S}_g\times\mathcal{S}_l^n$},\label{theorem: performance_difference_lemma_applied}
\begin{align*}V^{\pi^*}(s) - V^{{{\pi}}^\est_{k,m}}(s) &\leq \frac{2\tilde{r}}{(1-\gamma)^2}\left(\sqrt{\frac{n-k+1}{2nk} \ln(2|\mathcal{S}_l||\mathcal{A}_g|\sqrt{k})}+\frac{1}{\sqrt{k}}\right)+\frac{2\epsilon_{k,m}}{1-\gamma}. \\ \end{align*}
\end{theorem}

\subsection{Proof of Lipschitz-Continuity Bound}

\begin{definition}[Joint Stochastic Kernels] 
\label{definition: joint_transition_probability}\emph{The joint stochastic kernel on $(s_g,s_\Delta)$ for $\Delta\subseteq [n]$ where $|\Delta|=k$ is defined as $\mathcal{J}_k:\mathcal{S}_g\times\mathcal{S}_l^{k}\times\mathcal{S}_g\times\mathcal{A}_g\times\mathcal{S}_l^{k}\to[0,1]$, where} \begin{equation}\mathcal{J}_k(s_g',s_\Delta'|s_g, a_g, s_\Delta) := \Pr[(s_g',s_\Delta')|s_g, a_g, s_\Delta].\end{equation}
\end{definition}

\begin{theorem}[$\hat{Q}_k^T$ is $(\sum_{t=0}^{T-1} 2\gamma^t)\|r_l(\cdot,\cdot)\|_\infty$-Lipschitz continuous with respect to $F_{s_\Delta}$ in total variation distance]\label{lemma: Q lipschitz continuous wrt Fsdelta}
Suppose $\Delta,\Delta'\subseteq[n]$ such that $|\Delta|=k$ and $|\Delta'|=k'$. Then, we have that
\begin{equation}\left|\hat{Q}^T_k(s_g,F_{s_\Delta},a_g) - \hat{Q}^T_{k'}(s_g, F_{s_{\Delta'}}, a_g)\right| \leq \left(\sum_{t=0}^{T-1} 2\gamma^t\right)\|r_l(\cdot,\cdot)\|_\infty \cdot \mathrm{TV}\left(F_{s_\Delta}, F_{s_{\Delta'}}\right). \end{equation}
\end{theorem}
\begin{proof} 
We prove this inductively. Note that  $\hat{Q}_k^0 = \hat{Q}_{k'}^0 =0$ by construction, proving the base case since $\mathrm{TV}(\cdot,\cdot)\geq 0$. For the remainder of this proof, we adopt the shorthand $\E_{s_g',s_\Delta'}$ to refer to $\E_{s_g'\sim P_g(\cdot|s_g, a_g), s_i'\sim \bar{P}^{\pi_\ell}_l(\cdot|s_i,s_g),\forall i\in\Delta}$. At $T=1$:
\begin{align*}
|\hat{Q}_k^1(s_g,F_{s_\Delta},a_g) &-\hat{Q}_{k'}^1(s_g, F_{s_{\Delta'}}, a_g)| \\ &= \left|\hat{\mathcal{T}}_k\hat{Q}_k^0(s_g,F_{s_\Delta},a_g)-\hat{\mathcal{T}}_{k'}\hat{Q}_{k'}^0(s_g, F_{s_{\Delta'}}, a_g)\right| \\
&= |\bar{r}^{\pi_\ell}(s_g, F_{s_\Delta}, a_g)+\gamma \E_{s_g', s'_\Delta}\max_{a_g'\in\mathcal{A}_g} \hat{Q}_k^0(s_g',F_{s'_\Delta},a'_g) \\
&\quad\quad - \bar{r}^{\pi_\ell}(s_g,F_{s_{\Delta'}},a_g) -\gamma \E_{s_g', s'_{\Delta'}} \max_{a_g'\in\mathcal{A}_g}\hat{Q}_{k'}^0(s'_g, F_{s'_{\Delta'}}, a'_g)| \\
&= |\bar{r}^{\pi_\ell}(s_g, F_{s_\Delta}, a_g)-\bar{r}^{\pi_\ell}(s_g, F_{s_{\Delta'}},a_g)| \\
&= \left|\frac{1}{k}\sum_{i\in\Delta}\bar{r}^{\pi_\ell}_l(s_g,s_i) -\frac{1}{k'}\sum_{i\in\Delta'}\bar{r}^{\pi_\ell}_l(s_g,s_i)\right| \\
&= |\E_{s_l \sim F_{s_\Delta}}\bar{r}^{\pi_\ell}_l(s_g, s_l) - \E_{s_l' \sim F_{s_{\Delta'}}}\bar{r}^{\pi_\ell}_l(s_g, s_l')|.
\end{align*}

In the first and second equalities, apply the adapted Bellman operators $\hat{\mathcal{T}}_k$ and $\hat{\mathcal{T}}_{k'}$ to $\hat{Q}_k^0$ and $\hat{Q}_{k'}^0$, respectively, and expand.  Noting the property that for any function $f:\mathcal{X}\to\mathcal{Y}$ for $|\mathcal{X}|<\infty$ we can write $f(x) = \sum_{y\in\mathcal{X}}f(y)\mathbbm{1}\{y=x\}$, we have:
\begin{align*}
|\hat{Q}^1_k(s_g,F_{s_\Delta},a_g) &- \hat{Q}_{k'}^1(s_g, F_{s_{\Delta'}}, a_g)| \\ &= \left|\E_{s_l \sim F_{s_\Delta}}\left[\sum_{z\in \mathcal{S}_l}\bar{r}^{\pi_\ell}_l(s_g, z)\mathbbm{1}\{s_l = z\}\right] - \E_{s_l' \sim F_{s_{\Delta'}}}\left[\sum_{z\in\mathcal{S}_l}\bar{r}^{\pi_\ell}_l(s_g, z)\mathbbm{1}\{s_l' = z\}\right]\right| \\
&= \left|\sum_{z\in\mathcal{S}_l}\bar{r}^{\pi_\ell}_l(s_g,z)\cdot (\E_{s_l\sim F_{s_\Delta}}\mathbbm{1}\{s_l=z\} - \E_{s_l'\sim F_{s_{\Delta'}}}\mathbbm{1}\{s_l'=z\})\right| \\
&= \left|\sum_{z\in\mathcal{S}_l}\bar{r}^{\pi_\ell}_l(s_g,z)\cdot (F_{s_\Delta}(z) - F_{s_{\Delta'}}(z))\right|  \\
&\leq \left|\max_{z\in\mathcal{S}_l}\bar{r}^{\pi_\ell}_l(s_g,z)|\cdot \sum_{z\in\mathcal{S}_l}|F_{s_\Delta}(z) - F_{s_{\Delta'}}(z)\right| \leq 2\|r_l(\cdot,\cdot)\|_\infty \cdot \mathrm{TV}(F_{s_\Delta}, F_{s_{\Delta'}}).
\end{align*}
The third equality uses that for any random variable $X\sim\mathcal{X}$, $\E_X\mathbbm{1}[X=x]=\Pr[X=x]$. The first inequality uses triangle inequality and Cauchy-Schwarz, and the second inequality uses the definition of TV-distance. Thus, when $T=1$, $\hat{Q}$ is $(2\|r_l(\cdot,\cdot)\|_\infty)$-Lipschitz continuous with respect to TV-distance, proving the base case. 

Now, assume that for $T\leq t'\in\N$:
\begin{align*}\left|\hat{Q}^{T}_k(s_g,F_{s_\Delta}, a_g) - \hat{Q}_{k'}^{T}(s_g,F_{s_{\Delta'}}, a_g)\right| \leq \left(\sum_{t=0}^{T-1}2\gamma^t\right)\|r_l(\cdot,\cdot)\|_\infty\cdot\mathrm{TV}\left(F_{s_\Delta}, F_{s_{\Delta'}}\right).\end{align*}
Then, inductively we have:
\begin{align*}
    |\hat{Q}_k^{T+1}(s_g,F_{s_\Delta}, a_g) &- \hat{Q}_{k'}^{T+1}(s_g,F_{s_{\Delta'}}, a_g)|\\
&\leq\left|\frac{1}{|\Delta|}\sum_{i\in\Delta}\bar{r}^{\pi_\ell}_l(s_g,s_i)-\frac{1}{|\Delta'|}\sum_{i\in\Delta'}\bar{r}^{\pi_\ell}_l(s_g,s_i)\right| \\
&\quad\quad\quad +\gamma\left|\E_{s_g', s'_{\Delta}}\max_{a_g'\in\mathcal{A}_g}\hat{Q}_k^{T}(s_g',F_{s_\Delta'}, a_g')-\E_{s_g', s'_{\Delta'}}\max_{a_g'\in\mathcal{A}_g}\hat{Q}_{k'}^{T}(s_g',F_{s_{\Delta'}'}, a_g')\right| \\
&\leq 2\|r_l(\cdot,\cdot)\|_\infty\cdot\mathrm{TV}\left(F_{s_\Delta}, F_{s_{\Delta'}}\right)  \\
&\quad\quad\quad + \gamma\left|\E_{s_g', s'_{\Delta}}\max_{a_g'\in\mathcal{A}_g}\hat{Q}_k^{T}(s_g',F_{s_\Delta'}, a_g')-\E_{s_g', s'_{\Delta'}}\max_{a_g'\in\mathcal{A}_g}\hat{Q}_{k'}^{T}(s_g',F_{s_{\Delta'}'}, a_g')\right|.\\
\end{align*}
\noindent The first equality uses the time evolution property of $\hat{Q}_k^{T+1}$ and $\hat{Q}_{k'}^{T+1}$ by applying the adapted-Bellman operators $\hat{\mathcal{T}}_k$ and $\hat{\mathcal{T}}_{k'}$ to $\hat{Q}_k^T$ and $\hat{Q}_{k'}^T$, respectively. The first term of the second inequality uses the Lipschitz bound from the base case. The second term then rewrites the expectation over the states $s_g', s_\Delta', s_{\Delta'}'$ into an expectation over the joint transition probabilities $\mathcal{J}_k$ and $\mathcal{J}_{k'}$ from \cref{definition: joint_transition_probability}.

Therefore, using the shorthand $\E_{(s_g',s_\Delta')\sim\mathcal{J}_k}$ to denote $\E_{(s_g',s_\Delta')\sim\mathcal{J}_k(\cdot,\cdot|s_g,a_g,s_\Delta)}$, we have: 
\begin{align*}
    |\hat{Q}_k^{T+1}&(s_g,F_{s_\Delta}, a_g) - \hat{Q}_{k'}^{T+1}(s_g,F_{s_{\Delta'}}, a_g)| \\
    &\quad\leq 2\|r_l(\cdot,\cdot)\|_\infty\cdot\mathrm{TV}(F_{s_\Delta}, F_{s_{\Delta'}})  \\ &\quad\quad\quad+ \gamma|\E_{(s_g', s'_{\Delta})\sim\mathcal{J}_k}\max_{a_g'\in\mathcal{A}_g} \hat{Q}_k^{T}(s_g',F_{s_\Delta'}, a_g') - \E_{(s_g', s'_{\Delta'})\sim\mathcal{J}_{k'}}\max_{a_g'\in\mathcal{A}_g} \hat{Q}_{k'}^{T}(s_g',F_{s_{\Delta'}'}, a_g')|\\
    &\quad\leq 2\|r_l(\cdot,\cdot)\|_\infty\cdot\mathrm{TV}(F_{s_\Delta}, F_{s_{\Delta'}})  \\ &\quad\quad\quad + \gamma\max_{a_g'\in\mathcal{A}_g}|\E_{(s_g', s'_{\Delta})\sim\mathcal{J}_k} \hat{Q}_k^{T}(s_g', F_{s'_\Delta}, a_g') - \E_{(s_g', s'_{\Delta'})\sim\mathcal{J}_{k'}} \hat{Q}_{k'}^{T}(s_g', F_{s'_{\Delta'}}, a_g')| \\
&\quad\quad\leq 2\|r_l(\cdot,\cdot)\|_\infty\cdot\mathrm{TV}(F_{s_\Delta}, F_{s_{\Delta'}}) + \gamma\left(\sum_{\tau=0}^{T-1}2\gamma^\tau\right)\|r_l(\cdot,\cdot)\|_\infty\cdot\mathrm{TV}(F_{s_\Delta},  F_{s_{\Delta'}})\\
&\quad\quad=\left(\sum_{\tau=0}^{T}2\gamma^\tau\right) \|r_l(\cdot,\cdot)\|_\infty\cdot\mathrm{TV}(F_{s_\Delta}, F_{s_{\Delta'}}).
\end{align*}

\noindent In the first inequality, we rewrite the expectations over the states as the expectation over the joint transition probabilities. The second inequality then follows from \cref{lemma: expectation_expectation_max_swap}.   Finally, the third inequality follows from \cref{lemma: expectation Q lipschitz continuous wrt Fsdelta}. Then, by the inductive hypothesis, the claim is proven.\qedhere\\
\end{proof}

\begin{lemma}{\label{lemma: expectation Q lipschitz continuous wrt Fsdelta}For all $T\in\N$, for any $a_g, a_g'\in\mathcal{A}_g,s_g\in\mathcal{S}_g, s_\Delta\in\mathcal{S}_l^{k}$, and for all joint stochastic kernels $\mathcal{J}_k$ (defined in \cref{definition: joint_transition_probability}), $\E_{(s_g',s_\Delta')\sim\mathcal{J}_k(\cdot,\cdot|s_g,a_g,s_\Delta)} \hat{Q}_k^T(s_g', F_{s'_\Delta}, a_g')$ is $(\sum_{t=0}^{T-1} 2\gamma^t)\|r_l(\cdot,\cdot)\|_\infty$-Lipschitz continuous with respect to $F_{s_\Delta}$ in TV-distance:}
\begin{align*}|\E_{(s_g',s_\Delta')\sim\mathcal{J}_k(\cdot,\cdot|s_g,a_g,s_\Delta)}\hat{Q}_k^T(s'_g,F_{s'_\Delta},a_g') &- \E_{(s_g',s_{\Delta'}')\sim\mathcal{J}_{k'}(\cdot,\cdot|s_g,a_g,s_{\Delta'})}\hat{Q}_{k'}^T(s'_g, F_{s'_{\Delta'}}, a_g')| \\
&\quad\quad\quad\quad \leq \left(\sum_{\tau=0}^{T-1}2\gamma^\tau\right)\|r_l(\cdot,\cdot)\|_\infty\cdot\mathrm{TV}\left(F_{s_\Delta}, F_{s_{\Delta'}}\right).
\end{align*}
\end{lemma}
\begin{proof}
We prove this inductively.

At $T=0$, the statement is true since $\hat{Q}_k^0(\cdot,\cdot,\cdot) = \hat{Q}_{k'}^0(\cdot,\cdot,\cdot) = 0$ and $\mathrm{TV}(\cdot,\cdot)\geq 0$. 

For $T=1$, applying the adapted Bellman operator yields: 
\begin{align*}
&|\E_{(s_g',s_\Delta')\sim\mathcal{J}_k(\cdot,\cdot|s_g,a_g,s_\Delta)}\hat{Q}_k^1(s'_g,F_{s'_\Delta},  a_g')- \E_{(s_g',s_{\Delta'}')\sim\mathcal{J}_{k'}(\cdot,\cdot|s_g,a_g,s_{\Delta'})}\hat{Q}_{k'}^1(s'_g, F_{s'_{\Delta'}}, a_g')| \\
&\quad\quad =\left| \E_{(s_g',s_{\Delta\cup\Delta'}')\sim\mathcal{J}_{|\Delta\cup\Delta'|}(\cdot,\cdot|s_g, a_g,s_{\Delta\cup\Delta'})}\left[\frac{1}{|\Delta|}\sum_{i\in\Delta}\bar{r}^{\pi_\ell}_l(s'_g,s'_i) - \frac{1}{|\Delta'|}\sum_{i\in\Delta'}\bar{r}^{\pi_\ell}_l(s_g',s_i')\right]\right| \\
&\qquad = \left|\E_{(s_g',s'_{\Delta\cup\Delta'})\sim\mathcal{J}_{|\Delta\cup\Delta'|}(\cdot,\cdot|s_g,a_g,s_{\Delta\cup\Delta'})}\left[\sum_{z\in\mathcal{S}_l}\bar{r}^{\pi_\ell}_l(s_g',z)\cdot (F_{s'_\Delta}(z) -  F_{s'_{\Delta'}}(z))\right]\right|.
\end{align*}
Similarly to \cref{lemma: Q lipschitz continuous wrt Fsdelta}, we implicitly write the result as an expectation over the reward functions and use the general property that for any function $f:\mathcal{X}\to\mathcal{Y}$ for $|\mathcal{X}|<\infty$, we can write $f(x) = \sum_{y\in\mathcal{X}}f(y)\mathbbm{1}\{y=x\}$. Then, taking the expectation over the indicator variable yields the second equality. As a shorthand, let $\mathfrak{D}$ denote the distribution of $s_{g}'\sim \sum_{s'_{\Delta\cup\Delta'}\in\mathcal{S}_l^{|\Delta\cup\Delta'|}}\mathcal{J}_{|\Delta\cup\Delta|}(\cdot,s'_{\Delta\cup\Delta'}|s_g,a_g,s_{\Delta\cup\Delta'})$.  

Then, by the law of total expectation,
\begin{align*}
&|\E_{(s_g',s_\Delta')\sim\mathcal{J}_k  (\cdot,\cdot|s_g,a_g,s_\Delta)}\hat{Q}_k^1(s'_g,F_{s'_\Delta},a_g') - \E_{(s_g',s_{\Delta'}')\sim\mathcal{J}_{k'}(\cdot,\cdot|s_g,a_g,s_{\Delta'})}\hat{Q}_{k'}^1(s'_g, F_{s'_{\Delta'}}, a_g')| \\
&\qquad = \left|\E_{s_g'\sim \mathfrak{D}} \sum_{z\in\mathcal{S}_l}r_l(s_g',z)\E_{s'_{\Delta\cup\Delta'}\sim\mathcal{J}_{|\Delta\cup\Delta'|}(\cdot|s_g',s_g,a_g,s_{\Delta\cup\Delta'})}(F_{s'_\Delta}(z) - F_{s'_{\Delta'}}(z))\right|\\
&\qquad \leq  \|r_l(\cdot,\cdot)\|_\infty\cdot\E_{s_g'\sim\mathfrak{D}}\sum_{z\in\mathcal{S}_l}|\E_{s'_{\Delta\cup\Delta'}\sim\mathcal{J}_{|\Delta\cup\Delta'|}(\cdot|s_g',s_g, a_g,s_{\Delta\cup\Delta'})}(F_{s'_\Delta}(z)-F_{s'_{\Delta'}}(z))|\\
&\qquad \leq 2\|r_l(\cdot,\cdot)\|_\infty\cdot \E_{s_g'\sim\mathfrak{D}} \mathrm{TV}(\E_{s'_{\Delta\cup\Delta'}|s_g'}F_{s'_\Delta}, \E_{s'_{\Delta\cup\Delta'}|s_g'}F_{s'_{\Delta'}}) \\
&\qquad \leq 2\|r_l(\cdot,\cdot)\|_\infty  \cdot  \mathrm{TV}(F_{s_\Delta}, F_{s_{\Delta'}}), 
\end{align*}
where we use Jensen's inequality and triangle inequality to pull $\E_{s_g'}\sum_{z\in\mathcal{S}_l}$ from the absolute value, and then use Cauchy-Schwarz to factor $\|r_l(\cdot,\cdot)\|_\infty$. The second inequality follows from \cref{lemma: generalized tvd linear bound}.  We now assume that for $T\leq t'\in\N$, for all joint stochastic kernels $\mathcal{J}_k$ and $\mathcal{J}_{k'}$, and for all $a_g'\in \mathcal{A}_g$:
\begin{align*}
|\E_{(s_g',s_\Delta')\sim\mathcal{J}_k(\cdot,\cdot|s_g,a_g,s_\Delta)}\hat{Q}_k^T(s'_g,F_{s'_\Delta},a_g') &- \E_{(s_g',s_{\Delta'}')\sim\mathcal{J}_{k'}(\cdot,\cdot|s_g,a_g,s_{\Delta'})}\hat{Q}_{k'}^T(s'_g, F_{s'_{\Delta'}}, a_g')| \\
&\leq \left(\sum_{t=0}^{T-1} 2\gamma^t\right)\|r_l(\cdot,\cdot)\|_\infty \cdot \mathrm{TV}(F_{s_\Delta}, F_{s_{\Delta'}}).\\
\end{align*}
For the remainder of the proof, we adopt the shorthand $\E_{(s_g',s_\Delta')\sim{\mathcal{J}}}$ to denote $\E_{(s_g',s_\Delta')\sim\mathcal{J}_{|\Delta|}(\cdot,\cdot|s_g,a_g,s_\Delta)}$, and $\E_{(s_g'',s_\Delta'')\sim{\mathcal{J}}}$ to denote $\E_{(s_g'',s_\Delta'')\sim\mathcal{J}_{|\Delta|}(\cdot,\cdot|s_g',a_g',s_\Delta')}$. 

Then, inductively, we have:
\begin{align*}
&|\E_{(s_g',s_\Delta')\sim{\mathcal{J}}} \hat{Q}_k^{T+1}(s'_g,F_{s'_\Delta},a_g') -\E_{(s_g',s_{\Delta'}')\sim{\mathcal{J}}}\hat{Q}_{k'}^{T+1}(s'_g, F_{s'_{\Delta'}}, a_g')|  \\
&\quad\quad=|\E_{(s_g',s_{\Delta\cup\Delta'}')\sim\mathcal{J}}[\bar{r}^{\pi_\ell}(s_g',s'_\Delta,a_g')-\bar{r}^{\pi_\ell}(s_g',s'_{\Delta'},a_g')  \\ &\quad\quad \quad+\gamma\E_{(s_g'',s_{\Delta\cup\Delta'}'')\sim{\mathcal{J}}}[\max_{a_g''\in\mathcal{A}_g}  \hat{Q}_k^T(s_g'',F_{s_\Delta''},a_g'')-\max_{a_g''\in\mathcal{A}_g} \hat{Q}_{k'}^T(s_g'',F_{s_{\Delta'}''}, a_g'')]]| \\
&\quad\quad\leq 2\|r_l(\cdot,\cdot)\|_\infty \cdot \mathrm{TV}(F_{s_\Delta}, F_{s_{\Delta'}})\\ &\quad\quad\quad + \gamma|\E_{(s_g',s_{\Delta\cup\Delta'}')\sim\mathcal{J}}[\E_{(s_g'',s_{\Delta\cup\Delta'}'')\sim{\mathcal{J}}}[\max_{a_g''\in\mathcal{A}_g} \hat{Q}_k^T(s_g'',F_{s_\Delta''},a_g'') -\max_{a_g''\in\mathcal{A}_g} \hat{Q}_{k'}^T(s_g'',F_{s_{\Delta'}''}, a_g'')]]|.\\
\end{align*}

\noindent Here, we expand out $\hat{Q}_k^{T+1}$ and $\hat{Q}_{k'}^{T+1}$ using the adapted Bellman operator. In the ensuing inequality, we apply the triangle inequality and bound the first term using the base case. 

Then, note that for some stochastic function $\mathcal{J}_{|\Delta\cup\Delta'|}'$, \begin{align*}
&\E_{(s_g',s_{\Delta\cup\Delta'}')\sim\mathcal{J}(\cdot,\cdot|s_g, a_g, s_{\Delta\cup\Delta'})} \E_{(s_g'',s_{\Delta\cup\Delta'}'')\sim\mathcal{J}(\cdot,\cdot|s_g', a_g', s'_{\Delta\cup\Delta'})}\max_{a_g''\in\mathcal{A}_g} \hat{Q}_k^T(s_g'',F_{s_\Delta''},a_g'') \\
&\quad\quad= \E_{(s_g'', s_{\Delta\cup\Delta'}'')\sim\mathcal{J}_{|\Delta\cup\Delta'|}'(\cdot,\cdot|s_g, a_g, s_{\Delta\cup\Delta'})}\max_{a_g''\in\mathcal{A}_g}\hat{Q}_k^T(s_g'', F_{s_\Delta''}, a_g''),
\end{align*}
where $\mathcal{J}'$ is an implicit function of $a_g'$ which is fixed at the start. 
As a shorthand, denote $\E_{(s_g'',s_{\Delta\cup\Delta'}'')\sim\mathcal{J}_{|\Delta\cup\Delta'|}'(\cdot,\cdot|s_g, a_g, s_{\Delta\cup\Delta'})}$ by $\E_{(s_g'',s_{\Delta\cup\Delta'}'')\sim\mathcal{J}'}$. 

Therefore,\looseness=-1
\begin{align*}
&|\E_{(s_g',s_\Delta')\sim{\mathcal{J}}} \hat{Q}_k^{T+1}(s'_g,F_{s'_\Delta},a_g') -\E_{(s_g',s_{\Delta'}')\sim{\mathcal{J}}}\hat{Q}_{k'}^{T+1}(s'_g, F_{s'_{\Delta'}}, a_g')| \\
&\quad\quad\leq 2\|r_l(\cdot,\cdot)\|_\infty \cdot \mathrm{TV}(F_{s_\Delta}, F_{s_{\Delta'}}) \\
&\quad \quad \quad\quad+ \gamma|\E_{(s_g'',s_{\Delta\cup\Delta'}'')\sim{\mathcal{J}'}}\max_{a_g''\in\mathcal{A}_g} \hat{Q}_k^T(s_g'',F_{s_\Delta''},a_g'')
 -\E_{(s_g'',s_{\Delta\cup\Delta'}'')\sim{\mathcal{J}'}}\max_{a_g''\in\mathcal{A}_g} \hat{Q}_{k'}^T(s_g'',F_{s_{\Delta'}''}, a_g'')| \\
&\quad\quad \leq 2\|r_l(\cdot,\cdot)\|_\infty\cdot \mathrm{TV}(F_{s_\Delta}, F_{s_{\Delta'}}) \\
&\quad\quad\quad\quad + \gamma\max_{a_g''\in\mathcal{A}_g}|\E_{(s_g'',s_{\Delta\cup\Delta'}'')\sim{\mathcal{J}'}} \hat{Q}^T_k(s_g'',F_{s_\Delta''},a_g'') 
 - \E_{(s_g'',s_{\Delta\cup\Delta'}'')\sim{\mathcal{J}'}} \hat{Q}^T_{k'}(s_g'',F_{s_{\Delta'}''},a_g'')| \\ 
&\quad\quad\leq 2\|r_l(\cdot,\cdot)\|_\infty\cdot\mathrm{TV}(F_{s_\Delta}, F_{s_{\Delta'}})  + \gamma\left(\sum_{t=0}^{T-1}2\gamma^t\right)\|r_l(\cdot,\cdot)\|_\infty\cdot \mathrm{TV}(F_{s_\Delta}, F_{s_{\Delta'}}) \\
&\quad\quad=  \left(\sum_{t=0}^{T}2\gamma^t\right)\|r_l(\cdot,\cdot)\|_\infty \cdot \mathrm{TV}(F_{s_\Delta},F_{s_{\Delta'}}).\\
\end{align*}

The second inequality uses \cref{lemma: expectation_expectation_max_swap} by setting the joint stochastic kernel to be $\mathcal{J}_{|\Delta\cup\Delta'|}'$, and using the induction assumption on the transition kernels $\mathcal{J}'_k$ and $\mathcal{J}'_{k'}$, proving the lemma.\qedhere 
\end{proof}

\begin{lemma}\label{lemma: generalized tvd linear bound} Given a joint transition probability $\mathcal{J}_{|\Delta\cup\Delta'|}$ as defined in \cref{definition: joint_transition_probability}, 
\[\mathrm{TV}(\E_{s'_{\Delta\cup\Delta'}\sim\mathcal{J}_{|\Delta\cup\Delta'|}(\cdot|s_g', s_g, a_g, s_{\Delta\cup\Delta'})} F_{s'_{\Delta}}, \E_{s'_{\Delta\cup\Delta'}\sim\mathcal{J}_{|\Delta\cup\Delta'|}(\cdot|s_g', s_g, a_g, s_{\Delta\cup\Delta'})} F_{s'_{\Delta'}})\leq \mathrm{TV}(F_{s_\Delta}, F_{s_{\Delta'}}).\]
\end{lemma}
\begin{proof}
Note that from \cref{lemma: expected next empirical distribution linearity}:
\begin{align*}\E_{s_{\Delta\cup\Delta'}'\sim\mathcal{J}_{|\Delta\cup\Delta'|}(\cdot,\cdot|s_g', s_g, a_g, s_{\Delta\cup\Delta'})}F_{s'_\Delta} 
&= \E_{s_{\Delta}'\sim\mathcal{J}_{|\Delta|}(\cdot,\cdot|s_g', s_g, a_g, s_{\Delta})} F_{s'_\Delta} \\ &= \mathcal{J}_1(\cdot|s_g(t+1),s_g(t),a_g(t),\cdot)F_{s_\Delta}.
\end{align*}

\noindent Then, by expanding the TV distance in $\ell_1$-norm:
    \begin{align*}
        &\mathrm{TV}(\E_{s'_{\Delta\cup\Delta'}\sim\mathcal{J}_{|\Delta\cup\Delta'|}(\cdot|s_g', s_g, a_g, s_{\Delta\cup\Delta'})} F_{s'_{\Delta}}, \E_{s'_{\Delta\cup\Delta'}\sim\mathcal{J}_{|\Delta\cup\Delta'|}(\cdot|s_g', s_g, a_g, s_{\Delta\cup\Delta'})} F_{s'_{\Delta'}}) \\
        &\quad\quad\quad\quad\quad\quad\quad= \frac{1}{2}\|\mathcal{J}_1(\cdot|s_g(t+1),s_g(t),a_g(t),\cdot)F_{s_\Delta} - \mathcal{J}_1(\cdot|s_g(t+1),s_g(t),a_g(t),\cdot)F_{s_{\Delta'}}\|_1 \\
        &\quad\quad\quad\quad\quad\quad\quad\leq \|\mathcal{J}_1(\cdot|s_g(t+1),s_g(t),a_g(t),\cdot)\|_1 \cdot \frac{1}{2}\|F_{s_\Delta}  -  F_{s_{\Delta'}}\|_1 \\
        &\quad\quad\quad\quad\quad\quad\quad \leq \frac{1}{2}\|F_{s_\Delta}  -  F_{s_{\Delta'}}\|_1 = \mathrm{TV}(F_{s_\Delta},F_{s_{\Delta'}}).
    \end{align*}
In the first inequality, we factorize $\|\mathcal{J}_1(\cdot|s_g(t+1),s_g(t), a_g(t))\|_1$ from the $\ell_1$-normed expression by the sub-multiplicativity of the matrix norm. Finally, since $\mathcal{J}_1$ is a column-stochastic matrix, we bound its norm by $1$ to recover the total variation distance between $F_{s_\Delta}$ and $F_{s_{\Delta'}}$. \qedhere
\end{proof}

\begin{lemma}\label{lemma: expected next empirical distribution linearity}Given the joint transition probability $\mathcal{J}_k$ from  \cref{definition: joint_transition_probability}, we have
\[\E_{s_{\Delta\cup\Delta'}(t+1)\sim\mathcal{J}_{|\Delta\cup\Delta'|}(\cdot|s_g(t+1), s_g(t), a_g(t), s_{\Delta\cup\Delta'}(t))} F_{s_\Delta (t+1)} :=  \mathcal{J}_{1}(\cdot|s_g(t+1),s_g(t), a_g(t), \cdot) F_{s_\Delta}(t).\]
\end{lemma}
\begin{proof}
First, observe that for all $x\in\mathcal{S}_l$:
\begin{align*}
&\E_{s_{\Delta\cup\Delta'}(t+1)\sim\mathcal{J}_{|\Delta\cup\Delta'|}(\cdot|s_g(t+1), s_g(t), a_g(t), s_{\Delta\cup\Delta'}(t))}F_{s_\Delta (t+1)}(x)  \\
&\quad\quad\quad\quad\quad\quad=\frac{1}{|\Delta|}\sum_{i\in\Delta}\E_{s_{\Delta\cup\Delta'}(t+1)\sim\mathcal{J}_{|\Delta\cup\Delta'|}(\cdot|s_g(t+1), s_g(t), a_g(t), s_{\Delta\cup\Delta'}(t))}\mathbbm{1}(s_i({t+1})=x)\\
    &\quad\quad\quad\quad\quad\quad= \frac{1}{|\Delta|}\sum_{i\in\Delta}\Pr[s_i(t+1) = x | s_g(t+1), s_g(t), a_g(t), s_{\Delta\cup\Delta'}(t))] \\
    &\quad\quad\quad\quad\quad\quad= \frac{1}{|\Delta|}\sum_{i\in\Delta}\Pr[s_i(t+1) = x | s_g(t+1), s_g(t), a_g(t), s_{i}(t))] \\
    &\quad\quad\quad\quad\quad\quad= \frac{1}{|\Delta|}\sum_{i\in\Delta}\mathcal{J}_{1}(x|s_g(t+1), s_g(t), a_g(t), s_{i}(t)),
\end{align*}
where in the first line, we expand on the definition of $F_{s_\Delta(t+1)}(x)$. Finally, we note that $s_i(t+1)$ is conditionally independent to $s_{\Delta\cup\Delta'\setminus i}$, which yields the equality above. 

Then, aggregating across every entry $x\in\mathcal{S}_l$, we have
\begin{align*}
    &\E_{s_{\Delta\cup\Delta'}(t+1)\sim\mathcal{J}_{|\Delta\cup\Delta'|}(\cdot|s_g(t+1), s_g(t), a_g(t), s_{\Delta\cup\Delta'}(t))} F_{s_\Delta (t+1)} \\ &\quad\quad\quad\quad  = \frac{1}{|\Delta|}\sum_{i\in\Delta}\mathcal{J}_{1}(\cdot|s_g(t+1),s_g(t), a_g(t), \cdot) \vec{\mathbbm{1}}_{s_i(t)} \\
    &\quad\quad \quad\quad = \mathcal{J}_{1}(\cdot|s_g(t+1),s_g(t), a_g(t), \cdot) F_{s_\Delta}.
\end{align*}
Notably, every $x$ corresponds to a choice of rows in $\mathcal{J}_{1}(\cdot|s_g(t+1),s_g(t), a_g(t), \cdot)$ and every choice of $s_i(t)$ corresponds to a choice of columns in $\mathcal{J}_{1}(\cdot|s_g(t+1),s_g(t), a_g(t), \cdot)$, making $\mathcal{J}_{1}(\cdot|s_g(t+1),s_g(t), a_g(t), \cdot)$ column-stochastic.\qedhere\\
\end{proof}

\begin{lemma} \label{lemma: tvd linear bound}The total variation distance between the expected empirical distribution of $s_\Delta({t+1})$ and $s_{\Delta'}({t+1})$ is linearly bounded by the total variation distance of the empirical distributions of $s_\Delta({t})$ and $s_{\Delta'}({t})$, for $\Delta,\Delta'\subseteq[n]$:
\begin{align*}\mathrm{TV}\bigg(\E_{\substack{s_i({t+1})\sim \bar{P}^{\pi_\ell}_l(\cdot| s_i(t),s_g(t)), \\ \forall i\in\Delta}}F_{s_\Delta({t+1})}, &\E_{\substack{s_i({t+1})\sim \bar{P}^{\pi_\ell}_l(\cdot| s_i(t),s_g(t)),\\ \forall i\in\Delta'}}F_{s_{\Delta'}({t+1})}\bigg) \\ &\quad\quad\quad\leq \mathrm{TV}\left(F_{s_\Delta(t)}, F_{s_{\Delta'}(t)}\right).\end{align*}
\end{lemma}
\begin{proof}\label{lemma: generalized linear transition}
We expand the TV-distance in $\ell_1$-norm and use the result in \cref{lemma: linear_transition} that \[\E_{\substack{s_i({t+1})\sim \bar{P}^{\pi_\ell}_l(\cdot| s_i(t),s_g(t)) \\ \forall i\in\Delta}}F_{s_\Delta({t+1})} = \bar{P}^{\pi_\ell}_l(\cdot|s_g(t)) F_{s_\Delta(t)}.\] as follows:
    \begin{align*}
\mathrm{TV}\bigg(&\E_{\substack{s_i({t+1})\sim \bar{P}^{\pi_\ell}_l(\cdot| s_i(t),s_g(t)) \\ \forall i\in\Delta}}F_{s_\Delta({t+1})}, \E_{\substack{s_i({t+1})\sim \bar{P}^{\pi_\ell}_l(\cdot| s_i(t),s_g(t)) \\ \forall i\in\Delta'}}F_{s_{\Delta'}(t+1)}\bigg) \\
&= \frac{1}{2}\left\|\E_{\substack{s_i({t+1})\sim \bar{P}^{\pi_\ell}_l(\cdot| s_i(t),s_g(t)) \\ \forall i\in\Delta}}F_{s_\Delta({t+1})} - \E_{\substack{s_i({t+1})\sim \bar{P}^{\pi_\ell}_l(\cdot| s_i(t),s_g(t)) \\ \forall i\in\Delta'}}F_{s_{\Delta'}{(t+1)}}\right\|_1 \\
&= \frac{1}{2}\left\|\bar{P}^{\pi_\ell}_l(\cdot|\cdot,s_g(t))F_{s_\Delta({t})} - \bar{P}^{\pi_\ell}_l(\cdot|\cdot,s_g(t)) F_{s_{\Delta'}(t)}\right\|_1 \\
&\leq \|\bar{P}^{\pi_\ell}_l(\cdot|\cdot,s_g(t))\|_1 \cdot \frac{1}{2}|F_{s_\Delta(t)} - F_{s_{\Delta'}(t)}|_1 \\
&= \|\bar{P}^{\pi_\ell}_l(\cdot|\cdot,s_g(t))\|_1 \cdot \mathrm{TV}(F_{s_\Delta(t)}, F_{s_{\Delta'}(t)}).\\
    \end{align*}
In the last line, we recover the total variation distance from the $\ell_1$ norm. Finally, by the column stochasticity of $\bar{P}^{\pi_\ell}_l(\cdot|\cdot,s_g)$, we have that $\|\bar{P}^{\pi_\ell}_l(\cdot|\cdot,s_g)\|_1 \leq 1$, which then implies 
\begin{align*}\mathrm{TV}\bigg(\E_{\substack{s_i({t+1})\sim \bar{P}^{\pi_\ell}_l(\cdot| s_i(t),s_g(t)), \\ \forall i\in\Delta}}F_{s_\Delta({t+1})}, &\E_{\substack{s_i({t+1})\sim \bar{P}^{\pi_\ell}_l(\cdot| s_i(t),s_g(t)),\\ \forall i\in\Delta'}}F_{s_{\Delta'}({t+1})}\bigg) \\ &\quad\quad\quad\leq \mathrm{TV}\left(F_{s_\Delta(t)}, F_{s_{\Delta'}(t)}\right),\end{align*}
proving the lemma.\qedhere\\
\end{proof}

\begin{lemma}{ The absolute difference between the expected maximums between $\hat{Q}_k$ and $\hat{Q}_{k'}$ is bounded by the maximum of the absolute difference between $\hat{Q}_k$ and $\hat{Q}_{k'}$, where the expectations are taken over any joint distributions of states $\mathcal{J}$, and the maximums are taken over the actions.}
\label{lemma: expectation_expectation_max_swap}\begin{align*}|&\E_{(s_g',s_{\Delta\cup\Delta'}')\sim\mathcal{J}_{|\Delta\cup\Delta'|}(\cdot,\cdot|s_g,a_g,s_{\Delta\cup\Delta'})}[\max_{a_g'\in\mathcal{A}_g} \hat{Q}_k^T(s_g',F_{s_\Delta'},a_g') -\max_{a_g'\in\mathcal{A}_g} \hat{Q}_{k'}^T(s_g',F_{s_{\Delta'}'}, a_g')]| \\
&\quad\quad \leq \max_{a_g'\in\mathcal{A}_g}|\E_{(s_g',s_{\Delta\cup\Delta'}')\sim\mathcal{J}_{|\Delta\cup\Delta'|}(\cdot,\cdot|s_g,a_g,s_{\Delta\cup\Delta'})}[\hat{Q}_k^T(s_g',F_{s_\Delta'},a_g') - \hat{Q}_{k'}^T(s_g',F_{s_{\Delta'}'},a_g')]|.\end{align*}
\end{lemma}
\begin{proof} Let
$a_g^* := \arg\max_{a_g'\in\mathcal{A}_g}\hat{Q}_k^T(s_g', F_{s'_\Delta}, a_g')$ and $\tilde{a}_g^* := \arg\max_{a_g'\in\mathcal{A}_g}\hat{Q}_{k'}^T(s_g', F_{s'_{\Delta'}}, a_g')$.
For the remainder of this proof, we adopt the shorthand $\E_{s_g',s_{\Delta\cup\Delta'}'}$ to refer to $\E_{(s_g',s_{\Delta\cup\Delta'}')\sim \mathcal{J}_{|\Delta\cup\Delta'|}(\cdot,\cdot|s_g, a_g, s_{\Delta\cup\Delta'})}$.\\

Then, if $\E_{s_g',s_{\Delta\cup\Delta'}'} \max_{a_g'\in\mathcal{A}_g}\hat{Q}_k^T(s_g', F_{s'_\Delta}, a_g') - \E_{s_g',s_{\Delta\cup\Delta'}'}\max_{a_g'\in\mathcal{A}_g}\hat{Q}_{k'}^T(s_g', F_{s'_{\Delta'}}, a_g')>0$,
\begin{align*}
&|\E_{s_g',s_{\Delta\cup\Delta'}'} \max_{a_g'\in\mathcal{A}_g}\hat{Q}_k^T(s_g', F_{s'_\Delta}, a_g') - \E_{s_g',s_{\Delta\cup\Delta'}'}\max_{a_g'\in\mathcal{A}_g}\hat{Q}_{k'}^T(s_g', F_{s'_{\Delta'}}, a_g')| \\
&\quad\quad= \E_{s_g',s_{\Delta\cup\Delta'}'} \hat{Q}_k^T(s_g', F_{s'_\Delta}, a_g^*) -\E_{s_g',s_{\Delta\cup\Delta'}'} \hat{Q}_{k'}^T(s_g', F_{s'_{\Delta'}}, \tilde{a}_g^*) \\
&\quad\quad\leq \E_{s_g',s_{\Delta\cup\Delta'}'}\hat{Q}^T_k(s_g', F_{s'_\Delta}, a_g^*) - \E_{s_g',s_{\Delta\cup\Delta'}'} \hat{Q}_{k'}^T(s_g', F_{s'_{\Delta'}}, a_g^*) \\
&\quad\quad\leq \max_{a_g'\in\mathcal{A}_g}|\E_{s_g',s_{\Delta\cup\Delta'}'}\hat{Q}_k^T(s_g', F_{s_\Delta'}, a_g') - \E_{s_g',s_{\Delta\cup\Delta'}'}\hat{Q}_{k'}^T(s_g', F_{s'_{\Delta'}}, a_g')|.
\end{align*}
Similarly, if $\E_{s_g',s_{\Delta\cup\Delta'}'} \max_{a_g'\in\mathcal{A}_g}\hat{Q}_k^T(s_g', F_{s'_\Delta}, a_g') - \E_{s_g',s_{\Delta\cup\Delta'}'}\max_{a_g'\in\mathcal{A}_g}\hat{Q}_{k'}^T(s_g', F_{s'_{\Delta'}}, a_g')<0$, an analogous argument by replacing $a_g^*$ with $\tilde{a}_g^*$ yields an identical bound. \qedhere  \\
\end{proof}

\begin{lemma} \label{lemma: linear_transition}For all $t\in\N$ and $\Delta\subseteq[n]$,
    \[    \E_{\substack{s_i({t+1})\sim \bar{P}^{\pi_\ell}_l(\cdot|s_i(t), s_g(t))\\ \forall i\in\Delta}} [F_{s_\Delta({t+1})}] = \bar{P}^{\pi_\ell}_l(\cdot|\cdot,s_g(t))F_{s_\Delta(t)}.\]
\end{lemma}
\begin{proof} 
For all  $x\in\mathcal{S}_l$:
\begin{align*}
\E_{\substack{s_i({t+1})\sim \bar{P}^{\pi_\ell}_l(\cdot|s_i(t), s_g(t))\\ \forall i\in\Delta}} [F_{s_\Delta({t+1})}(x)] &:= \frac{1}{|\Delta|}\sum_{i\in\Delta}\E_{s_i({t+1})\sim \bar{P}^{\pi_\ell}_l(s_i(t),s_g(t))}[\mathbbm{1}(s_i({t+1})=x)] \\
&= \frac{1}{|\Delta|}\sum_{i\in\Delta}\Pr[s_i({t+1}) = x|s_i({t+1})\sim \bar{P}^{\pi_\ell}_l(\cdot|s_i(t),s_g(t))]  \\
&= \frac{1}{|\Delta|}\sum_{i\in\Delta}\bar{P}^{\pi_\ell}_l(x|s_i(t),s_g(t)) .
\end{align*}
The first line writes out the definition of $F_{s_\Delta({t+1})}(x)$ and uses the conditional independence in the evolutions of $\Delta\setminus i$ and $i$. The second line uses that for any random variable $X\in\mathcal{X}$, $\E_{X}\mathbbm{1}[X=x] = \Pr[X=x]$. In line 3, we observe that this probability can be written as an entry of the local transition matrix $\bar{P}^{\pi_\ell}_l$. 

Then, aggregating across every entry $x\in\mathcal{S}_l$ gives:
\begin{align*}
\E_{\substack{s_i({t+1})\sim \bar{P}^{\pi_\ell}_l(\cdot|s_i(t), s_g(t))\\ \forall i\in\Delta}}[F_{s_\Delta({t+1})}] &= \frac{1}{|\Delta|}\sum_{i\in\Delta}\bar{P}^{\pi_\ell}_l(\cdot|s_i(t),s_g(t)) \\
&= \frac{1}{|\Delta|}\sum_{i\in\Delta}\bar{P}^{\pi_\ell}_l(\cdot|\cdot,s_g(t))\vec{\mathbbm{1}}_{s_i(t)} \\
&=: \bar{P}^{\pi_\ell}_l(\cdot|\cdot,s_g(t)) F_{s_\Delta(t)}.
\end{align*}
Here, $\vec{\mathbbm{1}}_{s_i(t)}\in  \{0,1\}^{|\mathcal{S}_l|}$ such that $\vec{\mathbbm{1}}_{s_i(t)}$ is $1$ at the index corresponding to $s_i(t)$, and is $0$ everywhere else. The last equality follows as $\bar{P}^{\pi_\ell}_l(\cdot|\cdot,s_g(t))$ is  column-stochastic, yielding that $\bar{P}^{\pi_\ell}_l(\cdot|\cdot,s_g(t))\vec{\mathbbm{1}}_{s_i(t)}=\bar{P}^{\pi_\ell}_l(\cdot|s_i(t),s_g(t)) $.\qedhere
\end{proof}
 
\begin{lemma}\label{lemma: combining transition probabilities}
For any joint transition probability function on $s_g, s_\Delta$, where $|\Delta|=k$, given by $\mathcal{J}_k:\mathcal{S}_g\times\mathcal{S}_l^{|\Delta|}\times\mathcal{S}_g\times\mathcal{A}_g\times\mathcal{S}_l^{|\Delta|}\to [0,1]$, we have:
\begin{align*}
\E_{(s_g',s_\Delta')\sim\mathcal{J}_k(\cdot,\cdot|s_g,a_g,s_\Delta)}&\left[\E_{(s_g'',s_\Delta'')\sim\mathcal{J}_k(\cdot,\cdot|s_g',a_g,s_\Delta')} \max_{a_g''\in\mathcal{A}_g}\hat{Q}_k^T(s_g'',F_{s_\Delta''},a_g'')\right]  \\
&\quad\quad\quad\quad = \E_{(s_g'',s_\Delta'')\sim\mathcal{J}_k^2(\cdot,\cdot|s_g,a_g,s_\Delta)} \max_{a_g''\in\mathcal{A}_g}\hat{Q}_{k}^T(s_g'',F_{s_\Delta''},a_g'').
\end{align*}
\end{lemma}
\begin{proof}
We start by expanding the expectations:
\begin{align*}
&\E_{(s_g',s_\Delta')\sim \mathcal{J}_k(\cdot,\cdot|s_g,a_g,s_\Delta)}\left[\E_{(s_g'',s_\Delta'')\sim \mathcal{J}_k(\cdot,\cdot|s_g',a_g,s_\Delta')}\max_{a_g'\in\mathcal{A}_g}\hat{Q}_k^T(s_g'',F_{s_\Delta''},a_g')\right] \\
&\!=\! \!\!\!  \sum_{(s_g',s_\Delta')\in\mathcal{S}_g\times\mathcal{S}_l^{|\Delta|}} \!\sum_{(s_g'',s_\Delta'')\in\mathcal{S}_g\times\mathcal{S}_l^{|\Delta|}} \!\!    \mathcal{J}_k[s_g',s_\Delta', s_g,a_g,s_\Delta]\mathcal{J}_k[s_g'',s_\Delta'', s_g',a_g,s_\Delta']\max_{a_g'\in\mathcal{A}_g}\hat{Q}_k^T(s_g'',F_{s_\Delta''},a_g') \\
&\!=\! \sum_{(s_g'',s_\Delta'')\in\mathcal{S}_g\times\mathcal{S}_l^{|\Delta|}}\mathcal{J}_k^2[s_g'',s_\Delta'', s_g, a_g, s_\Delta] \max_{a_g'\in\mathcal{A}_g}\hat{Q}^T_k(s_g'',F_{s_\Delta''},a_g') \\
&= \E_{(s_g'',s_\Delta'')\sim\mathcal{J}_k^2(\cdot,\cdot|s_g,a_g,s_\Delta)}\max_{a_g'\in\mathcal{A}_g}\hat{Q}_k^T(s_g'',F_{s_\Delta''},a_g')
\end{align*}
\noindent The right-stochasticity of $\mathcal{J}_k$ implies the right-stochasticity of $\mathcal{J}_k^2$.

Further, observe that $\mathcal{J}_k[s_g',s_\Delta',s_g,a_g,s_\Delta]\mathcal{J}_k[s_g'',s_\Delta'',s_g',a_g,s_\Delta']$ denotes the probability of the transitions $(s_g,s_\Delta)\to (s_g',s_\Delta')\to (s_g'',s_\Delta'')$ with actions $a_g$ at each step, where the joint state evolution is governed by $\mathcal{J}_k$. Thus, $\sum_{(s_g',s_\Delta')\in\mathcal{S}_g\times\mathcal{S}_l^{|\Delta|}}\mathcal{J}_k[s_g',s_\Delta',s_g,a_g,s_\Delta]\mathcal{J}_k[s_g'',s_\Delta'',s_g', a_g, s_g']$ is the stochastic probability function corresponding to the two-step evolution of the joint states from $(s_g,s_\Delta)$ to $(s_g'',s_\Delta'')$ under the action $a_g$, which is equivalent to $\mathcal{J}_k^2[s_g'',s_\Delta'',s_g,a_g,s_\Delta]$. 
The third equality lets the joint probabilities to be taken over $\mathcal{J}_k^2$, proving the lemma. \qedhere
\end{proof}

\subsection{Bounding the Total Variation Distance}

\begin{theorem}[Theorem C.5 of \cite{anand2024efficientreinforcementlearningglobal}]\label{theorem: sampling without replacement analog of DKW}
Consider a finite population $\mathcal{X}=(x_1,\dots,x_n)\in \mathcal{S}_l^n$. Let $\Delta\subseteq[n]$ be a random sample of size $k$ chosen uniformly and without replacement.  Then, for all $x\in\mathcal{S}_l$:
\[\Pr\left[\sup_{x\in\mathcal{S}_l}\left|\frac{1}{|\Delta|}\sum_{i\in\Delta}\mathbbm{1}{\{x_i = x\}} - \frac{1}{n}\sum_{i\in[n]}\mathbbm{1}{\{x_i = x\}}\right|<\epsilon\right] \geq 1 - 2|\mathcal{S}_l|e^{-\frac{2|\Delta|n\epsilon^2}{n-|\Delta|+1}}.\]
\end{theorem}

Combining the Lipschitz continuity bound  and the TV-distance bound yields \cref{theorem: Q-lipschitz of Fsdelta and Fsn}. \\

\begin{theorem}\label{theorem: Q-lipschitz of Fsdelta and Fsn}
    For all $s_g \in \mathcal{S}_g, s_1, \dots ,s_n \in \mathcal{S}_l^n, a_g \in \mathcal{A}_g$, we have that with probability at least $1 - \delta$:
\[|\hat{Q}_k^T(s_g, F_{s_\Delta}, a_g) - \hat{Q}_n^T(s_g, F_{s_{[n]}}, a_g)| \leq \frac{2\|r_l(\cdot,\cdot)\|_\infty}{1-\gamma}\sqrt{\frac{n-|\Delta|+1}{8n|\Delta|}\ln ( 2|\mathcal{S}_l|/\delta)}.\]

\begin{proof}
Let $\mathcal{X}=\mathcal{S}_l$ be the finite population in \cref{theorem: sampling without replacement analog of DKW} and
    recall the Lipschitz-continuity of $\hat{Q}_k^T$ from \cref{lemma: Q lipschitz continuous wrt Fsdelta}:
\begin{align*}\left|\hat{Q}_k^T(s_g,F_{s_\Delta},a_g) - \hat{Q}_n^T(s_g, F_{s_{[n]}}, a_g)\right| &\leq \left(\sum_{t=0}^{T-1} 2\gamma^t\right)\|r_l(\cdot,\cdot)\|_\infty \cdot \mathrm{TV}(F_{s_\Delta}, F_{s_{[n]}}) \\
&\leq \frac{2}{1-\gamma}\|r_l(\cdot,\cdot)\|_\infty\cdot\epsilon.
\end{align*}
By setting the error parameter in \cref{theorem: sampling without replacement analog of DKW} to $2\epsilon$, we find that the supremum is bounded with probability at least $1-2|\mathcal{S}_l|e^{-2|\Delta|n\epsilon^2/(n-|\Delta|+1)}$, and hence
\[\Pr\left[\left|\hat{Q}_k^T(s_g,F_{s_\Delta},a_g) - \hat{Q}_n^T(s_g, F_{s_{[n]}}, a_g)\right| \leq \frac{2\epsilon}{1-\gamma}\|r_l(\cdot,\cdot)\|_\infty\right]\geq 1 - 2|\mathcal{S}_l|e^{-\frac{8n|\Delta|\epsilon^2}{n-|\Delta|+1}}.\]
Finally, we parameterize the probability to $1-\delta$ to solve for $\epsilon$, which yields \[\epsilon = \sqrt{\frac{n-|\Delta|+1}{8n|\Delta|}\ln(2|\mathcal{S}_l|/\delta)}.\]
This proves the theorem.\qedhere \\
\end{proof}
\end{theorem}

\subsection{Using the Performance Difference Lemma to Bound the Optimality Gap}

Recall that the fixed-point of the empirical adapted Bellman operator $\hat{\mathcal{T}}_{k,m}$ is $\hat{Q}^\est_{k,m}$ and $\|\hat{Q}^*_k - \hat{Q}^\est_{k,m}\|_\infty \leq \epsilon_{k,m}$, where $\epsilon_{k,m}$ is the Bellman gap whose value we bound later. \\

\begin{lemma}\label{lemma:union_bound_over_finite_time}
Fix $s\in \mathcal{S}:=\mathcal{S}_g\times\mathcal{S}_l^n$. Suppose we are given a $T$-length sequence of i.i.d. random variables $\Delta_1, \dots, \Delta_T$, distributed uniformly over the support $\binom{[n]}{k}$. Further, suppose we are given a fixed sequence $\delta_1,\dots,\delta_T \in (0,1)$. Then, for each action $a_g\in\mathcal{A}_g$ and for $i\in[T]$, define events $B_i^{a_g}$ such that:
    \begin{align*}
        B_i^{a_g}\!\coloneqq\!\left\{\!\left|{Q}^*(s_g, s_{[n]}, a_g) \!-\!  \hat{Q}_{k,m}^{\est}(s_g, F_{s_{\Delta_i}}, a_g)\right| \!>\!  \sqrt{\frac{n-k+1}{8kn}\ln \frac{2|\mathcal{S}_l|}{\delta_i}}\cdot \frac{2}{1-\gamma}
    \|r_l(\cdot,\cdot)
    \|_\infty \!+\! \epsilon_{k,m}\!\right\}.
    \end{align*}
    Next, for $i\in[T]$, we define ``bad-events'' $B_i$ such that
    $B_i = \bigcup_{a_g\in\mathcal{A}_g} B_i^{a_g}$. Next, denote \[B \coloneqq \bigcup_{i=1}^T B_i = \bigcup_{i=1}^T \left\{\bigcup_{a_g \in \mathcal{A}_g} B_i^{a_g}\right\}.\]
    Then, the probability that  no ``bad event'' occurs is \[\Pr\left[\bar{B}\right] \geq 1- |\mathcal{A}_g|\sum_{i=1}^T \delta_i.\]
\end{lemma}
\begin{proof}
\begin{align*}
    \left|{Q}^*(s_g, s_{[n]}, a_g) - \hat{Q}_{k,m}^{\est}(s_g, F_{s_\Delta}, a_g)\right| &\leq \left|{Q}^*(s_g, s_{[n]}, a_g) - \hat{Q}_k^*(s_g, F_{s_\Delta}, a_g)\right| \\
    &\quad\quad + \left| \hat{Q}_k^*(s_g, F_{s_\Delta}, a_g) - \hat{Q}_{k,m}^{\est}(s_g, F_{s_\Delta}, a_g) \right| \\
    &\leq \left|{Q}^*(s_g, s_{[n]}, a_g) - \hat{Q}_k^*(s_g, F_{s_\Delta}, a_g)\right| + \epsilon_{k,m}.
\end{align*}
The first inequality above follows from the triangle inequality, and the second inequality uses $|{Q}^*(s_g, s_{[n]}, a_g) - \hat{Q}_k^*(s_g, F_{s_\Delta}, a_g)|\leq \|{Q}^*(s_g, s_{[n]}, a_g) - \hat{Q}_k^*(s_g, F_{s_\Delta}, a_g)\|_\infty \leq \epsilon_{k,m}$. Then, from \cref{theorem: Q-lipschitz of Fsdelta and Fsn}, we have that with probability at least $1-\delta_i$, 
\[\left|{Q}^*(s_g, s_{[n]}, a_g) - \hat{Q}_k^*(s_g, F_{s_\Delta}, a_g)\right| \leq \sqrt{\frac{n-k+1}{8nk} \ln\frac{2|\mathcal{S}_l|}{\delta_i}}\cdot \frac{2}{1-\gamma} \|r_l(\cdot,\cdot)\|_\infty.\]
So, event $B_i$ occurs with probability at most $\delta_i$. Thus, by repeated applications of the union bound, 
\begin{align*}
    \Pr[\bar{B}] &\geq 1 - \sum_{i=1}^T \sum_{a_g\in\mathcal{A}_g} \Pr[B_i^{a_g}] \geq 1 - |\mathcal{A}_g|\sum_{i=1}^T \Pr[B_i^{a_g}].
\end{align*}
Finally, substituting $\Pr[{B}_i^{a_g}] \leq \delta_i$ yields the lemma. \qedhere
\end{proof}

   Recall that for any $s\in\mathcal{S}:=\mathcal{S}_g\times\mathcal{S}_l^n\cong \mathcal{S}_g$, the policy function ${\pi}^\est_{k,m}(s)$ is defined as a uniformly random element in the maximal set of $\hat{\pi}_{k,m}^\est$ evaluated on all possible choices of $\Delta$. Formally:
\begin{equation}
{\pi}_{k,m}^\est(s) \sim\mathcal{U} \left\{\hat{\pi}^\est_{k,m}(s_g,F_{s_\Delta}): \Delta\in \binom{[n]}{k}\right\}.
\end{equation}

We now use the celebrated performance difference lemma from \cite{Kakade+Langford:2002}, restated below for convenience in \cref{theorem: performance difference lemma}, to bound the value functions generated between ${\pi}_{k,m}^\est$ and $\pi^*$.\\
 
\begin{theorem}[Performance Difference  \cite{Kakade+Langford:2002}]\label{theorem: performance difference lemma}{Given policies $\pi_1, \pi_2$, with corresponding value functions $V^{\pi_1}$, $V^{\pi_2}$:}
\begin{equation}
    V^{\pi_1}(s) - V^{\pi_2}(s) = \frac{1}{1-\gamma} \E_{\substack{s'\sim d^{\pi_1}_s \\ a_g'\sim \pi_1(\cdot|s')} } [A^{\pi_2}(s',a_g')].
\end{equation}
\emph{Here, $A^{\pi_2}(s',a_g'):= Q^{\pi_2}(s',a_g') - V^{\pi_2}(s')$ and $d_s^{\pi_1}(s') = (1-\gamma) \sum_{h=0}^\infty \gamma^h \Pr_h^{\pi_1}[s',s]$ where $\Pr_h^{\pi_1}[s',s]$ is the probability of $\pi_1$ reaching state $s'$ at time step $h$ starting from state $s$.}\\
\end{theorem}

\begin{theorem}[Bounding value difference] For any $s\in\mathcal{S}:=\mathcal{S}_g\times\mathcal{S}_l^n$ and $(\delta_1,\delta_2)\in(0,1]^2$,
\[V^{\pi^*}(s) - V^{{{\pi}}^{\est}_{k,m}}(s) \leq \frac{2\|r_l(\cdot,\cdot)\|_\infty}{(1-\gamma)^2}\sqrt{\frac{n-k+1}{2nk}}\sqrt{\ln \frac{2|\mathcal{S}_l|}{\delta_1}} +  \frac{2\tilde{r}}{(1-\gamma)^2}|\mathcal{A}_g|\delta_1 + \frac{2\epsilon_{k,m}}{1-\gamma}.\]  
\end{theorem}
\begin{proof}
Note that by definition of the advantage function,
\begin{align*}
\E_{a_g'\sim {{\pi}}_{k,m}^{\est}(\cdot|s')} A^{{\pi}^*}(s',a_g') &= \E_{a_g'\sim {{\pi}}_{k,m}^{\est}(\cdot|s')}[Q^{{\pi}^*}(s',a_g') - V^{{\pi}^*}(s')] \\
&= \E_{a_g'\sim {{\pi}}_{k,m}^\est(\cdot|s')}[Q^{{\pi}^*}(s',a_g') - \E_{a\sim {\pi}^*(\cdot|s')} Q^{{\pi}^*}(s',a_g)] \\
&= \E_{a_g'\sim {{\pi}}_{k,m}^\est(\cdot|s')}\E_{a_g\sim {\pi}^*(\cdot|s')}[Q^{{\pi}^*}(s',a_g') - Q^{{\pi}^*}(s',a_g)].
\end{align*}
Since ${\pi}^*$ is a deterministic policy, we can write:
\begin{align*}
\E_{a_g'\sim {{\pi}}_{k,m}^\est(\cdot|s')}\E_{a_g\sim {\pi}^*(\cdot|s')} A^{{\pi}^*}(s',a_g') &= \E_{a_g'\sim {{\pi}}_{k,m}^\est(\cdot|s')}[Q^{{\pi}^*}(s',a_g') - Q^{{\pi}^*}(s',{\pi}^*(s'))] \\
&= \frac{1}{\binom{n}{k}}\sum_{\Delta\in\binom{[n]}{k}}[Q^{{\pi}^*}(s',\hat{\pi}^\est_{k,m}(s_g',F_{s_\Delta'})) - Q^{{\pi}^*}(s',{\pi}^*(s'))].
\end{align*}
\noindent Then, by the linearity of expectations and the performance difference lemma (while noting that $Q^{\pi^*}(\cdot,\cdot) = Q^*(\cdot,\cdot)$):
\begin{align*} V^{{\pi}^*}(s) - V^{{{\pi}}_{k,m}^\est}(s) &= \frac{1}{1-\gamma}\sum_{\Delta\in\binom{[n]}{k}}\frac{1}{\binom{n}{k}}\E_{s'\sim d_s^{{{\pi}}_{k,m}^\est}}\left[Q^{{\pi}^*}(s',{\pi}^*(s')) - Q^{{\pi}^*}(s',{\hat{\pi}}_{k,m}^\est(s_g', F_{s_\Delta'}))\right] \\
&= \frac{1}{1-\gamma}\sum_{\Delta\in\binom{[n]}{k}}\frac{1}{\binom{n}{k}}\E_{s'\sim d_s^{{{\pi}}_{k,m}^\est}}\left[Q^*(s',{\pi}^*(s')) - Q^*(s',{\hat{\pi}}_{k,m}^\est(s'_g,F_{s_\Delta'}))\right].
\end{align*}
Next, we use \cref{lemma: expected_q*bound_with_different_actions} to bound this difference (where the probability distribution function of $\mathcal{D}$ is set as $ d_s^{{\pi}_{k,m}^\est}$ as defined in \cref{theorem: performance difference lemma}) while letting $\delta_1=\delta_2$:
\begin{align*}
V^{{\pi}^*}(s) &- V^{{{\pi}}_{k,m}^\est}(s)\\ &\leq \frac{1}{1-\gamma}\sum_{\Delta\in\binom{[n]}{k}}\frac{1}{\binom{n}{k}}\bigg[\frac{2\|r_l\|_\infty}{1-\gamma}\sqrt{\frac{n-k+1}{2nk}}\left(\sqrt{\ln \frac{2|\mathcal{S}_l|}{\delta_1}}\right)+\frac{2\tilde{r}}{1-\gamma}|\mathcal{A}_g|\delta_1 + 2\epsilon_{k,m}\bigg] \\
&\leq \frac{2\|r_l(\cdot,\cdot)\|_\infty}{(1-\gamma)^2}\sqrt{\frac{n-k+1}{2nk}}\left(\sqrt{\ln\frac{2|\mathcal{S}_l|}{\delta_1}}\right) +  \frac{2\tilde{r}}{(1-\gamma)^2}|\mathcal{A}_g|\delta_1+ \frac{2\epsilon_{k,m}}{1-\gamma},
\end{align*}
 proves the theorem. \qedhere \\
 \end{proof}

\begin{lemma}\label{lemma: expected_q*bound_with_different_actions}
For any arbitrary distribution $\mathcal{D}$ of states $\mathcal{S} := \mathcal{S}_g\times\mathcal{S}_l^n$, for any $\Delta\in\binom{[n]}{k}$ and for $\delta_1,\delta_2 \in (0,1]$, we have:
\begin{align*}\E_{s'\sim \mathcal{D}}&[Q^*(s',\pi^*(s')) - Q^*(s',\hat{\pi}^{\est}_{k,m}(s_g',F_{s_\Delta')})] \\
&\leq \frac{2\|r_l\|_\infty}{1-\gamma}\sqrt{\frac{n-k+1}{8nk}}\left(\sqrt{\ln \frac{2|\mathcal{S}_l|}{\delta_1}} + \sqrt{\ln \frac{2|\mathcal{S}_l|}{\delta_2}}\right) +  \frac{\tilde{r}}{1-\gamma}|\mathcal{A}_g|(\delta_1 + \delta_2) + 2\epsilon_{k,m}.
\end{align*}
\end{lemma}
\begin{proof}
Denote $\zeta_{k,m}^{s,\Delta}:=Q^*(s,\pi^*(s)) - Q^*(s, \hat{\pi}_{k,m}^\est(s_g,F_{s_\Delta})$.

We define the indicator function $\mathcal{I}:\mathcal{S}\times \mathbb{N} \times (0,1]\times(0,1]$ by:
\[\mathcal{I}(s, k, \delta_1,\delta_2)=\mathbbm{1}\left\{\zeta_{k,m}^{s,\Delta}\leq\frac{2\|r_l(\cdot,\cdot)\|_\infty}{1-\gamma}\sqrt{\frac{n-k+1}{8nk}}\left(\sqrt{\ln\frac{2|\mathcal{S}_l|}{\delta_1}}+ \sqrt{\ln\frac{2|\mathcal{S}_l|}{\delta_2}}\right)+ 2\epsilon_{k,m}\right\}.\]
We then study the expected difference between $Q^*(s',\pi^*(s'))$ and $Q^*(s',\hat{\pi}_{k,m}^{\est}(s'_g,F_{s'_\Delta}))$. Observe that:
\begin{align*}
\E_{s'\sim \mathcal{D}}[\zeta_{k,m}^{s,\Delta}] &= \E_{s'\sim \mathcal{D}}[Q^*(s',\pi^*(s')) - {Q}^*(s', \hat{\pi}_{k,m}^{\est}(s'_g,F_{s'_\Delta}))] \\
    &= \E_{s'\sim \mathcal{D}}\left[\mathcal{I}(s',k,\delta_1,\delta_2)(Q^*(s',\pi^*(s')) - Q^*(s',\hat{\pi}^{\est}_{k,m}(s_g',F_{s_\Delta'})))\right] \\
    &\quad\quad\quad + \E_{s'\sim \mathcal{D}}[(1-\mathcal{I}(s',k,\delta_1,\delta_2))(Q^*(s',\pi^*(s')) - Q^*(s',\hat{\pi}^{\est}_{k,m}(s_g',F_{s_\Delta'})))]
\end{align*}

Here, we have used the general property for a random variable $X$ and constant $c$ that $\E[X]=\E[X\mathbbm{1}\{X\leq c\}] + \E[(1-\mathbbm{1}\{X\leq c\})X]$. 
Then,
\begin{align*}
    &\E_{s'\sim \mathcal{D}}[Q^*(s',\pi^*(s')) -{Q}^*(s', \hat{\pi}_{k,m}^{\est}(s_g',F_{s_\Delta'})] \\
    & \leq \frac{2\|r_l\|_\infty}{1-\gamma}\sqrt{\frac{n-k+1}{8nk}}\left(\sqrt{\ln\frac{2|\mathcal{S}_l|}{\delta_1}} + \sqrt{\ln\frac{2|\mathcal{S}_l|}{\delta_2)}}\right) + 2\epsilon_{k,m} +  \frac{\tilde{r}}{1-\gamma}\left(1-\E_{s'\sim\mathcal{D}} \mathcal{I}(s',k,\delta_1,\delta_2)\right)\\
    &\leq \frac{2\|r_l\|_\infty}{1-\gamma}\sqrt{\frac{n-k+1}{8nk}}\left(\sqrt{\ln\frac{2|\mathcal{S}_l|}{\delta_1}} + \sqrt{\ln\frac{2|\mathcal{S}_l|}{\delta_2)}}\right) + 2\epsilon_{k,m}  +  \frac{\tilde{r}}{1-\gamma}|\mathcal{A}_g|(\delta_1 + \delta_2) .
\end{align*} The first inequality uses  $\E[X\mathbbm{1}\{X\leq c\}] \leq c$ and trivially bounds $Q^*(s',\pi^*(s')) - Q^*(s',\hat{\pi}_{k,m}^{\est}(s_g',F_{s_\Delta'}))$ by the maximum value of $Q^*$, which is $\frac{\tilde{r}}{1-\gamma}$. The second inequality uses that the expectation of an indicator function is the conditional probability of the underlying event and applies \cref{lemma: q_star_different_action_bounds}.\qedhere\\
\end{proof}

\begin{lemma}\label{lemma: q_star_different_action_bounds} For a fixed $s'\in\mathcal{S}:=\mathcal{S}_g\times\mathcal{S}_l^n$, for any $\Delta\in\binom{[n]}{k}$, and for $\delta_1,\delta_2\in (0,1]$, we have that with probability at least $1 - |\mathcal{A}_g|(\delta_1 + \delta_2)$:
    \[Q^*(s',\pi^*(s')) - Q^*(s',\hat{\pi}_{k,m}^{\est}(s_g',F_{s_\Delta'}))  \leq  \frac{2\|r_l\|_\infty}{1-\gamma} \sqrt{\frac{n-k+1}{8nk}}  \left(   \sqrt{\ln\frac{2|\mathcal{S}_l|}{\delta_1}} \!+\!  \sqrt{\ln\frac{2|\mathcal{S}_l|}{\delta_2}}\right) + 2\epsilon_{k,m}.\]
\end{lemma}
\begin{proof}
\begin{align*}
    Q^*(s',\pi^*(s'))&-Q^*(s',\hat{\pi}_{k,m}^{\est}(s_g',F_{s_\Delta'})) \\ &= Q^*(s',\pi^*(s')) - Q^*(s',\hat{\pi}_{k,m}^{\est}(s_g',F_{s_\Delta')})  \\ 
    &+ \hat{Q}^{\est}_{k,m}(s_g',s_\Delta',\pi^*(s')) -\hat{Q}^{\est}_{k,m}(s_g',s_\Delta',\pi^*(s'))\\
    & + \hat{Q}_{k,m}^{\est}(s_g',s_\Delta',\hat{\pi}_{k,m}^{\est}(s_g',F_{s_\Delta'})) - \hat{Q}_{k,m}^{\est}(s_g',F_{s_\Delta'},\hat{\pi}_{k,m}^{\est}(s_g',F_{s_\Delta'})).
\end{align*}
By the monotonicity of the absolute value and by the triangle inequality,
\begin{align*}
Q^*(s',\pi^*(s'))&-Q^*(s',\hat{\pi}_{k,m}^{\est}(s_g',F_{s_\Delta'})) \\
&\leq |Q^*(s',\pi^*(s'))-\hat{Q}_{k,m}^{\est}(s_g',F_{s_\Delta'},\pi^*(s'))| \\
&\quad\quad\quad + |\hat{Q}_{k,m}^{\est}(s_g',F_{s_\Delta'},\hat{\pi}_{k,m}^{\est}(s_g',F_{s_\Delta'}))-Q^*(s',\hat{\pi}_{k,m}^{\est}(s_g',F_{s_\Delta'}))|.
\end{align*}
The above inequality uses that the residual term \[\hat{Q}_{k,m}^{\est}(s_g',F_{s_\Delta'},\pi^*(s')) - \hat{Q}_{k,m}^{\est}(s_g',F_{s_\Delta'},\hat{\pi}_{k,m}^{\est}(s_g',F_{s_\Delta'})) \leq 0,\] since $\hat{\pi}^{\est}_{k,m}$ is the optimal greedy policy for $\hat{Q}_{k,m}^{\est}$. 

Applying the error bound from \cref{lemma:union_bound_over_finite_time} for two timesteps completes the proof. \qedhere 
\end{proof}

\begin{corollary}\emph{Optimizing parameters in \cref{theorem: application of PDL} yields}:
\label{corollary: performance_difference_lemma_applied}
\[V^{\pi^*}(s) - V^{{{\pi}}^\est_{k,m}}(s) \leq \frac{2\tilde{r}}{(1-\gamma)^2}\left(\sqrt{\frac{n-k+1}{2nk} \ln(2|\mathcal{S}_l||\mathcal{A}_g|\sqrt{k})} +  \frac{1}{\sqrt{k}}\right) + \frac{2\epsilon_{k,m}}{1-\gamma}.\]
\end{corollary}
\begin{proof} Recall from \cref{theorem: application of PDL} that:
    \[V^{\pi^*}(s) - V^{{{\pi}}^\est_{k,m}}(s) \leq \frac{2\|r_l(\cdot,\cdot)\|_\infty}{(1-\gamma)^2}\sqrt{\frac{n-k+1}{2nk}}\left(\sqrt{\ln \frac{2|\mathcal{S}_l|}{\delta_1}}\right) +  \frac{2\|r_l(\cdot,\cdot)\|_\infty}{(1-\gamma)^2}|\mathcal{A}_g|\delta_1+ \frac{2\epsilon_{k,m}}{1-\gamma}.\] 
    Note $\|r_l(\cdot,\cdot)\|_\infty \leq \tilde{r}$ by assumption. Then,
    \[V^{\pi^*}(s) - V^{{{\pi}}^\est_{k,m}}(s) \leq \frac{2\tilde{r}}{(1-\gamma)^2}\left(\sqrt{\frac{n-k+1}{2nk} \ln\frac{2|\mathcal{S}_l|}{\delta_1}} +  |\mathcal{A}_g|\delta_1\right) + \frac{2\epsilon_{k,m}}{1-\gamma}.\]
Finally, setting $\delta_1 = \frac{1}{k^{1/2}|\mathcal{A}_g|}$ yields the claim.\qedhere \\
\end{proof}

\subsection{Bounding the Bellman Error}

To bound the Bellman error $\epsilon_{k,m}$, we first recall Hoeffding's inequality \cite{409cf137-dbb5-3eb1-8cfe-0743c3dc925f}. 

\begin{lemma}
    [Hoeffding's inequality]
    Let $X_1, \dots, X_n$ be independent random variables such that $a_i \leq X_i \leq b_i$ almost surely. Let $S_n = \sum_{i=1}^n X_i$. Then, for all $t>0$, we have that
    \[\Pr[|S_n - \E[S_n]| \geq t] \leq 2\exp\left(-\frac{2t^2}{\sum_{i=1}^n (b_i-a_i)^2}\right).\]
\end{lemma}

 \begin{lemma}\label{lemma: yuejie bound}
    Fix $k\ge 1$ and let $\hat{\cT}_k$ and $\hat{\cT}_{k,m}$ be the adapted Bellman operator and adapted Bellman operator, respectively.
Let $\hat Q_k^*$ denote the unique fixed point of $\hat {\cT}_k$ and $\hat Q_{k,m}^*$
the unique fixed point of $\hat{\cT}_{k,m}$.
Let $N_k  \coloneqq  |\cS_g| |\cA_g| |\cS_l|^k$ be the size of the $Q$-function under the standard parameterization, and let $M_k \coloneqq |\cS_g| |\cA_g| |\cS_l| k^{|\cS_l|}$ be the size of the $Q$-function under the mean-field parameterization. Then for any $\rho\in(0,1)$, with probability at least $1-\rho$ over the sampling used to form $\hat{\cT}_{k,m}$, we have that in the standard parameterization,
\[\epsilon_{k,m}  \coloneqq  \|\hat Q_{k,m}^*-\hat Q_k^*\|_\infty
\leq \frac{\gamma\|r_l\|_\infty}{(1-\gamma)^2}\sqrt{\frac{2\ln(2N_k/\rho)}{m}},\]
and in the mean-field parameterization,
\[\epsilon_{k,m}  \coloneqq  \|\hat Q_{k,m}^*-\hat Q_k^*\|_\infty
\leq \frac{\gamma\|r_l\|_\infty}{(1-\gamma)^2}\sqrt{\frac{2\ln(2M_k/\rho)}{m}}.\]
\end{lemma}

\begin{proof}
    We first control the Bellman operator's deviation at the fixed point $\hat Q_k^*$.
For any $(s_g, s_\Delta, a_g)\in \cS_g\times \cS_l^k \times \cA_g$, define
\[Y \coloneqq  \max_{a_g\in\cA_g} \hat Q_k^*(s_g,s_\Delta,a_g).\] By Lemma \ref{lemma: Q-bound}, we have $\|\hat Q_k^*\|_\infty \leq \frac{\|r_l\|_\infty}{1-\gamma}$, hence $|Y|\le \frac{\|r_l\|_\infty}{1-\gamma}$ almost surely. The empirical operator uses i.i.d. samples $Y_1,\dots,Y_m$ of $Y$ and forms their average.
By Hoeffding's inequality,
\[\Pr\left[\Big|\frac{1}{m}\sum_{\ell=1}^m Y_\ell - \mathbb{E}[Y]\Big|\ge \eta\right]
\le 2\exp\Big(-\frac{m(1-\gamma)^2\eta^2}{2\|r_l\|_\infty^2}\Big).
\]
Taking a union bound over all $N_k$ tuples $(s_g, s_\Delta, a_g)$ gives
\[\Pr\left[\|\hat \cT_{k,m}\hat Q_k^*-\hat {\cT}_k \hat Q_k^*\|_\infty\ge \gamma\eta\right] \leq 2N_k \exp\Big(-\frac{m(1-\gamma)^2\eta^2}{2\|r_l\|^2}\Big).
\]
Setting the right-hand side to $\rho$ and solving for $\eta$, we have
\[\eta = \frac{\|r_l\|_\infty}{1-\gamma}\sqrt{\frac{2\ln(2N_k/\rho)}{m}}.\]
Thus, with probability at least $1-\rho$, we have
\[\|\hat {\cT}_{k,m}\hat Q_k^*-\hat {\cT}_k\hat Q_k^*\|_\infty
\le \gamma \frac{\|r_l\|_\infty}{1-\gamma}\sqrt{\frac{2\ln(2N_k/\rho)}{m}}. \]
Finally, using the contraction bound since $\hat Q_{k,m}^*=\hat T_{k,m}\hat Q_{k,m}^*$, we have
\begin{align*}\|\hat Q_{k,m}^*-\hat Q_k^*\|_\infty
&\leq \frac{1}{1-\gamma}\|\hat {\cT}_{k,m}\hat Q_k^*-\hat {\cT}_k\hat Q_k^*\|_\infty \\
&\leq \gamma \frac{\|r_l\|_\infty}{(1-\gamma)^2}\sqrt{\frac{2\ln(2N_k/\rho)}{m}},
\end{align*}
which yields the stated bound for the standard parameterization. Re-writing the $Q$-function with inputs $(s_g, F_{s_\Delta}, a_g)\in \cS_g\times \mu_k(\mathcal S_l)\times\cA_g$ proves the bound for the mean-field parameterization.\qedhere\\
\end{proof}

\begin{lemma}\label{lemma: epsilon_km_is_k} If $T= \frac{2}{1-\gamma}\log \frac{\tilde{r}\sqrt{k}}{1-\gamma}$, the global agent's $Q$-learning algorithm requires a sample complexity of \[\frac{2k^3\gamma^2\|r_l\|^2_\infty |\cS_g| |\cA_g|}{(1-\gamma)^5}\log \frac{\tilde{r}k^{3/2}}{1-\gamma} \cdot \ln(200e |\cS_g||\cA_g| |\cS_l|)\cdot \min\left\{|\cS_l|^2  k^{|\cS_l|} ,  k |\cS_l|^k \right\}\] to accrue a Bellman noise $\epsilon_{k,m}\leq \frac{2}{\sqrt{k}}$ with probability at least $1 - \frac{1}{100e^k}$ .\end{lemma}
\begin{proof}
We first prove that $\|\hat{Q}_k^T - \hat{Q}_k^*\|_\infty \leq \frac{1}{\sqrt{k}}$.  For this, it suffices to show $\gamma^T \frac{\tilde{r}}{1-\gamma} \leq \frac{1}{\sqrt{k}} \implies \gamma^T \leq \frac{1-\gamma}{\tilde{r}\sqrt{k}}$. Then, using $\gamma = 1 - (1 - \gamma) \leq e^{-(1-\gamma)}$, it again suffices to show $e^{-(1-\gamma)T}\leq \frac{1-\gamma}{\tilde{r}\sqrt{k}}$. Taking logarithms, we have
\begin{align*}
    \exp(-T(1-\gamma)) &\leq \frac{1-\gamma}{\tilde{r}\sqrt{k}} \\
    -T(1-\gamma) &\leq \log\frac{1-\gamma}{\tilde{r}\sqrt{k}} \\
    T&\geq \frac{1}{1-\gamma}\log\frac{\tilde{r}\sqrt{k}}{1-\gamma}.
\end{align*}
Since $T = \frac{2}{1-\gamma}\log\frac{\tilde{r}\sqrt{k}}{1-\gamma} > \frac{1}{1-\gamma}\log\frac{\tilde{r}\sqrt{k}}{1-\gamma}$, the condition holds and $\|\hat{Q}_k^T - \hat{Q}_k^*\|_\infty \leq \frac{1}{\sqrt{k}}$. \\

Then, rearranging \cref{lemma: yuejie bound} and incorporating the convergence error of the $\hat{Q}_k$-function, one has that with probability at least $1-\rho$, 
\begin{equation}\label{arranging_yuejie}
        \epsilon_{k,m} \leq \frac{1}{\sqrt{k}} + \gamma \frac{\|r_l\|_\infty}{(1-\gamma)^2}\sqrt{\frac{2\ln(2B_k/\rho)}{m}},
    \end{equation}
    where $B_k$ is the number of samples used in the parameterization of the $Q$-function. \\
    
    \textbf{In the mean-field parameterization,} we have $B_k \leq |\mathcal{S}_g||\mathcal{A}_g||\mathcal{S}_l| k^{|\mathcal{S}_l|}$. Substituting this in \cref{arranging_yuejie} yields that  with probability at least $1-\rho$,
    \begin{equation}
        \epsilon_{k,m} \leq \frac{1}{\sqrt{k}}+ \gamma \frac{\|r_l\|_\infty}{(1-\gamma)^2}\sqrt{\frac{2\ln(2|\cS_g||\cA_g||\cS_l| k^{|\cS_l|} /\rho)}{m}},
    \end{equation}
    
    \textbf{In the standard parameterization,} we have   $B_k \leq |\mathcal{S}_g||\mathcal{S}_l|^k |\mathcal{A}_g|$. Substituting this in \cref{arranging_yuejie} yields that with probability at least $1-\rho$,
    \begin{equation}
        \epsilon_{k,m} \leq \frac{1}{\sqrt{k}}+ \gamma \frac{\|r_l\|_\infty}{(1-\gamma)^2}\sqrt{\frac{2\ln(2|\cS_g||\cA_g||\cS_l|^k /\rho)}{m}}.
    \end{equation}

Therefore, to attain a Bellman error of $\epsilon_{k,m} \leq 2/\sqrt{k}$ with probability at least $1-\delta$, we set the number of samples $m$ accordingly.  For the mean-field parameterization, we set the number of samples to be
\begin{equation}
    m_1 = \frac{2k^2 \gamma^2\|r_l\|^2_\infty}{(1-\gamma)^4} \ln(200e^k |\cS_g||\cA_g| |\cS_l| k^{|\cS_l|}),
\end{equation}
and for the standard parameterization we set the number of samples to be
\[m_2 = \frac{2k^2 \gamma^2\|r_l\|^2_\infty}{(1-\gamma)^4} \ln(200e^k |\cS_g||\cA_g| |\cS_l|^k).\]
Therefore, the sample complexity of the global agent's learning algorithm is \begin{align*}\min&\bigg\{Tm_1|\mathcal{S}_g||\mathcal{S}_l||\mathcal{A}_g| k^{|\mathcal{S}_l| }, T m_2 |\cS_g||\cS_l|^k |\cA_g|\bigg\} \\
&=   \frac{2k^3\gamma^2\|r_l\|^2_\infty |\cS_g| |\cA_g|}{(1-\gamma)^5}\log \frac{\tilde{r}k^{3/2}}{1-\gamma} \cdot \ln(200e |\cS_g||\cA_g| |\cS_l|)\cdot \min\left\{|\cS_l|^2  k^{|\cS_l|} ,  k |\cS_l|^k \right\},
\end{align*}
which proves the claim.\qedhere\\
\end{proof}

\begin{corollary}[Bounding the optimality gap of the best-response of the global agent] \label{pdl result in otilde form}
    From \cref{corollary: performance_difference_lemma_applied}, we have that if in the mean-field parameterization \[ m \geq \frac{2k^2 \gamma^2\|r_l\|^2_\infty}{(1-\gamma)^4} \ln(200e^k |\cS_g||\cA_g| |\cS_l|^k),\] 
    or in the standard parameterization 
    \[m \geq  \frac{2k^2 \gamma^2\|r_l\|^2_\infty}{(1-\gamma)^4} \ln(200e^k |\cS_g||\cA_g| |\cS_l| k^{|\cS_l|}),\]
    then with probability at least $1 -\frac{1}{100e^k}$:
    \begin{align*}V^{\pi^*}(s) - V^{{{\pi}}^\est_{k,m}}(s) &\leq \frac{2\tilde{r}}{(1-\gamma)^2}\left(\sqrt{\frac{n-k+1}{2nk} \ln(2|\mathcal{S}_l||\mathcal{A}_g|\sqrt{k})} +  \frac{1}{\sqrt{k}}\right) + \frac{4}{\sqrt{k}(1-\gamma)} \\ &\leq \tilde{O}\left(\frac{1}{\sqrt{k}}\right).\\
    \end{align*}
\end{corollary}

\section{Optimality of the Local Agent Policy}
\label{sec: optimality of the local agent policy}
\subsection{Upper Confidence Fixed-Horizon Episodic Reinforcement Learning}

We begin by recalling the \textbf{U}pper-\textbf{C}onfidence \textbf{F}ixed-\textbf{H}orizon (UCFH) episodic reinforcement learning algorithm in \cite{dann2015sample}. We first formalize the notion of episodic fixed-horizon MDPs which can be formalized as a tuple $M=(\cS,\cA, r, p, p_0, H)$, where $\cS$ and $\cA$ are finite sets.

The learning agent interacts with the MDP in episodes of $H$ time steps. At time $t=1,\dots,H$, the agent observes state $s_t$, chooses action $a_t$ using time-varying policy $a_t = \pi_t(s_t)$ for $t=1,\dots, H$. The next state is sampled from the stationary transition kernel $s_{t+1}\sim p(\cdot|s_t,a_t)$ and the initial state from $s_1\sim p_0$. The agent receives a reward drawn from a distribution with mean $r(s_t)$ determined by the reward function, which is non-stationary. The quality of the policy is given by $R_M^\pi = \E[\sum_{t=1}^H r_t(s_t)]$, and the problem to solve is how many episodes does a learning agent follow a policy $\pi$ that is not $\epsilon$-optimal with probability at least $1-\delta$, for any chosen accuracy $\epsilon$ and failure probability $\delta$.

We state the algorithm in \cite{dann2015sample} below for self-containedness, and refer the interested reader for intuitive descriptions and explanations, detailed algorithm implementations, and proofs of correctness and sample complexity analyses of \texttt{UCFH} and \texttt{FixedHorizonEVI} to  \cite{dann2015sample}.\\

\begin{algorithm}[hbt]
\begin{algorithmic}    
\caption{\texttt{UCFH}: \textbf{U}pper-\textbf{C}onfidence \textbf{F}ixed-\textbf{H}orizon (UCFH) episodic RL algorithm from \cite{dann2015sample}}
        \REQUIRE Target accuracy $\epsilon\in(0,1]$, failure tolerance $\delta\in(0,1]$, and fixed horizon MDP $M$.
        \STATE Set $i\coloneqq 1, w_{min}\coloneqq \frac{\epsilon}{4H|S|}$, $U_{max}\coloneqq |\cS\times\cA|\log_2 \frac{|\cS|H}{w_{min}}$, and $\delta_1 =\frac{\delta}{20 U_{max}}$.
        \STATE Let $m = 512 (\log_2 \log_2 H)^2 \frac{10 H^2}{\epsilon^2}\log^2 \left(\frac{8H^2|\cS|^2}{\epsilon}\right) \ln \frac{60|\cS\times\cA|\log_2^2 (4|\cS|^2 H^2/\epsilon)}{\delta}$
        \STATE Let $n(s,a) = v(s,a) = n(s,a,s') \coloneqq 0, \forall s\in \cS, a\in \cA, s'\in \cS(s,a)$.
        \medskip
        \STATE \textbf{while do}         \STATE  \quad \textcolor{blue}{\texttt{/* Optimistic planning}}
        \STATE \quad $\hat{p}(s'|s,a) \coloneqq n(s,a,s')/n(s,a)$ for all $(s,a)$ with $n(s, a) > 0$ and $s' \in \cS(s,a);$
        \STATE \quad  Let $\cM_i\coloneqq \{\widetilde{M}\in \cM_{nonst.}: \forall (s,a)\in \cS\times\cA, t=1,\dots,H, s'\in \cS(s,a)$,
        \STATE \quad \hspace{5cm}$\tilde{p}_t(s'|s,a)\in\texttt{ConfidenceSet}(\hat{p}(s'|s,a), n(s,a))\}$
        \STATE \quad $\widetilde{M}_i, \pi^i \coloneqq \texttt{FixedHorizonEVI}(\cM_i)$
        \STATE \quad \textbf{repeat}\textcolor{blue}{\quad \texttt{/* Execute policy}}
        \STATE \quad \quad \texttt{SampleEpisode}$(\pi^i)$ \quad  \textcolor{blue}{\texttt{// from $M$ using $\pi^i$}}
        \STATE \quad \textbf{until}  {there is a $(s,a) \in \cS\times\cA$ with $v(s,a)\geq \max\{mw_{min}, n(s,a)\}$ and $n(s,a)< |\cS|mH$}
        \STATE  \quad \textcolor{blue}{\texttt{/* Update model statistics for one $(s,a)$-pair with condition above}}
        \STATE \quad  Let $n(s,a)\coloneqq n(s,a) + v(s,a)$
        \STATE \quad  Let $n(s,a,s') \coloneqq n(s,a,s') + v(s,a,s'), \forall s'\in \cS(s,a)$
        \STATE \quad  Let $v(s,a) \coloneqq v(s,a,s')\coloneqq 0, \forall s'\in \cS(s,a)$ 
        \STATE \quad Update $i\gets i+1$
        \medskip
        \STATE \textbf{Procedure }{\texttt{SampleEpisode}$(\pi)$}
        \STATE \quad $s_0\sim p_0$
        \STATE \quad \textbf{for} $t=0,\dots,H-1$ \textbf{do}
        \STATE \quad\quad  $a_t \coloneqq \pi_{t+1}(s_t)$ and $s_{t+1}\sim p(\cdot|s_t,a_t)$
        \STATE \quad\quad  $v(s_t, a_t) \coloneqq v(s_t, a_t)+1$ and $v(s_t, a_t, s_{t+1}) \coloneqq v(s_t, a_t, s_{t+1}) + 1$
    \STATE 
    \medskip
    \textbf{Function} \texttt{ConfidenceSet}$(p,n)$
    \STATE\quad $\cP\coloneqq \bigg\{p'\in[0,1]: \text{if }n>1: |\sqrt{p'(1-p')} - \sqrt{p(1-p)}| \leq \sqrt{\frac{2\ln (6/\delta_1)}{n-1}}$, 
    \STATE \quad\quad\quad \[|p-p'| \leq \min\left(\sqrt{\frac{1}{2n} \ln (6/\delta_1)}, \sqrt{\frac{2p(1-p)}{n}\ln(6/\delta_1)} + \frac{7}{3(n-1)} \ln \frac{6}{\delta_1} \right)\bigg\}\]
    \STATE \quad \textbf{return} $\mathcal{P}$
    \label{algorithm from brunskill and dann}
\end{algorithmic}
\end{algorithm}

\begin{algorithm}[hbt]\begin{algorithmic}\caption{\texttt{FixedHorizonEVI}$(\cM)$ \cite{dann2015sample}}
    \STATE Let $\tilde{Q}^\star_{H:H}(s,a) = r_H(s), \forall  s,a\in \cS\times\cA$
    \FOR{$t=H-1,\dots,1$}
    \STATE $\pi_{t+1}(s) \coloneqq \arg\max_{a\in cA} \tilde{Q}^\star_{t+1:H}(s,a), \forall s\in \cS$
    \STATE Sort states $s^{(1)},\dots,s^{(|\cS|)}$ such that
    $\tilde{Q}_{t+1:H}^\star(s^{(i)}, \pi_{t+1}(s^{(i)})) \geq \tilde{Q}_{t+1:H}^\star(s^{(i+1)}, \pi_{t+1}(s^{(i+1)}))$
    \medskip
\FOR{$(s,a)\in\cS\times\cA$}
    \STATE $\tilde{p}_t(s'|s,a)\coloneqq \min  $\texttt{ ConfidenceSet}$(\hat{p}(s'|s,a), n(s,a)), \quad \forall s'\in \cS(s,a)$
    \STATE Set $\Delta\coloneqq 1 - \sum_{s'\in \cS(s,a)}\tilde{p}_t(s'|s,a)$ and $i\coloneqq 1$
    \medskip
    \WHILE{$\Delta>0$}
    \STATE $s'\coloneqq s^{(i)}$
    \STATE $\Delta' \coloneqq \min\{\Delta, \max \texttt{ConfidenceSet}(\hat{p}(s'|s,a), n(s,a)) - \tilde{p}_t(s'|s,a)\}$
    \STATE $\tilde{p}_t(s'|s,a)\coloneqq \tilde{p}_t(s'|s,a) + \Delta'$
    \STATE Set $\Delta\coloneqq \Delta - \Delta'$ and $i' \coloneqq i+1$
    \ENDWHILE
    \STATE $\tilde{Q}^\star_{t:H}(s,a) = r_t(s) + \sum_{s'\in \cS(s,a)}\tilde{p}_t(s'|s,a)\tilde{Q}_{t+1:H}^\star(s', \pi_{t+1}(s'))$
    \ENDFOR
    \ENDFOR
    \medskip
    \STATE $\pi_1(s) \coloneqq \arg\max_{a\in\cA} \tilde{Q}_{1:H}^\star(s,a), \forall s\in \cS$
    \STATE \textbf{Return} MDP with transition probabilities $\tilde{p}_t$ and optimal policy $\pi$.
    \end{algorithmic}  
\end{algorithm}
 
\subsection{Black-Box Reduction of Local Agent Best-Responses to Upper Confidence Fixed-Horizon Episodic RL}
\label{subsection: reduction episodic}
We first show how to reduce from the setting in \cite{dann2015sample} which has state rewards to our setting in the induced MDP which has state/action rewards. Algorithm \ref{algorithm: reduce state-action reward MDP} below works by adding an intermediate ``deciding state'' where we get no reward but have to choose an action.

\begin{algorithm}[hbt!]
    \begin{algorithmic}
 \caption{Reduce State-Action Reward MDP to State-Reward MDP with Deciding States}
        \REQUIRE Input Markov Decision Process $\cM = (\cS, \cA, P, r, H)$
        \STATE Introduce $\perp \notin \cA$.
        \STATE Set $\tilde{\cS} \gets (\cS\times\cA) \cup (\cS\times \{\perp\})$ and $\tilde{\cA} \gets \cA\cup\{\perp\}$. 
        \STATE Let $\tilde{r}(s, a) \gets r(s, a)$ for $a\in \cA$ and $\tilde{r}(s, \perp) \gets 0$. 
        \STATE Define simulator \texttt{Step}$((s,z), u)$ that returns $(x',\rho')$ where $\rho=\tilde{r}(x')$
        \medskip
        \IF{$x=(s,a)$ for some $a\in\cA$} 
        \STATE \textcolor{blue}{\texttt{/* committed state}}
        \STATE Sample $s'\sim P(\cdot|s, a)$
        \STATE $x' \gets (s', \perp)$
        \STATE \textbf{Return} $(x', 0)$
        \medskip
        \ELSIF{$x=(s,\perp)$}
        \STATE \textcolor{blue}{\texttt{/*  deciding state}}
        \smallskip
        \IF{$u=\perp$}
        \STATE $x'\gets (s, \perp)$
        \STATE \textbf{Return} $(x', 0)$
        \smallskip
        \ELSE
        \STATE $x' \gets (s, u)$
        \STATE \textbf{Return} $(x', r(s, u))$
        \ENDIF
        \ENDIF
    \label{algorithm: reduce state-action reward MDP}
    \end{algorithmic}
\end{algorithm} We next show how to reduce our infinite horizon setting the finite horizon setting in \cite{dann2015sample}.

\begin{lemma}\label{truncation lemma}
    Suppose $\epsilon>0$ and $\gamma\in(0,1)$. Then, for $H\coloneqq \frac{1}{1-\gamma}\log \frac{\|r\|_\infty}{\epsilon(1-\gamma)}$, we have 
    \[0\leq \sum_{t=0}^\infty \gamma^t r_t - \sum_{t=0}^H \gamma^t r_t \leq \epsilon.\]
\end{lemma}
\begin{proof}
    Fix any $T\geq 0$. Note that
    \begin{align*}
        \sum_{t=0}^\infty \gamma^t r_t - \sum_{t=0}^T \gamma^t r_t &= \sum_{t=T+1}^\infty \gamma^t r_t \leq \frac{\gamma^{T+1}}{1-\gamma}\|r\|_\infty < \epsilon.
    \end{align*}
    Now, note that for $T\geq \frac{1}{1-\gamma}\log \frac{\|r\|_\infty}{\epsilon(1-\gamma)}$, we have
    $-T(1-\gamma)\leq \log \frac{\|r\|_\infty}{\epsilon(1-\gamma)}$. This implies
    \begin{align*}\gamma^T &= (1-(1-\gamma))^T \leq \exp(-T(1-\gamma)) \leq \frac{\epsilon(1-\gamma)}{\|r\|_\infty},
    \end{align*}
    which completes the proof.\qedhere\\
\end{proof}

\begin{corollary}
    Suppose $\cM =(\cS, \cA, P, r, \gamma)$ is an infinite-horizon discounted MDP. Let $\epsilon>0$ and $H\coloneqq \frac{1}{1-\gamma}\log \frac{\|r\|_\infty}{\epsilon(1-\gamma)}$. Then, for any policy $\pi:\cS\to\cA$, for all $s\in S$, define
    \[V_\infty^\pi(s)\coloneqq \E\left[\sum_{t\geq 0}\gamma^t r(s_t, \pi(s_t))\mid s_0=s\right]\] and \[V_H^\pi(s)\coloneqq \E\left[\sum_{t=0}^H\gamma^t r(s_t, \pi(s_t))\mid s_0=s\right].\]
Then, for any $\pi:\cS\to\cA$ and $s\in S$, we have
\[0\leq V_\infty^\pi(s) - V_H^\pi(s) \leq \epsilon.\]
\end{corollary}
 
Finally, we state the theoretical guarantee of \cite{dann2015sample} which gives us provable guarantees on the best-response nature of the learned policy. Specifically, we seek to show that the local agent policy update is an $\epsilon_\ell$-best-response, which will later be used to show that the joint policy dynamics converges to an approximate Nash equilibrium.\looseness=-1

\begin{theorem}
    [Theorem 1 of \cite{dann2015sample}] 
    \label{brunskill sample bound} For any $0 < \epsilon, \delta \leq 1$, we have that with  probability at least $1-\delta$, \texttt{UCFH} produces a policy $\pi_l^\star$ with \[\tilde{O}\left(\frac{H^2 |\cS\times\cA|}{\epsilon^2} \log \frac1\delta\right)\] episodes such that \[R^\star - R^{\pi_l^\star} = V_{1:H}^\star(s_0) - V_{1:H}^{\pi_l^\star}(s_0) \leq \epsilon.\] Importantly, sampling one episode and updating the respective $v$ variables has $O(H)$ runtime. This results in a total runtime for sampling of \[\tilde{O}\left(\frac{H^3 |\cS||\cA|}{\epsilon^2} \log \frac{1}{\delta}\right).\] Each update of the policy involves updating the $n$ variables and $\cM_k$, and calling \emph{\texttt{FixedHorizonEVI}} with runtime cost $O(H|\cS||\cA| + H|\cS|\log |\cS|)$. Therefore, the total runtime for policy updates is \begin{align*}\tilde{O}\left(H|\cS|^2 |\cA|^2 \log \frac1\epsilon\right).
    \end{align*}
\end{theorem}

\subsection{Formulating the Local Agent's Best-Response MDP} 

Our formulation of the best-response MDP works by serializing each macro-step into $k$ micro-steps so the action space remains 
$\cA_l$. We provide a reduction for the MDP in Algorithm \ref{algorithm: k chained induced MDP} and the mean-field MDP in Algorithm \ref{algorithm: S chained induced MDP}, while constraining the RL solver's policy class so that the selected action is a function only of $s_j$ and $s_g$

\begin{figure}[hbt!]
    \centering
    \includegraphics[width=1\linewidth]{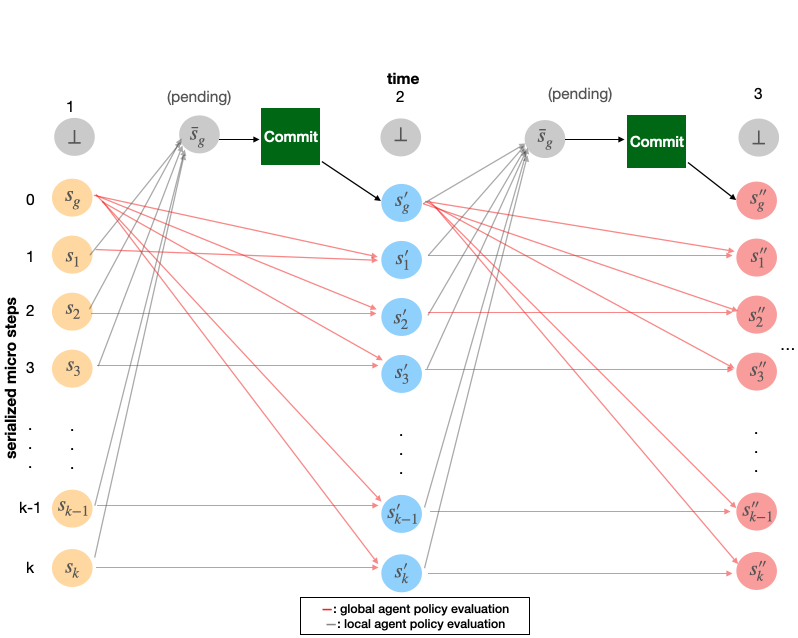}
    \caption{Construction of the $k$-chained MDP from \cref{algorithm: k chained induced MDP}. In this construction, along the chain of micro nodes, only one local replica is ``active'' at a time. At stage $j$, the learner selects a single action $a_\ell\in\cA_l$, interpreted as the action for replica $j$. The environment updates only that replica via $s_j' \sim P_l(\cdot|s_j, s_g, a_\ell)$ leaving the other $k-1$ local replica states unchanged. Next, the stage pointer increments $j\gets j+1$. At the end of the chain (stage $j=k$), we commit the cached global transition by setting $s_g\gets \bar{s}_g$ and resetting $\bar{S}_g \gets \perp$, completing one macro step. 
}
    \label{fig:k-chained-mdp}
\end{figure}

\textbf{Algorithm \ref{algorithm: k chained induced MDP}} ($k$-chained MDP for learning the local agent's best response to policy $\pi_g:\cS_g\times\cS_l^k\to\cA_g$). We begin by recapping the original (macro) process, where at macro time $t$ the environment is in a pair $(s_{g,t}, s_{i,t})$ and the global policy is queried on $(s_{g,t}, s_{1:k, t})$ and then each local agent transitions according to $P_l(\cdot|s_{i,t}, s_{g,t}, a_{i,t})$, while the global agent's state transitions according to $P_g(\cdot|s_{g,t}, a_{g,t})$. The difficulty here is that the global agent's policy $\pi_g$ requires a $k$-tuple of local states whereas a single local agent only sees $s_{g,t}$ and $s_{i,t}$. Our reduction resolves this by explicitly maintaining a $k$-tuple of local states and serializing the $k$ local decisions within each macro step.

We begin by replacing each macro time step $t$ by a chain of $k$ micro steps $\tau = tk, tk+1, \dots, tk+(k-1)$, indexed by a stage pointer $j\in\{1,\dots,k\}$. The state of the unfolded system at micro step $\tau$ is given by $(j, s_g, \bar{s}_g, s_{1:k})$, where $s_{1:k}$ are the local ``replica'' states, $s_g$ is the current global state, and $\bar{s}_g$ is a pending next global state. At the first micro node of each chain (where the stage $j=1$ and $\bar{s}_g = \perp$), we evaluate the global agent's policy using the entire replica tuple and compute $a_g\sim \pi_g(\cdot|s_g, s_{1:k})$, and sample a pending global next state $\bar{s}_g\sim P_g(\cdot|s_g, a_g)$.  Importantly, this pending global transition is computed before any local replica is updated, thereby making it consistent with the original macro-step semantics. 

Along the chain of micro nodes, only one local replica is ``active'' at a time. At stage $j$, the learner selects a single action $a_\ell\in\cA_l$, interpreted as the action for replica $j$. The environment updates only that replica via $s_j' \sim P_l(\cdot|s_j, s_g, a_\ell)$ leaving the other $k-1$ local replica states unchanged. Next, the stage pointer increments $j\gets j+1$. At the end of the chain (stage $j=k$), we commit the cached global transition by setting $s_g\gets \bar{s}_g$ and resetting $\bar{S}_g \gets \perp$, completing one macro step. 

Finally, to preserve the original discounted objective, we place the discounted reward $\gamma^t\cdot \frac{1}{k} r_l(\cdot)$ on each micro-node within the chain (without loss of generality). Summing over micro steps then exactly reproduces the intended macro step discounted return, up to the chosen horizon truncation.

\begin{algorithm}
    \begin{algorithmic}
        \caption{$k$-chained MDP for the local agent's best-response to policy $\pi_g:\cS_g\times\cS_l^k\to\cA_g$}
        \REQUIRE Transition functions $s_g' \sim P_g(\cdot|s_g, a_g)$ and $s_l' \sim P_l(\cdot|s_l, s_g, a_l)$
        \REQUIRE Global agent policy $\pi_g:\cS_g\times\cS_l^k \to \cA_g$, and stage reward $r_l:\cS_g\times\cS_l\times\cA_l\to \mathbb{R}_+$.
        \REQUIRE Subsampling parameter $k$ and horizon $H$.
        \REQUIRE $(s_{g,0}, s_{1:k,0})\sim \rho$, where $\rho$ is the initial distribution over $\cS_g\times\cS_l^k$.
        \medskip
        \STATE Define the micro-horizon $\widetilde{H} = H\cdot k$
        \STATE Let $\tilde{\cS}_l = [k] \times \cS_g\times (\cS_g\cup\{\perp\}) \times \cS_l^k$.
        \STATE Let $\tilde{\cA}_l = \cA_l$.
        \STATE Let $\tilde{s}_0 = (1, s_{g,0}, \perp, s_{1:k,0}) \sim \tilde{\rho}$, where $(s_{g,0}, s_{1:k,0})\sim \rho$.
        \STATE Let $t(\tau) = \left\lfloor\frac{\tau}{k}\right\rfloor$.
        \STATE Let $\tilde{r}_\tau((j, s_g, \bar{s}_g, s_{1:k}), a) \coloneqq \mathbbm{1}\{j=1\}\cdot \gamma^{t(\tau)}\cdot \frac{1}{k}r_l(s_1, s_g, a)$
        \medskip
        \STATE \underline{Transition kernel $\tilde{P}$}, given current state $(j, s_g, \bar{s}_g, s_{1:k})$ and action $a\in \cA_l$:
        \medskip
        \STATE \quad\textbf{if} {$\bar{s}_g=\perp$} \textbf{then}
        \STATE \quad\quad \textcolor{blue}{\texttt{/* Set the pending state of the global agent}}
        \STATE \quad\quad Compute $a_g\sim \pi_g(\cdot|s_g, s_{1:k})$
        \STATE \quad\quad $\bar{s}_g\sim P_g(\cdot|s_g, a_g)$
        \medskip
        \STATE \quad\textbf{else}
        \STATE \quad\quad $\bar{s}_g$ remains unchanged.
        \STATE \quad Update the acting replica $j$: $s_j'\sim P_l(\cdot|s_j, s_g, a)$
        \smallskip
        \STATE \quad\textbf{if} {$j<k$} \textbf{then}
        \STATE \quad\quad \textcolor{blue}{\texttt{/* Set the new state of the $j$'th local agent}}

        \STATE \quad\quad Next state is $(j+1, s_g, \bar{s}_g, (s'_1, \dots, s'_j, s_{j+1},\dots, s_k)$
        \smallskip
        \STATE \quad\textbf{else}
        \STATE \quad\quad \textcolor{blue}{\texttt{/* Commit the cached global agent update and reset scratch variables}}
        \STATE \quad\quad Next state is $(1, \bar{s}_g, \perp, s'_{1:k})$
        \medskip
        \STATE \textbf{Return} the new $k$-chained MDP $\tilde{\cM} = (\tilde{S}, \tilde{A}_l, \tilde{H}, \tilde{p}, \{\tilde{r_\tau}\}_{\tau=0}^{H'-1}, \tilde{P})$.

    \label{algorithm: k chained induced MDP}
    \end{algorithmic}
\end{algorithm}

\textbf{Algorithm \ref{algorithm: S chained induced MDP}} ($|S|$-chained mean-field MDP for learning the local agent's best-response to policy $\pi_g:\cS_g\times\cS_l\times\mu_{k-1}(\cS_l)\to\cA_g$). We also formulate a mean-field version of the MDP for the local agent's best-response to the global agent's mean-field policy $\pi_g:\cS_g\times\cS_l\times \mu_{k-1}(\cS_l)\to\cA_g$ which uses similar ideas as in Algorithm \ref{algorithm: k chained induced MDP}. 

In this parameterization, each macro step is unfolded into $|\cS_l|+1$ micro steps. At micro step $j=0$, we evaluate the global agent's mean-field policy on the population and cache its resulting pending  transition. Then, for the tagged local agent, we apply the action to the tagged agent $s_\ell$ and collect its reward. Next, for micro-steps $j=1,\dots,|\cS_l|$, we sweep through each local state-bin $u=j$. If $h(u)$ agents occupy bin $u$, we apply the shared local action for state $u$ and move the entire mass $h(u)$ to next states according to the multinomial induced by $P_l(\cdot|u,s_g,a)$, accumulating these counts in a new histogram $\bar{F}_\Delta$. After processing all bins, we commit $(s_g, s_\ell, F_\Delta)\gets (\bar{s}_g, \bar{s}_\ell, \bar{F}_\Delta)$ and repeat.\\

\begin{figure}[hbt!]
    \centering
    \includegraphics[width=1\linewidth]{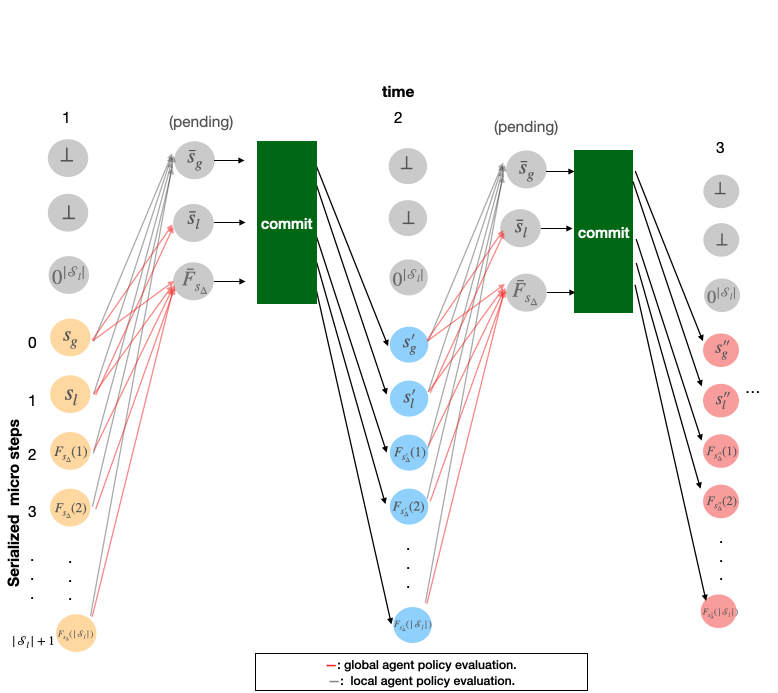}
    \caption{Construction of the $|\mathcal{S}_l|$-chained MDP construction from \cref{algorithm: S chained induced MDP}. In this parameterization, each macro step is unfolded into $|\cS_l|+1$ micro steps. At micro step $j=0$, we evaluate the global agent's mean-field policy on the population and cache its resulting pending  transition. Then, for the tagged local agent, we apply the action to the tagged agent $s_\ell$ and collect its reward. Next, for micro-steps $j=1,\dots,|\cS_l|$, we sweep through each local state-bin $u=j$. If $h(u)$ agents occupy bin $u$, we apply the shared local action for state $u$ and move the entire mass $h(u)$ to next states according to the multinomial induced by $P_l(\cdot|u,s_g,a)$, accumulating these counts in a new histogram $\bar{F}_\Delta$. After processing all bins, we commit $(s_g, s_\ell, F_\Delta)\gets (\bar{s}_g, \bar{s}_\ell, \bar{F}_\Delta)$ and repeat.\\}
    \label{figure: sl-chained MDP construction}
\end{figure}

\begin{algorithm}[t]
    \begin{algorithmic}
        \caption{$|\cS_l|$-chained MF-MDP for learning the local agent's best-response to policy $\pi_g:\cS_g\times\cS_l\times\mu_{k-1}(\cS_l)\to\cA_g$}
       \REQUIRE Transition functions $s_g'\sim P_g(\cdot|s_g, a_g)$ and $s_l' \sim P_l(\cdot|s_l, s_g, a_l)$.
        \REQUIRE Global agent policy $\pi_g:\cS_g\times\cS_l\times\mu_{k-1}(\cS_l) \to \cA_g$,  stage reward $r_l:\cS_g\times\cS_l\times\cA_l\to \mathbb{R}_+$.
        \REQUIRE Subsampling parameter $k$ and horizon $H$.
        \REQUIRE $(s_{g,0}, s_{\ell,0}, h_0)\sim \rho$, where $\rho$ is the initial distribution over $\cS_g\times\cS_l \times\mu_{k-1}(\cS_l)$.\medskip
                \STATE Let $\phi:[|\cS_l|]\to \cS_l$ enumerate $\cS_l$, i.e. $\phi(j)$ is the $j$'th element of $\cS_l$.
\STATE Define the micro-horizon $\tilde{H} = H\cdot (|\cS_l|+1)$
        \STATE Define $\tilde{\cS} = \{0,1,\dots,|\cS_l|\} \times \cS_g \times (\cS_g \cup \{\perp\}) \times \cS_l \times (\cS_l \cup \{\perp\}) \times \mu_{k-1}(\cS_l) \times \mu_{k-1}(\cS_l)$.
        \STATE Let $\tilde{\cA} = \cA_l$.
        \STATE Let $\tilde{s}_0 = (0, s_{g,0}, \perp, s_0, \perp, F_{s_\Delta}, \mathbf{0}^{|\cS_l|})\sim \tilde{\rho}$, where $(s_g, s_0, F_{s_\Delta}) \sim \rho$
        \STATE Let $t(\tau) = \lfloor \frac{\tau}{|\cS_l|+1}\rfloor$.
        \STATE Let $\tilde{r}_\tau((j, s_g, \bar{s}_g, s_\ell, \bar{s}_\ell, F_{\Delta}, \bar{F}_{\Delta}), a) \coloneqq \mathbbm{1}\{j=0\}\cdot \gamma^{t(\tau)}\cdot \frac{1}{k}r_l(s_\ell, s_g, a)$.
        \medskip
        \STATE \underline{Transition kernel $\tilde{P}$}, given state $(j, s_g, \bar{s}_g, s_\ell, \bar{s}_\ell, F_{\Delta}, \bar{F}_\Delta)$ and action $a\in \cA_l$, the new state is $(j', s'_g, \bar{s}'_g, s'_\ell, \bar{s}'_\ell, F'_{\Delta}, \bar{F}'_\Delta)$:
            \medskip
        \IF{$j=0$}
        \STATE \textcolor{blue}{\texttt{/* Set the pending states of the global agent and local agent $0$}}
        \STATE Compute $a_g\sim \pi_g(\cdot|s_g, s_\ell, F_{\Delta})$
        \STATE $\bar{s}_g\sim P_g(\cdot|s_g, a_g)$
        \STATE $\bar{s}_\ell \sim P_l(\cdot|s_\ell,s_g,a)$
        \STATE Move to state $(1, s_g, \bar{s}_g, s_\ell, \bar{s}_\ell, F_{\Delta}, \bar{F}_\Delta)$
        \medskip
        \ELSIF{$j=u\in\{1,\dots,|\cS_l|\}$}
        \STATE Let $c = (k-1) F_{s_\Delta}(\phi(u))$ be the number of agents in state $\phi(u)$.
        \STATE Let $(N_{u\to 1}, \dots, N_{u\to |\cS_l|}) \sim \text{Multinomial}\left(c, \left\{P_l(\phi(z)|\phi(u), s_g, a)\right\}_{z=1}^{|\cS_l|}\right)$
        \smallskip
        \STATE \textcolor{blue}{\texttt{/* Shift mass of the $c$ local agents at state $\phi(z)$ to their new states}}
        \smallskip
        \FOR{$z\in \{1,\dots,|\cS_l|\}$}
        \STATE Let $\bar{F}_\Delta(\phi(z))\coloneqq \bar{F}_\Delta(\phi(z)) + \frac{N_{u\to z}}{k-1}$
        \ENDFOR
        \smallskip
        \IF{$u < |\cS_l|$}
        \smallskip
        \STATE \textcolor{blue}{\texttt{/* Update counter to examine the next state}}
        \STATE Move to state $(u+1, s_g, \bar{s}_g, s_\ell, \bar{s}_\ell, F_{\Delta}, \bar{F}_\Delta)$
        \smallskip
        \ELSE
        \smallskip
        \STATE \textcolor{blue}{\texttt{/* Commit the cached updates and reset scratch variables}}
        \STATE Move to state $(0, \bar{s}_g, \perp, \bar{s}_\ell, \perp, \bar{F}_{\Delta}, \mathbf{0}^{|\cS_l|})$
        \ENDIF
        \ENDIF
        
        \medskip
        \STATE \textbf{Return} the new $|\cS_l|$-chained mean-field MDP $\tilde{\cM} = (\tilde{S}, \tilde{A}_l, \tilde{H}, \tilde{p}, \{\tilde{r_\tau}\}_{\tau=0}^{H'-1}, \tilde{P}$).
        \STATE 
        \label{algorithm: S chained induced MDP}
    \end{algorithmic}
\end{algorithm}
 
\begin{theorem}[Local agent best-response optimality and sample-complexity guarantees] 
\label{appendix theorem: local agent br optimality}
With probability at least $1-\delta$, we have that the learned local-agent's policy is an $\epsilon$ best-response to the global agent's fixed policy $\pi_g$ with sample complexity
\[\tilde{O}\left(\min\left\{\frac{|\cS_l|^5|\cS_g|^2 k^{2|\cS_l|}|\cA_l|}{\epsilon^2(1-\gamma)^2} \ln \frac1\delta, \frac{k^5  |\cS_g|^2 |\cS_l|^k |\cA_l|}{\epsilon^2(1-\gamma)^3} \ln \frac1\delta\right\}\right),\]
and runtime
\[\tilde{O}\left(\frac{|\cS_l|^9 |\cS_g|^4  k^{4|\cS_l|}|\cA_l|^2}{\epsilon^2(1-\gamma)^3} \log \frac1\delta \log \frac1\epsilon, \frac{k^5  |\cS_g|^4 |\cS_l|^{2k}|\cA_l|^2}{\epsilon^2(1-\gamma)^3} \log \frac1\delta \log \frac1\epsilon\right).\]
\label{thm: local agent optimality and sample complexity}
\end{theorem} 

\begin{proof}
    With the $|\cS_l|$-chained mean-field MDP, the size of the MDP is $|\cS_g|^2|\cS_l|^3\cdot k^{2|\cS_l|}$, from Algorithm \ref{algorithm: S chained induced MDP}. The sample complexity of solving this is $\tilde{O}\left(\frac{\tilde{H}^2 |\cS_g|^2|\cS_l|^3\cdot k^{2|\cS_l|} |\cA_l|}{\epsilon^2} \ln \frac1\delta\right)$
    and the runtime is  $\tilde{O}\left(\frac{\tilde{H}^3 |\cS_g|^4  |\cS_l|^6 k^{4|\cS_l|}|\cA_l|^2}{\epsilon^2} \log \frac1\delta \log \frac1\epsilon\right)$,
    from Theorem \ref{brunskill sample bound}.  Next, using $H \coloneqq \frac{1}{1-\gamma}\log \frac{\|r\|_\infty}{\epsilon(1-\gamma)}$, we have
    \begin{align*} 
    \tilde{H} &= H(|\cS_l|+1) \leq \frac{2|\cS_l|}{1-\gamma}\log \frac{\|r\|_\infty}{\epsilon(1-\gamma)}.
    \end{align*}
    Plugging this, we have that the sample complexity of solving the $|\cS_l|$-chained mean-field MDP is 
    \[\tilde{O}\left(\frac{|\cS_l|^5|\cS_g|^2 k^{2|\cS_l|}|\cA_l|}{\epsilon^2(1-\gamma)^2}\ln \frac1\delta\right),\]
    and the runtime is
    \begin{align*}\tilde{O}\left(\frac{|\cS_l|^9 |\cS_g|^4  k^{4|\cS_l|}|\cA_l|^2}{\epsilon^2(1-\gamma)^3} \log \frac1\delta \log \frac1\epsilon \right).
    \end{align*}

     Similarly, for the $k$-chained MDP, the size of the MDP is $k |\cS_g|^2|\cS_l|^k$, from Algorithm \ref{algorithm: k chained induced MDP}. Therefore, the sample complexity of solving this is $\tilde{O}\left(\frac{ \tilde{H}^2 |\cS_g|^2 |\cS_l|^k |\cA_l|}{\epsilon^2} \ln \frac1\delta\right)$ and the runtime is $\tilde{O}\left(\frac{H^3 k^2  |\cS_g|^4 |\cS_l|^{2k}|\cA_l|^2}{\epsilon^2} \log \frac1\delta \log \frac1\epsilon\right)$ from Theorem \ref{brunskill sample bound}.  Now, using $H\coloneqq \frac{1}{1-\gamma}\log \frac{\|r\|_\infty}{\epsilon(1-\gamma)}$, we have
     \begin{align*}\tilde{H} &= kH  \leq \frac{k}{1-\gamma} \log \frac{\|r\|_\infty}{\epsilon(1-\gamma)}.\end{align*}
     Plugging this, we have that the sample complexity of solving the $k$-chained MDP is \[\tilde{O}\left(\frac{k^5  |\cS_g|^2 |\cS_l|^k |\cA_l|}{\epsilon^2(1-\gamma)^3} \ln \frac1\delta\right)\] and the runtime is \[\tilde{O}\left(\frac{k^5  |\cS_g|^4 |\cS_l|^{2k}|\cA_l|^2}{\epsilon^2(1-\gamma)^3} \log \frac1\delta \log \frac1\epsilon\right).\]
     Taking the minimum of these sample complexities proves the theorem.\qedhere
\end{proof}

\section{Proof of convergence of the best-response dynamics}

\label{sec: proof of convergence of BR dynamics}
We begin with an idealized analysis where the global agent does not subsample (i.e., $k=n$), and the local agent does a best-response policy. We show that this best-response dynamics converges at an exponential rate to a Nash Equilibrium. To determine if the best-response oracle function converges, we define the Bellman optimality gap.

\begin{definition}[Bellman optimality gap] We formalize the Bellman optimality gap, which for a joint policy $\pi^t = (\pi_g^t, \pi_l^t)$, represents the incentive to deviate: \[\mathsf{\Delta}_g^t(s) \coloneqq V^{\mathsf{BR}_g(\pi_l^t), \pi_l^t}(s) - V^{\pi_g^t, \pi_l^t}(s)\]
\[
\mathsf{\Delta}_l^t(s) \coloneqq V^{\pi_g^t, \mathsf{BR}_l(\pi_g^t)}(s) - V^{\pi_g^t, \pi_l^t}(s)\]
Then, define the global agent's and local agent's gap function as
\[\mathsf{Gap}^t_1(\pi_g,\pi_l) \coloneqq \max_{s\in \cS} \mathsf{\Delta}_g^t(s) =  \max_{s\in\cS}\left(\max_{\pi_g^*}V^{\pi^*_g,\pi_l}(s) - V^{\pi_g,\pi_l}(s)\right) > 0.\]
and 
\[\mathsf{Gap}^t_2(\pi_g,\pi_l) \coloneqq \max_{s\in \cS}\mathsf{\Delta}_l^t(s) =  \max_{s\in\cS}\left(\max_{\pi_l^*}V^{\pi_g,\pi_l^*}(s) - V^{\pi_g,\pi_l}(s)\right) > 0.\]
Let \[G^t \coloneqq \max\{\mathsf{Gap}_1^t, \mathsf{Gap}_2^t\}.\]
Then, a policy $\pi^t$ is an $\epsilon$-approximate Nash-Equilibrium (henceforth, denoted $\epsilon$-NE) if $G^t\leq \epsilon$.\\
\end{definition}

\begin{theorem}[Convergence of Best-Response Dynamics
    \cite{Monderer1996,Young:2004}] \label{thm: convergence of br dynamics} The sequence of policies in the discrete time best-response dynamic converge to pure Nash equilibria in potential games.\\
\end{theorem}

\begin{remark}
Note that \cref{thm: convergence of br dynamics} does not, a priori, imply a \emph{rate} on the convergence to a Nash equilibrium. In fact, in general, the rate of convergence is not guaranteed to be fast since only games with very special structure are known to have best-response dynamics that converge in (even) polynomial time, e.g. potential games, congestion games, etc. For instance, while best-response dynamics in continuous time can converge exponentially quickly, the discrete case is not well understood (Karlin's conjecture), and under certain adversarial tie-breaking schemes in fictitious play, there are problem instances that converge exponentially \emph{slowly} in some dimension parameters \cite{lazarsfeld2025fastfurioussymmetriclearning,Daskalakis_2014,craven2005karlin,doi:10.1137/1.9781611976465.84}.  In this section, we show that our problem instance is a potential game \cite{guo2025markovalphapotentialgames,yan2025markovpotentialgameconstruction,ding2022independent,arefizadeh2024characterizationspotentialordinalpotential,leonardos2025globalconvergencemultiagentpolicy,overman2025oversightgamelearningcooperatively} and use it to obtain provable convergence guarantees.\\
\end{remark}

\begin{lemma}[$\mathcal{M}$ is a Markov Potential Game]\label{lemma: M is a Markov potential game} Consider the finite discounted Markov game  $\mathcal{M}$ defined in \cref{sec: preliminaries} with players $\cI = \{g, 1, \dots, n\}$. For any joint policy $\pi = (\pi_g, \pi_1, \dots, \pi_n)$, define every player's utility from initial state $s_0$ to be the common welfare return \[U_i^\pi(s_0) = \E_\pi\left[\sum_{t=0}^\infty \gamma^t r(s_t, a_t) \mid s_0\right],\quad i\in\cI,\]where $r(s,a) = r_g(s_g, a) + \frac1n \sum_{j=1}^n r_l(s_j, s_g, a_j)$. Then, the game is an exact Markov potential game with stage potential $\Phi(s,a) = r(s,a)$.
\end{lemma}
\begin{proof}
Fix any player $i\in\cI$, any joint policy $\pi = (\pi_i, \pi_{-i})$, any unilateral deviation $\pi_{i'}'$, and any initial state $s_0$. Let $\mathbb{P}_{\pi_i', \pi_{-i}}$ and $\mathbb{P}_{\pi_i, \pi_{-i}}$ denote the trajectory laws induced by the two policy profiles. Since every player’s stage utility is the common stage reward $r(s,a)$, we have that
\[U_i^{\pi'_i, \pi_{-i}}(s_0) - U_i^{\pi_i, \pi_{-i}}(s_0) = \E_{\pi_i', \pi_{-i}}\left[\sum_{t=0}^\infty \gamma^t r(s_t, a_t)|s_0\right] - \E_{\pi_i, \pi_{-i}}\left[\sum_{t=0}^\infty  \gamma^t r(s_t, a_t)|s_0\right].\]
By definition, $\Phi=r$, and hence the right hand side is exactly 
\[\E_{\pi_i', \pi_{-i}}\left[\sum_{t=0}^\infty \gamma^t \Phi(s_t, a_t)|s_0\right] - \E_{\pi_i, \pi_{-i}}\left[\sum_{t=0}^\infty  \gamma^t \Phi(s_t, a_t)|s_0\right].\]
Therefore, for every unilateral deviation, we have that
\[U_i^{\pi'_i, \pi_{-i}}(s_0) - U_i^{\pi_i, \pi_{-i}}(s_0) = \Phi^{\pi'_i, \pi_{-i}}(s_0) - \Phi^{\pi_i, \pi_{-i}}(s_0),\]
where $\Phi^\pi(s_0) = \E_\pi[\sum_{t=0}^\infty \gamma^t \Phi(s_t, a_t) | s_0]$.

Therefore, the game is an exact Markov potential game with potential $\Phi$. Specifically, in our idealized analysis (exact best-response, finite policy class, finite action space), the alternating global/local best-response dynamics performs coordinate ascent on $\Psi(\pi)$ and therefore converges to a pure Nash equilibrium (a local maximizer of $\Phi$).\qedhere\\
\end{proof}

\paragraph{Best-Response Dynamics Formulation.} We begin by formulating the problem of studying the best-response dynamics for the setting of $k=n$. We then generalize the formulation to $k\leq n$ and study the convergence rate of the best-response dynamics.

When $k=n$, in the case where there is no subsampling, we have that the local agents are homogeneous with a shared stationary policy $\pi_l$. 
Then, the joint policy is parametrized by a pair $(\pi_g, \pi_l)$, and the best-response dynamics are ones both players have exact BR oracles.
 Suppose we have a $\mathsf{BR}$ oracle for the local agent $\pi_l^*(\pi_g,s)$ and an oracle for the global agent $\pi_g(s)$ (i.e., where $k=n$, so the global agent does not do any sampling of the local agents). In the idealized sequential dynamics, the policy updates are given by \begin{equation}\pi_g^{t+1} = \mathsf{BR}_g(\pi_1^t, \pi_2^t, \dots, \pi_n^t) = \overline{\mathsf{BR}}_g(\pi_\ell^t)\end{equation} and \begin{equation} \pi_i^{t+1} = \mathsf{BR}_\ell(\pi_g^{t+1}), \forall i\in[n].\end{equation}
 
 Specifically, $BR_g(\pi_l)$ is the optimal global policy against an infinite-horizon MDP whose local agents are i.i.d. following $\pi_l$. Similarly, $\mathsf{BR}_l(\pi_g)$ is a local policy that is optimal for a single local agent in the induced MDP $M_l(\pi_g)$.
 
\begin{definition}[Best-Response ($\mathsf{BR}$) Dynamics Operator] \emph{Let the space of all joint policy profiles be $\Pi = \Pi_g\times\Pi_l$. Define the best-response operator $\Gamma:\Pi\to \Pi$ to be} \begin{equation}\Gamma(\pi_g, \pi_l) = (\mathsf{BR}_g(\pi_l), \mathsf{BR}_l(\pi_g)),\end{equation}
where $\mathsf{BR}_g(\pi_l)$ is the global agent's best response to $\pi_l$ and $\mathsf{BR}_l(\pi_g)$ is the local agent's best response to $\pi_g$. A fixed point of this operator $\pi^* = \Gamma(\pi^*)$ is a Nash Equilibrium because each agent is playing a best response to the other.
\end{definition}

The below theorem gives an exponential rate for continuous dynamics.

\begin{theorem}[Theorem 2 in \cite{doi:10.1137/17M1139461}] Suppose $\Gamma$ is a potential game. Then, for all (but a measure $0$ set of) initial points $x_0$, the continuous best-response dynamics in the idealized sequential dynamics setting converges to a Nash equilibrium at an  exponential rate.
\end{theorem}

In the discrete time dynamics, the worst-case bound can be exponential in the number of players, and can depend on the size of the policy simplex which is exponential in the number of actions. 
\begin{theorem}[Theorem 1 in \cite{DurandGaujal:2016:ACBRA}] For an $N$-player game with joint action set $\mathcal{A}$, the round-robin dynamics in a Markov potential game reaches a Nash Equilibrium in $T_{BRA} = N |A|^{N-1}$ iterations.
\end{theorem}

\subsection{ Analysis with Probabilistic Approximate Best-Response Dynamics with Subsampling}

\begin{remark}
    Subsampling and the introduction of Bellman noise may only affect how the agents choose policies, but leaves the reward structure invariant for the agents. Therefore, the true Markov game and its potential $\Phi$ are unchanged, i.e. for any joint policy $\pi$ and any unilateral deviation $\pi'_i$, we still have that for any $s_0 \in \mathcal{S}$,
    \begin{equation}\Phi(\pi_i', \pi_{-i}) - \Phi(\pi) = V^{\pi'_i, \pi_{-i}}(s_0) - V^\pi(s_0).\end{equation}
\end{remark}

\begin{definition}[Best-Response Dynamics Operator with Sampling] \emph{Consider the policy simplexes $\Pi = \Pi_g^{(k)}\times\Pi_l$, where $\Pi_g^{(k)}$ is the simplex corresponding to policies $\mathcal{S}_g\times\mathcal{S}_l^k \to \mathcal{A}_g$, or the mean-field space of policies $\cS_g \times \cS_l \times \mu_k{(\cS_l)} \to \cA_g$. Define the best-response operator with $k$-agent subsampling $\Gamma_k:\Pi\to \Pi$ to be \begin{equation}\Gamma(\pi^{(k)}_g, \pi_l) = (\mathsf{BR}^{(k)}_g(\pi_l), \mathsf{BR}_l(\pi_g^{(k)})),
\end{equation}
where $\mathsf{BR}^k_g(\pi_l)$ is the global agent's approximate best-response to $\pi_l$ and $\mathsf{BR}_l(\pi^{(k)}_g)$ is the local agent's best response to $\pi^{(k)}_g$. A fixed point of this operator $\pi_k^* = \Gamma_k(\pi_k^*)$ is a Nash Equilibrium because each agent is playing a best response to the other.}\\
\end{definition}

We first characterize the notion of a global-agent playing an $\epsilon_g$-best response to $\pi_\ell$.
\begin{definition}
    \emph{Fix a local agent policy $\pi_\ell$. Then, a global agent policy $\tilde{\pi}_g$ is called an $\epsilon_g$-best response to $\pi_\ell$ if 
    \begin{equation}\max_{s\in \mathcal{S}} \left[V_{\mathsf{BR}_g(\pi_\ell),\pi_\ell}(s) - V_{\tilde{\pi}_g, \pi_\ell}(s)\right] \leq \epsilon_g.\end{equation}
    Equivalently, this can be characterized by $\mathsf{Gap}_1(\tilde{\pi}_g, \pi_\ell) \leq \epsilon_g$.}
\end{definition}

\begin{remark}From \cref{theorem: application of PDL}, we see that by using $k$-agent subsampling, the global agent can compute a policy $\pi_{g,k}^{\mathrm{est}}$ such that for any fixed $\pi_\ell$, we have an $\epsilon_g(k)$ best-response with probability at least $1 - \delta_1$, where
\begin{equation}\epsilon_g(k)\coloneqq \sup_{s\in \mathcal{S}}\left[V^{\mathsf{BR}_g(\pi_\ell),\pi_\ell}(s) - V^{\pi_{g,k,m}^{\mathrm{est}}(\pi_\ell),\pi_\ell}(s)\right] \leq  \tilde{O}\left(\frac{1}{\sqrt{k}}\right).\end{equation}
Similarly, the local-agent $Q$-learning algorithm that we use for learning the policy uses $N_{\delta_2}$ samples to, with probability at least $1-\delta_2$, yield an $\epsilon_\ell$ approximate best-response where 
\begin{equation}\epsilon_\ell = \sup_{s\in \mathcal{S}}\left[V^{\pi_g, \mathsf{BR}_\ell(\pi_g)}(s) - V^{\pi_g, \pi_\ell^{\mathrm{alg}}}(s)\right].\end{equation}
\end{remark}

\begin{definition}[Per-update approximate scale] \emph{Let the per-update approximation scale, the largest approximation error of each learning algorithm, be given by $\eta \coloneqq \max\{\epsilon_g(k), \epsilon_\ell\}$.}\\
\end{definition}

\begin{remark} The dynamics of the system are that under sampling, each iteration returns an $\epsilon_g(k)$-BR to $\pi_\ell^t$ and an $\epsilon_\ell$-BR to $\tilde{\pi}_g^{t+1}$ with probability at least $1-\delta$. Since the meta-algorithm decides whether to adopt these proposals as the new policies or not using the accept rule \begin{equation}V^{\tilde{\pi}_i, \pi_{-i}}(s_0) > V^{\pi_i, \pi_{-i}}(s_0),\end{equation} we have that whenever we move to a new joint policy the potential $\Phi$ strictly increases. Hence, if a proposed approximate BR worsens the value by at least $2\eta$, we simply keep the old policy, so that the potential does not decrease. In the case where the approximate BR has a value within $2\eta$ of the original value, then by definition we are at a $2\eta$-approximate Nash equilibrium. We show that we reach an $2\eta$-approximate Nash equilibrium in finite time.
\end{remark}

We first argue the correctness of the interval comparison in the \texttt{UPDATE} function in \cref{algorithm: alternating marl} (\texttt{ALTERNATING-MARL}).
\begin{lemma}
   The interval comparison in the \emph{\texttt{UPDATE}} function in \emph{\texttt{ALTERNATING-MARL}} is correct.\label{lemma: correctness of interval comparison} 
\end{lemma}
\begin{proof}We consider policy updates in the order $\pi_g, \pi_\ell, \pi_g', \pi_\ell'$, where $\pi_\ell$ is the local agent's best-response to $\pi_g$, $\pi_g'$ is the global agent's best-response to $\pi_\ell$, and $\pi_\ell'$ is the local agent's best-response to $\pi_g'$. Recall that for all $s\in\mathcal{S}_g\times\mathcal{S}_l^k$, we have that with probability at least $1-\delta$, where $V$ is the true value of the best-response policy and $\hat{V}$ is the value of the learned best-response policy,
 \begin{equation}|V^{\pi_\ell,\pi_g}(s) - \hat{V}^{\pi_\ell,\pi_g}(s)| < \eta,\end{equation}
    \begin{equation}|V^{\pi_\ell',\pi_g'}(s) - \hat{V}^{\pi_\ell',\pi'_g}(s)| < \eta,\end{equation}    \begin{equation}|V^{\pi_\ell,\pi_g'}(s) - \hat{V}^{\pi_\ell,\pi'_g}(s)| < \eta.\end{equation}

\begin{figure*}[t]
\centering
\begin{tikzpicture}[
  >=Latex,
  font=\small,
  x=1.15cm,y=1cm,
  region/.style={rounded corners=2pt, draw=black!20, line width=0.4pt},
  act/.style={rounded corners=2pt, draw=black!30, fill=white, line width=0.4pt,
              inner sep=4pt, align=center, text width=3.2cm},
  tinybox/.style={rounded corners=2pt, draw=black!35, fill=white, line width=0.4pt,
                  inner sep=4pt, align=center, text width=3.0cm},
  label/.style={align=center, text width=3.2cm},
]

% ---------------- Top bar (Δ_k) ----------------
\def\ybar{3.0}
\def\h{0.55}
\def\xL{-4.4}
\def\xR{ 4.4}
\def\xm{-1.9}   % -epsilon location
\def\xp{ 1.9}   % +epsilon location

% subtle backgrounds (lighter than your screenshot)
\fill[red!10]   (\xL,\ybar+\h) rectangle (\xm,\ybar-\h);
\fill[black!7]  (\xm,\ybar+\h) rectangle (\xp,\ybar-\h);
\fill[green!10] (\xp,\ybar+\h) rectangle (\xR,\ybar-\h);

% outline + separators
\draw[region] (\xL,\ybar+\h) rectangle (\xR,\ybar-\h);
\draw[black!25] (\xm,\ybar+\h) -- (\xm,\ybar-\h);
\draw[black!25] (\xp,\ybar+\h) -- (\xp,\ybar-\h);

% axis with ticks (minimal)
\draw[->,black!65] (\xL-0.2,\ybar) -- (\xR+0.2,\ybar)
  node[right] {$\Delta_k$};
\draw (\xm,\ybar) -- ++(0,0.16) -- ++(0,-0.32) node[below] {$-2\eta$};
\draw (0,\ybar)  -- ++(0,0.16) -- ++(0,-0.32) node[below] {$0$};
\draw (\xp,\ybar) -- ++(0,0.16) -- ++(0,-0.32) node[below] {$+2\eta$};

% region headers (short!)
\node[label] at (-3.1,\ybar+0.95) {$\Delta_k \le - 2\eta$\\ \textbf{reject}};
\node[label] at (0,\ybar+1.05)   {$|\Delta_k|\le 2\eta$\\ \textbf{stop}};
\node[label] at (3.1,\ybar+0.95) {$\Delta_k \ge +2\eta$\\ \textbf{accept}};

% Δ definition (keep it compact and below the bar, centered)

% ---------------- Actions row (aligned under regions) ----------------
\def\yact{0.55}

\node[act] (rej) at (-3.1,\yact)
{\textbf{Reject proposal}\\
keep $(\pi_g^{k},\pi_\ell^{k})$};

\node[act] (stop) at (0,\yact)
{\textbf{Stop \& certify}\\
output $(\pi_g^{k},\pi_\ell^{k})$ as a \\
$2\eta$-approximate Nash certificate};

\node[act] (acc) at (3.1,\yact)
{\textbf{Accept proposal}\\
use $(\pi_g',\pi_\ell')$};

% arrows down from regions to actions (clean and symmetric)
\draw[->,black!55] (-3.1,\ybar-\h-0.05) -- (rej.north);
\draw[->,black!55] ( 0.0,\ybar-\h-0.05) -- (stop.north);
\draw[->,black!55] ( 3.1,\ybar-\h-0.05) -- (acc.north);

% short middle-case hint (optional; keep tiny)
\node[align=center, font=\scriptsize, text width=5.4cm] at (0,-0.8)
{If $|\Delta_k|\le2\eta$, where
$\Delta_k \coloneqq V^{\pi_g',\pi_\ell'} - V^{\pi_g,\pi_\ell}$ then no unilateral approximate\\
best-response gain exceeds $2\eta$.};

\end{tikzpicture}
\caption{Decision rule for consecutive best-response iterates using the progress statistic
$\Delta_k = V^{\pi_g',\pi_\ell'} - V^{\pi_g,\pi_\ell}$. The middle region $|\Delta_k|\le2\eta$ triggers stopping and provides a $2\eta$-approximate Nash certificate. }
\end{figure*}

Therefore, we delineate to the following cases:

\textbf{Case 1:} Accept the new policy because the new pair is better on every state by at least $2\eta$, i.e. the $\eta$-interval of the new policy lies ahead of the $\eta$-interval of the old policy. In other words, $\hat{V}^{\pi_\ell',\pi_g} > 
\hat{V}^{\pi_\ell,\pi_g} + 2\eta$.

  \textbf{Case 2:} Reject the new policy because there is a state for which the new pair is clearly worse by at least $2\eta$, i.e. the $\eta$-interval of the old policy lies ahead of the $\eta$-interval of the new policy.  In other words, $\hat{V}^{\pi_\ell',\pi_g} < \hat{V}^{\pi_\ell,\pi_g} - 2\eta$.

   \textbf{Case 3:} The intervals of the policies overlap, i.e. $\hat{V}^{\pi_\ell,\pi_g} - 2\eta \leq \hat{V}^{\pi_\ell',\pi_g'} \leq \hat{V}^{\pi_\ell,\pi_g} + 2\eta$. In this case, we terminate the algorithm since $\pi_g, \pi_\ell$ is already a $2\eta$-NE. To see this, suppose for the sake of contradiction that this is not the case. Then, without loss of generality, there exists a local agent policy $\tilde{\pi}_\ell$ such that $V^{\tilde{\pi}_\ell,\pi_g} \geq V^{\pi_\ell,\pi_g} + 2\eta$, with probability at least $1-\delta$. However, this implies that the (approximate probabilistic) best-response of the local agents algorithm cannot have been $\pi_\ell$, which contradicts \cref{thm: local agent optimality and sample complexity}. Therefore, the initial assumption is wrong and $\pi_g, \pi_\ell$ is a $2\eta$-NE.

   Together, these imply the correctness of the \texttt{UPDATE} function in \texttt{ALTERNATING-MARL}.\qedhere
\end{proof} 

\begin{lemma}[Convergence of the \texttt{ALTERNATING-MARL} dynamics to a $2\eta$-approximate Nash-Equilibrium]\label{lemma: bound nsteps} Let $\Pi$ be the restricted joint space, and let \[\Psi(\pi) = J_\rho(\pi)\coloneqq \E_{s_0\sim \rho} \E_{\pi}\left[\sum_{t=0}^\infty \gamma^t r(s_t, a_t)\right]\]
be the common welfare potential. Assume $0\leq r(s, a)\leq \tilde{r}$ for all $ (s,a)$, and let $\{\pi^q\}_{q\geq 0}$ denote the subsequence of policy profiles after accepted updates only. Suppose the update rule accepts a proposed policy profile only if $\Psi(\pi^{q+1}) \geq \Psi(\pi^q) + \eta$ for some chosen threshold $\eta>0$. Then, the number of accepted updates satisfies 
\[N_{acc} \leq \frac{\tilde{r}}{\eta(1-\gamma)},\] and the total number of iterations until we are at a $2\eta$-approximate Nash Equilibrium satisfies \[N_{steps}\leq \left\lceil \frac{\tilde{r}}{\eta(1-\gamma)}\right\rceil + 1.\]

\end{lemma}
\begin{proof} 
Since the reward is nonnegative and bounded by $\tilde{r}$, we have that $0\leq r(s_t, a_t)\leq \tilde{r}$ for all $t$. Therefore, for every policy profile $\pi$, we have that \[0\leq \Psi(\pi) = \E_\pi\left[\sum_{t=0}^\infty \gamma^t r(s_t, a_t)\right]\leq \sum_{t=0}^\infty \gamma^t \tilde{r} = \frac{\tilde{r}}{1-\gamma}.\]
Thus, the potential range is bounded as\[\Psi_{max} - \Psi_{min} \leq \frac{\tilde{r}}{1-\gamma}.\]
Now, suppose there are $N_{acc}$ accepted updates. By the threshold accept rule, each accepted update raises the potential by at least $\eta$. Therefore, for $q=\{0,\dots,N_{acc} - 1\}$, we have that \[\Psi(\pi^{q+1}) - \Psi(\pi^q) \geq \eta.\]
Summing over all accepted updates gives \[\sum_{q=0}^{N_{acc}-1}[\Psi(\pi^{q+1}) - \Psi(\pi^q)] \geq N_{acc} \eta,\]
and the left-hand side telescopes give \[\sum_{q=0}^{N_{acc}-1}[\Psi(\pi^{q+1}) - \Psi(\pi^q)]  = \Psi(\pi^{N_{acc}}) - \Psi(\pi^0).\]
Therefore, we have $N_{acc} \eta \leq \Psi(\pi^{N_{acc}}) - \Psi(\pi^0) \leq \Psi_{max} - \Psi_{min} \leq \frac{\tilde{r}}{1-\gamma}$. Dividing by $\eta$, we have $N_{acc} \leq \frac{\tilde{r}}{\eta(1-\gamma)}$, which proves the first bound. 
The second bound follows by noting that if the current joint policy $\pi$ is not a $2\eta$-NE, then player $i\in \{g,\ell\}$ has a profitable deviation such that
\[\Delta_i(\pi) = \max_{\pi_i’} \left[V^{\pi_i’, \pi_{-i}}(s_0) - V^\pi(s_0)\right] > \eta.\]
Therefore, the approximate BR strictly improves player $i$’s value by at least $\Delta_i(\pi) - \eta$, and hence the potential increases by the same value:
\[\Phi(\tilde{\pi}_{i}, \pi_{-i}) - \Phi(\pi) = V^{\tilde{\pi}_i, \pi_{-i}}(s_0) - V^\pi(s_0) \geq \Delta_i(\pi) - \eta.\]
Hence, as long as the current policy is not a $2\eta$-NE, whichever player has gap $\epsilon$ can, when chosen to update, increase the potential by at least $\epsilon - \eta$ (where $\epsilon=2\eta)$. 
When we are outside the $\epsilon$-NE set and a player with gap $> \epsilon$ is updated, the potential increases by at least $\epsilon - \eta$. The number of such improving updates is at most $\frac{\tilde{r}}{\eta(1-\gamma)}$, we proved previously. Moreover, if any alternating sweep has no accepted update, then the algorithm terminates. Hence, the number of total sweeps is at most\[N_{steps}\leq \left\lceil \frac{\tilde{r}}{\eta(1-\gamma)}\right\rceil + 1,\]
which proves the theorem.\qedhere\\
\end{proof}

\begin{theorem}[Bounding $\eta$ and computing the sample complexity required for a $2\eta$-approximate Nash Equilibrium]\label{thm: bounding sample complexity} Fix $\delta\in(0,1)$. Let \[\eta = \frac{2\tilde{r}}{(1-\gamma)^2}\left(\sqrt{\frac{n-k+1}{2nk} \ln(2|\mathcal{S}_l||\mathcal{A}_g|\sqrt{k})} +  \frac{1}{\sqrt{k}}\right) + \frac{4}{\sqrt{k}(1-\gamma)} + \frac{1}{\sqrt{k}}.\]
Then, with probability at least $1-\delta$, \emph{\texttt{ALTERNATING-MARL}} learns a  $2\eta$-NE with total sample complexity
\[\tilde{O}\left(\min\left\{\frac{  \|r_l\|^3_\infty |\mathcal S_g|^2 |\mathcal A_l| |\mathcal A_g| |\mathcal S_l|^5  k^{2|\mathcal S_l|+4.5}}{(1-\gamma)^6} \log\frac1\delta, \frac{k^{6.5} \|r_l\|^3_\infty |\cS_g|^2 |\mathcal A_l| |\cA_g| |\cS_l|^k }{(1-\gamma)^6}\log \frac1\delta\right\}\right),\]
\end{theorem}
\begin{proof}
From \cref{lemma: epsilon_km_is_k}, we have that with probability at least $1-\delta_1$, with $M_{\delta_1}$ samples where 
\begin{align*}M_{\delta_1} &= \frac{2k^3\gamma^2\|r_l\|^2_\infty |\cS_g| |\cA_g|}{(1-\gamma)^5}\log \frac{\tilde{r}\sqrt{k}}{1-\gamma} \cdot \ln \left(\frac{2  |\cS_g||\cA_g| |\cS_l|}{\delta_1}\right)\cdot \min\left\{|\cS_l|^2  k^{|\cS_l|} ,  k |\cS_l|^k \right\},
\end{align*}

that the global agent learns an $\epsilon_g(k)\coloneqq \sup_{s\in \cS}(V^{\pi^*}(s) - V^{\pi_{k,m}}(s))$ best-response, where
\[\epsilon_g(k) = \frac{2\tilde{r}}{(1-\gamma)^2}\left(\sqrt{\frac{n-k+1}{2nk} \ln(2|\mathcal{S}_l||\mathcal{A}_g|\sqrt{k})} +  \frac{1}{\sqrt{k}}\right) + \frac{4}{\sqrt{k}(1-\gamma)}.\]
Similarly, with probability at least $1-\delta_2$, we have from Theorem \ref{thm: local agent optimality and sample complexity} that the local agent learns an $\epsilon_\ell=1/\sqrt{k}$ best-response with $N_{\delta_2}$ samples, where
\begin{align*}N_{\delta_2} &= \tilde{O}\left(\min\left\{\frac{|\cS_l|^5|\cS_g|^2 k^{2|\cS_l|+1}|\cA_l|}{(1-\gamma)^2} \ln \frac1{\delta_2}, \frac{k^6  |\cS_g|^2 |\cS_l|^k |\cA_l|}{(1-\gamma)^3} \ln \frac1{\delta_2}\right\}\right).\\
\end{align*}
Then, with probability at least $1 - (\delta_1+\delta_2)T$, we have that the agent dynamics reaches a $2\eta$-NE, where 
\begin{align*}\eta &\coloneqq \max\{\epsilon_g(k), \epsilon_\ell\}\\
&\leq  \frac{2\tilde{r}}{(1-\gamma)^2}\left(\sqrt{\frac{n-k+1}{2nk} \ln(2|\mathcal{S}_l||\mathcal{A}_g|\sqrt{k})} +  \frac{1}{\sqrt{k}}\right) + \frac{4}{\sqrt{k}(1-\gamma)} + \frac{1}{\sqrt{k}} \leq \tilde{O}\left(\frac{1}{\sqrt{k}}\right).
\end{align*}
Then, via Lemma \ref{lemma: bound nsteps}, we have that 
\[ N_{steps}(k) \leq 
\frac{\tilde{r}}{\eta(1-\gamma)}
 \leq \tilde{O}\left(\frac{\tilde{r}\sqrt{k}}{1-\gamma}\right).\]

If we seek a success probability of at least $1-\delta$, we can set $\delta_1 = \delta_2 = \delta/2N_{steps}$, and the result would follow by taking union-bounds.\\

For  the global agent, this gives a sample complexity of
\begin{align*}M_{\delta/2N_{steps}} &= \frac{2k^3\gamma^2\|r_l\|^2_\infty |\cS_g| |\cA_g|}{(1-\gamma)^5}\log \frac{\tilde{r}\sqrt{k}}{1-\gamma} \cdot \ln \left(\frac{2  |\cS_g||\cA_g| |\cS_l|\tilde{r}\sqrt{k}}{\delta(1-\gamma)}\right)\!\cdot\! \min\left\{|\cS_l|^2  k^{|\cS_l|} ,  k |\cS_l|^k \right\},
\end{align*}
For the local agent this gives a sample complexity of 
\[N_{\delta/2N_{steps}} = \min\left\{\frac{|\cS_l|^5|\cS_g|^2 k^{2|\cS_l|+1}|\cA_l|}{(1-\gamma)^2} \ln \frac{\tilde{r}\sqrt{k}}{\delta(1-\gamma)}, \frac{k^6  |\cS_g|^2 |\cS_l|^k |\cA_l|}{(1-\gamma)^3} \ln \frac{\tilde{r}\sqrt{k}}{\delta(1-\gamma)}\right\}.\]
Putting them together for $N_{steps}$, we have that the sample complexity of learning a $\tilde{O}(1/\sqrt{k})$-NE with probability at least $1-\delta$ is at most:
\begin{itemize}\item  In the mean-field regime:
\begin{align*}
    &N_{steps} \cdot (M_{\delta/2N_{steps}} + N_{\delta/2N_{steps}}) \\
    &\quad \leq \tilde{O}\left(\frac{k^{3.5}\gamma^2\|r_l\|^3_\infty |\cS_g| |\cA_g| |\cS_l|^2 k^{|\mathcal S_l|}}{(1-\gamma)^6}\log \frac1\delta + \frac{\sqrt{k}|\cS_l|^5|\cS_g|^2 k^{2|\cS_l|+1}|\cA_l| \|r_l\|_\infty}{(1-\gamma)^3} \log\frac1\delta \right) \\
    &\quad \leq \tilde{O}\left(\frac{k^{3.5} \|r_l\|^3_\infty |\mathcal S_g|^2 |\mathcal A_l| |\mathcal A_g| |\mathcal S_l|^5  k^{2|\mathcal S_l|+1}}{(1-\gamma)^6} \log\frac1\delta\right).
\end{align*}
\item In the standard-parameterization:
\begin{align*}
 &N_{steps} \cdot (M_{\delta/2N_{steps}} + N_{\delta/2N_{steps}}) \\
    &\quad \leq \tilde{O}\left(\frac{k^{4}\gamma^2\|r_l\|^3_\infty |\cS_g| |\cA_g| |\cS_l|^k }{(1-\gamma)^6}\log \frac1\delta +  \frac{k^{6.5}  |\cS_g|^2 |\cS_l|^k |\cA_l| \|r_l\|_\infty }{(1-\gamma)^4} \log\frac1\delta\right) \\
    &\quad \leq \tilde{O}\left(\frac{k^{6.5} \|r_l\|^3_\infty |\cS_g|^2 |\mathcal A_l| |\cA_g| |\cS_l|^k }{(1-\gamma)^6}\log \frac1\delta\right)
\end{align*}
\end{itemize}

Together, we have that with probability at least $1-\delta$, \texttt{ALTERNATING-MARL} learns a $\tilde{O}(1/\sqrt{k})$-NE in sample complexity
\[\tilde{O}\left(\min\left\{\frac{  \|r_l\|^3_\infty |\mathcal S_g|^2 |\mathcal A_l| |\mathcal A_g| |\mathcal S_l|^5  k^{2|\mathcal S_l|+4.5}}{(1-\gamma)^6} \log\frac1\delta, \frac{k^{6.5} \|r_l\|^3_\infty |\cS_g|^2 |\mathcal A_l| |\cA_g| |\cS_l|^k }{(1-\gamma)^6}\log \frac1\delta\right\}\right),\]
which proves the theorem.\qedhere
\end{proof}

\subsection{Comparison with Mean-Field Subsampling for Optimality Guarantees}

In this subsection, we provide a comparison of the sample complexities and types of theoretical guarantees between our work and the work in \cite{anand2025meanfield} in Table \ref{tab:comparison table}. While \cite{anand2025meanfield} provides guarantees for the stronger notion of $\epsilon$-approximate optimality of the learned policy (where $\epsilon=\tilde{O}(1/\sqrt{k})$), it does so with a restrictive cost on the size of the local agent's action space $\mathcal A_l$, which can be very large in practice. By focusing on $\epsilon$-approximate Nash Equilibria, Our work removes the expensive dependence on $|\mathcal A_l|$. Secondly, \cite{anand2025meanfield} requires full information: each agent needs to see every other agent; in contrast, the algorithmic framework in \texttt{ALTERNATING-MARL} only needs that the global agent should see $k$ local agents, while each local agent only needs to see the global agent's state. Finally, we note that in a number of settings, it is possible that a Nash Equilibrium solution might also be an optimum solution (under certain global stability regularity assumptions). In these settings, \texttt{ALTERNATING-MARL} provides a strict improvement over the mean-field sampling framework of  \cite{anand2025meanfield}.

\begin{table}[hbt!]
    \centering
    \begin{adjustbox}{max width=\textwidth}
    \begin{tabular}{lccc}
    \toprule
      & Mean-Field Sampling Work \cite{anand2025meanfield} & Our work 
      (\texttt{ALTERNATING-MARL}) \\
    \toprule
    \textbf{Learning Guarantees} &  {$\epsilon$-approximate Optimality} &  {$\epsilon$-approximate Nash Equilibria}\\
    & & &\\
    \textbf{Communication Restrictions} & Requires full information & Global agent only needs to see $k$ local agents\\  
    & & Local agents only need to see the global agent's state\\
    & & &\\
    \textbf{Sample Complexities} & $\tilde{O}(|\mathcal S_g||\mathcal A_g||\mathcal S_l|^k |\mathcal A_l|^k)$ & $\tilde{O}(k^{6.5} |\mathcal S_g|^2 |\mathcal A_g| |\mathcal A_l| |\mathcal S_l|^k)$\\
    & & &\\
    \textbf{Mean-Field Sample Complexities} & $\tilde{O}(|\mathcal S_g||\mathcal S_l||\mathcal A_g||\mathcal A_l| k^{|\mathcal S_l| |\mathcal A_l|})$ & $\tilde{O}(|\mathcal S_g|^2 |\mathcal A_g| |\mathcal A_l| |\mathcal S_l|^5 k^{2|\mathcal S_l|+4.5})$\\
    & & &\\
    \textbf{Error bound} & $\tilde{O}(1/\sqrt{k})$ &$\tilde{O}(1/\sqrt{k})$ \\
     \smallskip
    \end{tabular}
    \end{adjustbox}
    \caption{Comparing \texttt{ALTERNATING-MARL} with the mean-field sampling framework of \cite{anand2025meanfield}}
    \label{tab:comparison table}
\end{table}

\section{Extension to Off-Policy Learning}
\label{Appendix/off-policy}

As in \cite{anand2025meanfield}, this section extends the \texttt{G-LEARN} framework to off-policy learning. A limitation of the planning in Algorithm \labelcref{algorithm: g learn} is that it learns $\hat{Q}_k^*$ in an \emph{offline} manner by assuming a generative oracle access to the transition functions ${P}_g, {P}_l$, and reward function $r$. In certain realistic RL applications, such a generative oracle might not exist, and it is more desirable to perform off-policy learning where the agent continues to learn in an offline manner but from \emph{historical data} \cite{fujimoto2019off}. In this setting, the agents learn the target policy $\hat{\pi}_\kappa^*$ using data generated by a different behavior policy (the strategy it uses to explore the environment). There is a significant body of work on the theoretical guarantees in off-policy learning \cite{chen2021lyapunov,pmlr-v151-chen22i,chen2021finite,chen2025concentration,anand_et_al:LIPIcs.ICALP.2024.10}. In fact, these previous results are amenable to transforming guarantees in offline $Q$-learning to off-policy learning, at the cost of $\log |\cS| |\cA|$ factors in the runtime. 
 
 Thus, this section is devoted to showing that our result in \cref{theorem: main result} satisfies the assumptions of transforming offline $Q$-learning to off-policy $Q$-learning for the subsampled $\hat{Q}_k$-function. We further show, in expectation, how to maintain the decaying optimality gap of $\tilde{O}(1/\sqrt{k})$ of the learned policy $\pi_{k}$, when the randomness is over the exploration policy $\pi_b$.  The iterative off-policy $\hat{Q}_k$-learning algorithm estimates the optimal $\hat{Q}_k$-function as follows: first, a sample trajectory $\{(s_g, s_\Delta, a_g)\}$ is collected using a suitable behavior policy $\pi_{k,b}$. Then, initialize $\hat{Q}_k^{0}: \cS_g\times\cS_l^k \times  \cA_g \to\mathbb{R}$ and let $\alpha>0$ be determined later. For each $t\geq 0$  the iterate $\hat{Q}_k^{t}(s_g, s_\Delta, a_g)$ is updated by
\begin{equation}\label{eqn: off-policy}\hat{Q}_k^{t+1}(s_g, s_\Delta, a_g) = (1-\alpha) \hat{Q}_k^{t}(s_g, s_\Delta, a_g) + \alpha  \left(r(s, a) + \gamma \max_{a_g'} \hat{Q}_k^{t} (s'_g, s'_\Delta, a'_g) \right).\end{equation}
Note that the update in \cref{eqn: off-policy} does not include an expectation and can be computed in a single trajectory via historical data. We make the following ergodicity assumption:\\

\begin{assumption}\label{assumption: ergodic}
    The behavior policy $\pi_b$ satisfies $\pi_b(a|s_g, s_\Delta)>0$ for all $(s_g, s_\Delta, a_g)\in \cS_g \times \cS_l^k \times\cA_g$ and the Markov chain $\cM_{s} = \{s_t\}_{t\geq 0}$ induced by $\pi_b$ is irreducible and aperiodic with stationary distribution $\mu$ and mixing time \[t_\delta(\cM) = \min \left\{t\geq 0: \max_{S_\Delta \in \cS_g\times\cS_l^k} \|P^t(s_g,s_\Delta,\cdot) - \mu(\cdot)\|_{TV} \leq \delta\right\}.\] There are a number of heuristic behavior/exploration  policies that satisfy this assumption \citep{fujimoto2019off}.\\
\end{assumption}

\begin{theorem}
    Let $\pi_\kappa$ be the policy learned through off-policy $\hat{Q}_k$-learning. Then, under \cref{assumption: ergodic}, we have that with probability at least $1 - \frac{1}{100e^k}$,
\begin{align*}
\E[V^{\pi^*}(s_0) - V^{\pi_{k}}(s_0)] &\leq \frac{\tilde{r}}{(1-\gamma)^2}\sqrt{\frac{n-k+1}{2n k}} \sqrt{\ln\frac{40\tilde{r}|\mathcal{S}|_g |\cS_l| |\mathcal{A}_g| k^{|\mathcal{S}_l|+\frac{1}{2}}}{(1-\gamma)^2}}  + \frac{1}{10\sqrt{k}} \\
&= \tilde{O}(1/\sqrt{k}),\end{align*}
where the randomness in the expectation is over the stochasticity of the exploration policy $\pi_b$.\\
\end{theorem}

\begin{proposition}\label{fake assumptions}
    Recall that the following are true:
    
    \begin{enumerate}
        \item $\|\hat{Q}_k(s_g, s_\Delta, a_g) - \hat{Q}_k(s_g, s_\Delta', a_g)\| \leq \frac{1}{1-\gamma}\|r_l(\cdot,\cdot)\|_\infty \cdot \mathrm{TV}(F_{s_\Delta}, F_{s_\Delta'})$, for any $(s_g,s_\Delta,a_g) \in \cS_g\times\cS_l^k \times\cA_g$ by Lemma \ref{lemma: Q lipschitz continuous wrt Fsdelta},
        \item $\|\hat{Q}_k\| \leq \frac{\tilde{r}}{1-\gamma}$ by Lemma \ref{lemma: Q-bound},
        \item $\|\hat{Q}^{t+1}_k(s_g, s_\Delta, a_g) - \tilde{Q}_k^{t+1}(s_g, s_\Delta, a_g)\|_\infty \leq \gamma\|\hat{Q}_k^t\ - \tilde{Q}_k^t \|_\infty$ by Lemma \ref{lemma: gamma-contraction of adapted Bellman operator},
        \item The Markov chain $\cM_S$ enjoys a rapid mixing property from \cref{assumption: ergodic},
    \end{enumerate}
\end{proposition}
Treating the single trajectory update of the $\hat{Q}_\kappa$-function as a noisy addition to the expected update from the ideal Bellman operator, \cite{chen2021lyapunov} uses Markovian stochastic approximation to bound $\E[\|\hat{Q}_\kappa^T - \hat{Q}_\kappa^*\|_\infty^2]$. We restate their result:\\

\begin{theorem}[Theorem 3.1 in \cite{chen2021lyapunov} adapted to our setting] Suppose $\alpha_t = \alpha$ for all $t\geq 0$, where $\alpha$ is chosen such that $\alpha t_\alpha(\cM_{s}) \leq c_{Q,0} \frac{(1-\beta_1)^2}{\log |\cS||\cA|}$, where $c_{Q,0}$ is a numerical constant. Then, under Propositions \ref{fake assumptions}, for all $t\geq t_\alpha(\cM_{s})$, we have
\begin{equation}\E[\|\hat{Q}^t_k - \hat{Q}_k^*\|^2_\infty] \leq c_{Q,1}\left(1 - \frac{(1-\gamma)\alpha}{2}\right)^{t-t_\alpha(\cM_{s})} + c_{Q,2} \frac{\log k^{|\cS| |\cA|}}{(1-\gamma)^2}\alpha t_\alpha(\cM_{s}),\end{equation}
    where $c_{Q,1} = 3(\frac{\tilde{r}}{1-\gamma} + 1)^2$, $c_{Q,2} = 912e(\frac{3\tilde{r}}{1-\gamma} + 1)^2$, and where the randomness in the expectation is over the randomness of the stochasticity of the behavior/exploration policy $\pi_b$.\\
\end{theorem}

\begin{corollary}[Corollary 3.2 in \cite{chen2021lyapunov} adapted to our setting] To make $\E[\|\hat{Q}_k^t - \hat{Q}_k^*\|_\infty \leq \frac{1}{100\sqrt{k}}]$ for $\epsilon>0$, we need 
\begin{equation}t>\tilde{O}\left(\frac{ k |\cS|_g |\cS_l| |\cA_g| k^{|\cS_l|} \log^2 k}{(1-\gamma)^5}\right).
\end{equation}
\end{corollary}
With this sample complexity, by the triangle inequality we also recover an expected-value analog of \cref{theorem: g-learn}.\\

\begin{corollary}For $\delta\in(0,1)^2$, with probability at least $1-\delta$, we have
    \begin{equation}\E[\hat{Q}_k^*(s_g, s_\Delta, a_g) - Q_n^*(s_g, s_{[n]}, a_g)] \leq \frac{\ln \frac{2|\cS_l|}{\delta}}{1-\gamma}\sqrt{\frac{n-k+1}{8k n}}\|r_l\|_\infty,\end{equation}
    where the randomness in the expectation is over the stochasticity of the exploration policy $\pi_b$.\\
\end{corollary}
In turn, following the argument in the proof of \cref{theorem: g-learn}, it is straightforward to verify that this leads to a result on the expected performance difference using off-policy learning:
\begin{corollary}\label{offpolicy-result} With probability at least $1 - \frac{1}{100e^\kappa}$,
\begin{equation}\E[V^{\pi^*}(s_0) - V^{\pi_{k}}(s_0)] \leq \frac{\tilde{r}}{(1-\gamma)^2}\sqrt{\frac{n-k+1}{2nk}} \sqrt{\ln\frac{40\tilde{r}|\mathcal{S}|_g |\cS_l|^k |\cA_g|}{(1-\gamma)^2}}  + \frac{1}{10\sqrt{k}}  = \tilde{O}\left(\frac{1}{\sqrt{k}}\right),\end{equation}
where the randomness in the expectation is over the stochasticity of the exploration policy $\pi_b$.
\end{corollary}

\section{Generalization to Stochastic Rewards}
\label{sec: generalize to stochastic rewards}

As in \cite{anand2025meanfield}, this section extends the \texttt{ALTERNATING-MARL} framework to handle environments where rewards are stochastic. While our primary analyses is for deterministic local rewards, many real-world multi-agent systems can be drawn from a probability distribution. First, note that the local-agent policy-learning algorithm can a priori handle stochastic rewards since it is a nonstationary MDP. Therefore, it suffices to handle the case for the global agent. We begin by fixing a local agent policy $\pi_\ell$, thereby giving that the local agent reward component has form $r_l(s_i, s_g, a_i) = r_l(s_i, s_g)$.

\label{stochastic generalization}

\paragraph{Reward distributions.} Suppose we are given two families of distributions, $\{\mathcal{G}_{s_g,a_g}\}_{s_g,a_g\in\mathcal{S}_g\times\mathcal{A}_g}$ and $\{\mathcal{L}_{s_i,s_g}\}_{s_i,s_g\in\mathcal{S}_l\times\mathcal{S}_g}$. Then, for $s\in\cS_g\times\cS_l^n$ and $a\in\cA_g$, let $R(s,a)$ denote a stochastic reward of the form $R(s,a) = r_g(s_g,a_g) + \frac1n  \sum_{i\in[n]}r_l(s_i,s_g)$
where the rewards of the global agent $r_g$ emerge from a distribution $r_g(s_g,a_g)\sim \mathcal{G}_{s_g,a_g}$, and the rewards of the local agents $r_l$ emerge from a distribution $r_l(s_i,s_g)\sim \mathcal{L}_{s_i,s_g}$. Similarly, for $\Delta\subseteq [n]$, let $R_\Delta(s,a)$ be defined as:
\begin{equation}R_\Delta(s,a) = r_g(s_g,a_g) + \frac1k \sum_{i\in\Delta}r_l(s_i,s_g ).\end{equation}
We make some standard assumptions of boundedness on $\mathcal{G}_{s_g,a_g}$ and $\mathcal{L}_{s_i,s_g}$.\\

\begin{assumption}\label{assumption: stochastic_bound}
    \emph{Define} \[\bar{\mathcal{G}}=\bigcup_{(s_g,a_g)\in\mathcal{S}_g\times\mathcal{A}_g}\mathrm{supp}\left(\mathcal{G}_{s_g,a_g}\right)\] and \[\bar{\mathcal{L}}=\bigcup_{(s_i,s_g)\in \mathcal{S}_l\times\mathcal{S}_g}\mathrm{supp}\left(\mathcal{L}_{s_i,s_g}\right),\] \emph{where for any distribution $\mathcal{D}$, $\mathrm{supp}(\mathcal{D})$ is the support (set of all random variables $\mathcal{D}$ that can be sampled with probability strictly larger than $0$) of $\mathcal{D}$.
    Then, let \[\hat{\mathcal{G}} = \sup\left(\bar{\mathcal{G}}\right), \hat{\mathcal{L}} = \sup\left(\bar{\mathcal{L}}\right), \widecheck{\mathcal{G}} = \inf\left(\bar{\mathcal{G}}\right),\widecheck{\mathcal{L}} = \inf\left(\bar{\mathcal{L}}\right).\]
    We assume that \[\hat{\mathcal{G}}<\infty, \hat{\mathcal{L}}<\infty, \widecheck{\mathcal{G}}>-\infty, \widecheck{\mathcal{L}}>-\infty,\] and that $\hat{\mathcal{G}},\hat{\mathcal{L}},\widecheck{\mathcal{G}},\widecheck{\mathcal{L}}$ are all known a priori.}\\
\end{assumption}

\begin{definition} Let the randomized empirical adapted Bellman operator be $\hat{\mathcal{T}}^{\text{random}}_{k,m}$ such that:
    \begin{equation}\hat{\mathcal{T}}^{\text{random}}_{k,m} \hat{Q}^t_{k,m}(s_g, s_\Delta, a_g) = R_\Delta(s,a) + \frac{\gamma}{m} \sum_{\ell \in [m]}\max_{a_g'\in\mathcal{A}_g} \hat{Q}_{k,m}^t(s_g'
    , s_\Delta', a_g'),\end{equation}

\end{definition}
\paragraph{\texttt{ALTERNATING-MARL}: Learning with Stochastic Rewards.} Following \cite{anand2025the}, our proposed extension of \texttt{ALTERNATING-MARL} averages $\Xi$ samples of the adapted randomized empirical adapted Bellman operator, $\mathcal{T}_{k,m}^{\text{random}}$ and updates the $\hat{Q}_{k,m}$ function using with the average. One can show that $\mathcal{T}_{k,m}^{\text{random}}$ is a contraction operator with modulus $\gamma$.

By Banach's fixed point theorem, $\mathcal{T}_{k,m}^{\text{random}}$ admits a unique fixed point $\hat{Q}_{k,m}^{\text{random}}$.\\

\begin{algorithm}[H]
\caption{\texttt{ALTERNATING-MARL}: Learning with Stochastic Rewards}\label{algorithm: approx-dense-tolerable-Q-learning_random_rewards}
\begin{algorithmic}[1]
\REQUIRE A multi-agent system as described in \cref{sec: preliminaries}.
\REQUIRE Parameter $T$ for the number of iterations in the initial value iteration step. 
\REQUIRE Sampling parameters $k \in [n]$ and $m\in \N$. Discount parameter $\gamma\in (0,1)$. 
\REQUIRE Oracle $\mathcal{O}$ to sample $s_g'\sim {P}_g(\cdot|s_g,a_g)$ and $s_i\sim {P}_l(\cdot|s_i,s_g)$.
\REQUIRE Number of reward samples $\Xi\in\N$.
\medskip
\STATE Set $\hat{Q}^0_{k,m}(s_g,s_\Delta, a_g)=0$, for $(s_g,s_\Delta,a_g)\in\cS_g\times\cS_l^k\times\cA_g$.
\medskip
\FOR{$t=0$ to $T-1$}
\FOR{$(s_g,s_\Delta,a_g)\in\cS_g\times\cS_l^k\times\cA_g$}
\STATE $\rho=0$
\smallskip
\FOR{$\xi\in\{1,\dots,\Xi\}$}
\STATE $\rho = \rho + \hat{\mathcal{T}}^{\text{random}}_{k,m} \hat{Q}^t_{k,m}(s_g, s_\Delta, a_g)$
\ENDFOR
\STATE $\hat{Q}^{t+1}_{k,m}(s_g, s_\Delta, a_g) = \frac\rho\Xi$
\ENDFOR
\ENDFOR
\medskip
\STATE For all $(s_g,s_\Delta) \in \mathcal{S}_g\times \mathcal{S}_l^k$, let \[\hat{\pi}_{k,m}^\est(s_g,s_\Delta) = \mathop{\argmax}_{a_g\in\mathcal{A}_g}\hat{Q}_{k,m}^T(s_g, s_\Delta, a_g).\]
\end{algorithmic}
\end{algorithm} 

To bound the error introduced by the stochasticity of the rewards, we use Hoeffding's inequality

\begin{theorem}[Hoeffding's Theorem \citep{10.5555/1522486}] Let $X_1,\dots,X_n$ be independent random variables such that $a_i\leq X_i\leq b_i$ almost surely. Then, let $S_n = X_1 + \dots + X_n$. Then, for all $\epsilon>0$,
\begin{equation}\Pr[|S_n - \E[S_n]|\geq \epsilon] \leq 2\exp\left(-\frac{2\epsilon^2}{\sum_{i=1}^n (b_i-a_i)^2}\right).\end{equation}
\end{theorem}
\begin{lemma}When ${\pi}^\est_{k,m}$ is derived from the randomized empirical value iteration operation and applying our online subsampling execution in \cref{algorithm: online execution}, we get\label{PDL: stochastic}
    \begin{equation}\Pr\left[V^{\pi^*}(s_0) - V^{\tilde{\pi}_{k,m}}(s_0) \leq \tilde{O}\left(\frac{1}{\sqrt{k}}\right)\right]\geq 1-\frac{1}{100\sqrt{k}}.\end{equation}
    \end{lemma}
    \begin{proof}
    \begin{align*}
        \Pr\left[\left|\frac{\rho}{\Xi} - \E[R_\Delta(s,a)]\right|\geq \frac{\epsilon}{\Xi}\right] &\leq 2\exp\left(-\frac{2\epsilon^2}{\sum_{i=1}^\Xi |\hat{\mathcal{G}}+\hat{\mathcal{L}} - \widecheck{\mathcal{G}}-\widecheck{\mathcal{L}}|^2}\right) \\ &= 2\exp\left(-\frac{2\epsilon^2}{\Xi|\hat{\mathcal{G}}+\hat{\mathcal{L}} - \widecheck{\mathcal{G}}-\widecheck{\mathcal{L}}|^2}\right)
        \end{align*}
    \emph{Rearranging this, we get:}
    \begin{equation}\Pr\left[\left|\frac{\rho}{\Xi} - \E[R_\Delta(s,a)]\right|\leq \sqrt{\ln\left(\frac{2}{\delta}\right) \frac{ {|\hat{\mathcal{G}}+\hat{\mathcal{L}}-\widecheck{\mathcal{G}}-\widecheck{\mathcal{L}}|^2}}{2\Xi}}\right]\geq 1 - \delta\end{equation}
Then, setting $\delta=\frac{1}{100\sqrt{k}}$, and setting $\Xi =10|\hat{\mathcal{G}}+\hat{\mathcal{L}}-\widecheck{\mathcal{G}}-\widecheck{\mathcal{L}}|^2 k\sqrt{\ln(200\sqrt{k})}$ gives:
\begin{equation*}\Pr\left[\left|\frac{\rho}{\Xi} - \E[R_\Delta(s,a)]\right|\leq \frac{1}{\sqrt{200 k}} \right]\geq 1 - \frac{1}{100\sqrt{k}}\end{equation*}
Applying the triangle inequality to $\epsilon_{k,m}$ allows us to derive a probabilistic bound on the optimality gap between $\hat{Q}_{k,m}^\est$ and $Q^*$, where the gap is increased by $\frac{1}{\sqrt{200k}}$, and where the randomness is over the stochasticity of the rewards. \\

Then, the optimality gap between $V^{\pi^*}$ and $V^{{\pi}_{k,m}^\est}$, for when the policy $\hat{\pi}^\est_{k,m}$ is learned using \cref{algorithm: approx-dense-tolerable-Q-learning_random_rewards} in the presence of stochastic rewards obeying \cref{assumption: stochastic_bound} follows
\begin{equation}
    \Pr\left[V^{\pi^*}(s_0) - V^{\tilde{\pi}_{k,m}}(s_0) \leq \tilde{O}\left(\frac{1}{\sqrt{k}}\right)\right]\geq 1-\frac{1}{100\sqrt{k}},
\end{equation}
which proves the lemma.\qedhere\\
\end{proof}

\section{Numerical Simulations}
This section describes our numerical simulations which scale to $n=1000$ agents and validate our theoretical framework.
In this section, we provide a formulation of \texttt{ALTERNATING-MARL} within the context of communication-constrained control. We then provide an empirical evaluation of the algorithm in a robotic coordination task with $n=1000$ agents. We note that the simulations are not intended to instantiate the conservative PAC sample-complexity constants: rather, they test the qualitative tradeoffs predicted by the theorem that increasing $k$ improved performance while increasing runtime.

\textbf{Compute.} All experiments were implemented in Python and ran on a 2-core CPU server with 12GB RAM, and all the experiments ran in $~\sim$ 6 hours. We chose a parameter complexity for each simulation that was sufficient to emphasize
characteristics of the theory, such as the complexity improvement and the decaying optimality gap. Finally, we provide and evaluate some heuristic solutions for extending our algorithm to the setting with continuous state spaces. \footnote{We provide code for the experiments in \url{https://anonymous.4open.science/r/alternating-marl-5C86/}.}. 

\label{sec: numerical simulations}

\subsection{Alternating MARL in a Communication-Constrained Robotic Coordination Task}

\paragraph{Experimental setup.} Consider a warehouse with $n$ mobile robots (local agents) operating across $|\cS_l|$ zones (local states). A central dispatcher (global agent) manages the warehouse's operational mode (global state), i.e. which zone gets priority resource allocation (charging stations, loading docks, clear pathways).  The dispatcher's problem is that it cannot poll the full population of robots for their current zone before deciding where to allocate resources. Therefore, it must sample states of local agents. With $k=1$, it gets a single noisy snapshot, and with $k=35$, it gets a reliable picture of where the fleet is concentrated. The local agents (robots) have their own decision-making too. Each robot chooses between $|\cA_l|$ actions: stay in its current zone, move up to the next zone, or move down to the previous zone. Specifically, each robot decides whether to stay put or migrate toward the prioritized zone; however, moving has a cost since the robot is less productive in transit, and it might arrive after the dispatcher has already shifted priority elsewhere. 
 Indeed, the robots are cooperative but decentralized. Each robot only sees its own zone $s_\ell$ and the current system priority $s_g$, but not the positions of other robots. Therefore, they must rely on the dispatcher to choose good priorities and implicitly coordinate through the shared global state.
 Each evaluation run is conducted over a horizon of $H=100$ time steps.

\paragraph{Dynamics and Rewards.} We let $\mathcal{S}_g=\{0,1,2,3,4\}, \mathcal{S}_l=\{0,1,2,3,4\}, \mathcal{A}_g=\{0,1,2,3,4\}$, and $\cA_l=\{0, 1, 2\}$. We now define the transition dynamics of the agents. First, let 
$d_{s_g',s_g,a_g} = 3\cdot \mathbbm{1}\{s_g'=a_g\} + 0.5 \cdot \mathbbm{1}\{s_g' = s_g\} + 0.1$.
Then, $P_g(s_g'|s_g, a_g)$ is given by the softmax over the logits, 
\[P_g(s_g' \mid s_g, a_g) = \frac{e^{d_{s_g',s_g,a_g}}}{\sum_{j\in\mathcal{S}_g} e^{d_{j,s_g,a_g}}},\]
which allows action $a_g$ to drive the global agent's state towards state $a_g$. 

Next, let
$f_{s_\ell',s_\ell,s_g,a_\ell} = 3.5 \cdot \mathbbm{1}\{s_\ell' = s_\ell\} + 1.5 \cdot \mathbbm{1}\{s_\ell' = \mathsf{target}(s_\ell, a_\ell)\} + 0.3 \cdot \mathbbm{1}\{s_\ell' = s_g\} + 0.1$, where $\mathsf{target}(s_\ell, 1) = (s_\ell + 1) \bmod 5$ and $\mathsf{target}(s_\ell, 2) = (s_\ell - 1) \bmod 5$.  Then, similarly, $P_l(s_l'|s_l,s_g,a_l)$ is given by the softmax over the logits \[P_l(s_l' \mid s_l, s_g, a_l) = \frac{e^{f_{s_l',s_l, s_g, a_l}}}{\sum_{j\in\mathcal{S}_l} e^{f_{j,s_l,s_g, a_l}}}.\]
We define the stage reward
\begin{equation}r_t = r_g(s_g, a_g) + \frac{1}{n}\sum_{i=1}^n r_l(s_\ell^i, s_g, a_\ell^i),\end{equation}
where: 
$r_g(s_g, a_g) = 4 - |s_g - a_g|$ and
$r_l(s_\ell, s_g, a_\ell) = \text{base}(s_\ell, s_g) + \text{bonus}(a_\ell)$, where
$$\text{base}(s_\ell, s_g) = \begin{cases} 10.0 &  s_\ell = s_g \\ 
\frac{2.0}{1 + d_{\text{circ}}(s_\ell, s_g)} & \text{otherwise} \end{cases},$$
and $\text{bonus}(a_\ell) \in \{0.0, +0.5, -0.3\}$, where
 $d_{\text{circ}}(i,j) = \min(|i-j|, 5-|i-j|)$. The local reward peaks at alignment ($s_\ell = s_g$), and hence the global agent maximizes the stage reward by steering $s_g$ to the mode of the population distribution.

\paragraph{Algorithm Configuration.} The hyperparameters used in our implementation are detailed in Table~\ref{tab:hyperparams}.

\begin{table}[hbt!]
\centering
\begin{tabular}{@{}ll@{}}
\toprule
Hyperparameter & Value \\ \midrule
Population size ($n$) & 1000\\
Discount Factor ($\gamma$) & 0.95 \\
Global agent state space size ($|\cS_g|$) & 5 \\
Local agent state space size ($|\cS_l|$) & 5 \\
Global agent action space size ($|\cA_g|$) & 5 \\
Local agent action space size ($|\cA_l|$) & 3 \\
Subsample Sizes ($k$) & $\{1,2,\dots, 50\}$ \\
Monte Carlo Samples ($m$) & 30 (per Bellman update) \\
Learning Rate ($\alpha$) & 1.0 (Exact Operator Update) \\
Exploration Rate ($\epsilon$) & 0.0 (Exhaustive Offline Sweep) \\
Number of alternations ($N_{\text{steps}}$) & 10 \\
Independent Runs & 15 seeds per $k$ \\

Number of rollouts & 50 per seed \\

Horizon & 100 \\\bottomrule\\
\end{tabular}
\caption{\texttt{ALTERNATING-MARL} Training and Evaluation Configuration}
\label{tab:hyperparams}
\end{table}

\paragraph{Experimental Results.} 
We plot our results in \cref{fig:k increase}, which shows the discounted cumulative rewards for each $k$ and the associated runtime for learning the $\tilde{O}(1/\sqrt{k})$-approximate Nash Equilibrium, and \cref{fig:sim result} which plots the  fraction of $n=1000$ robots occupying each of the $5$ zones, and the discrepancy between the global agent's action and the true population mode. We see that as $k$ increases, the quality of the learned policy, as measured by the reward observed by the system, generally increases, albeit with a higher required sampled complexity, further underscoring the tradeoff in the sampling parameter $k$ as mentioned in \cref{remark: tradeoff in k}.\looseness=-1

 \begin{figure}[t]
     \centering    \includegraphics[width=0.494\linewidth]{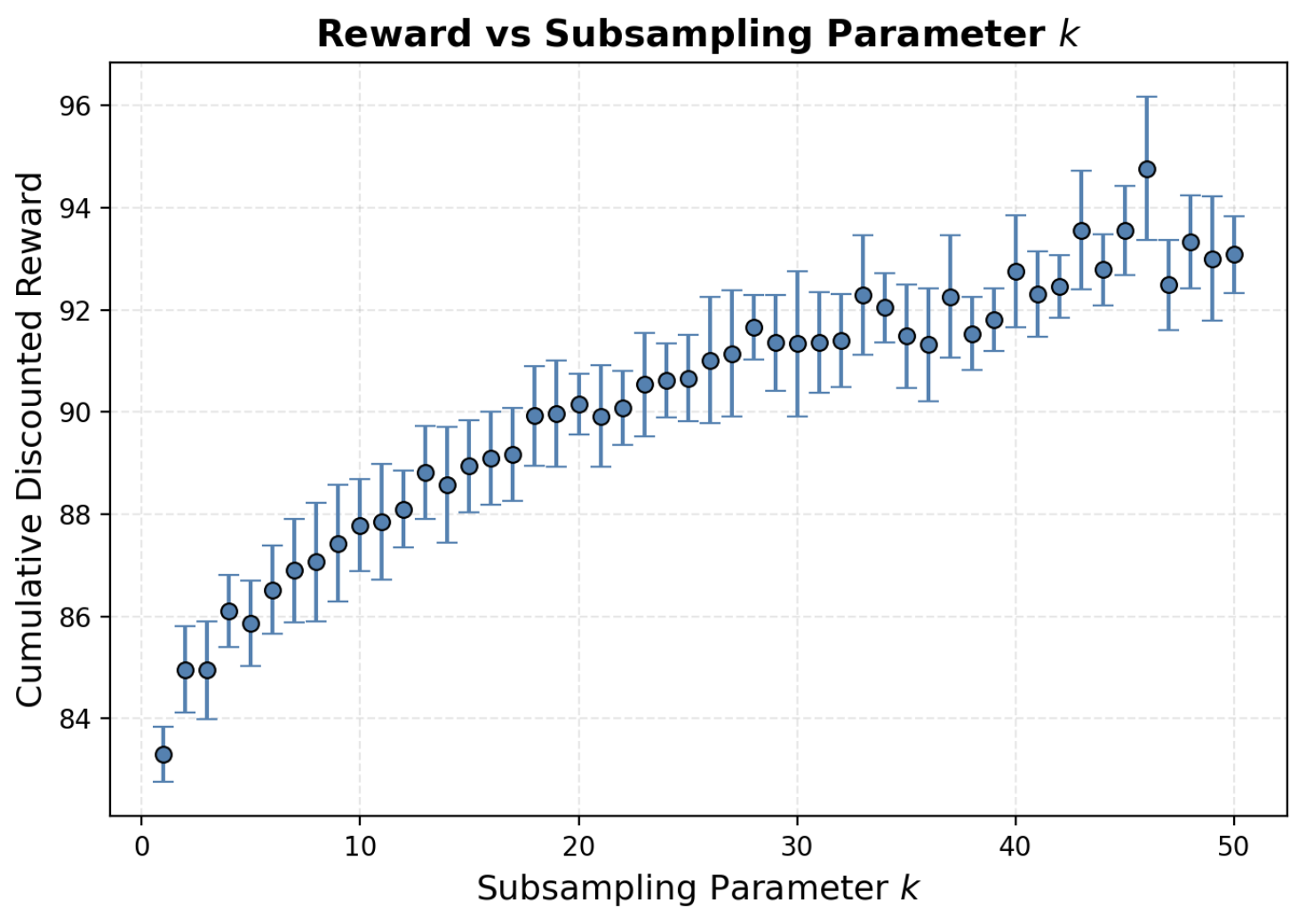}     
\includegraphics[width=0.5\linewidth]{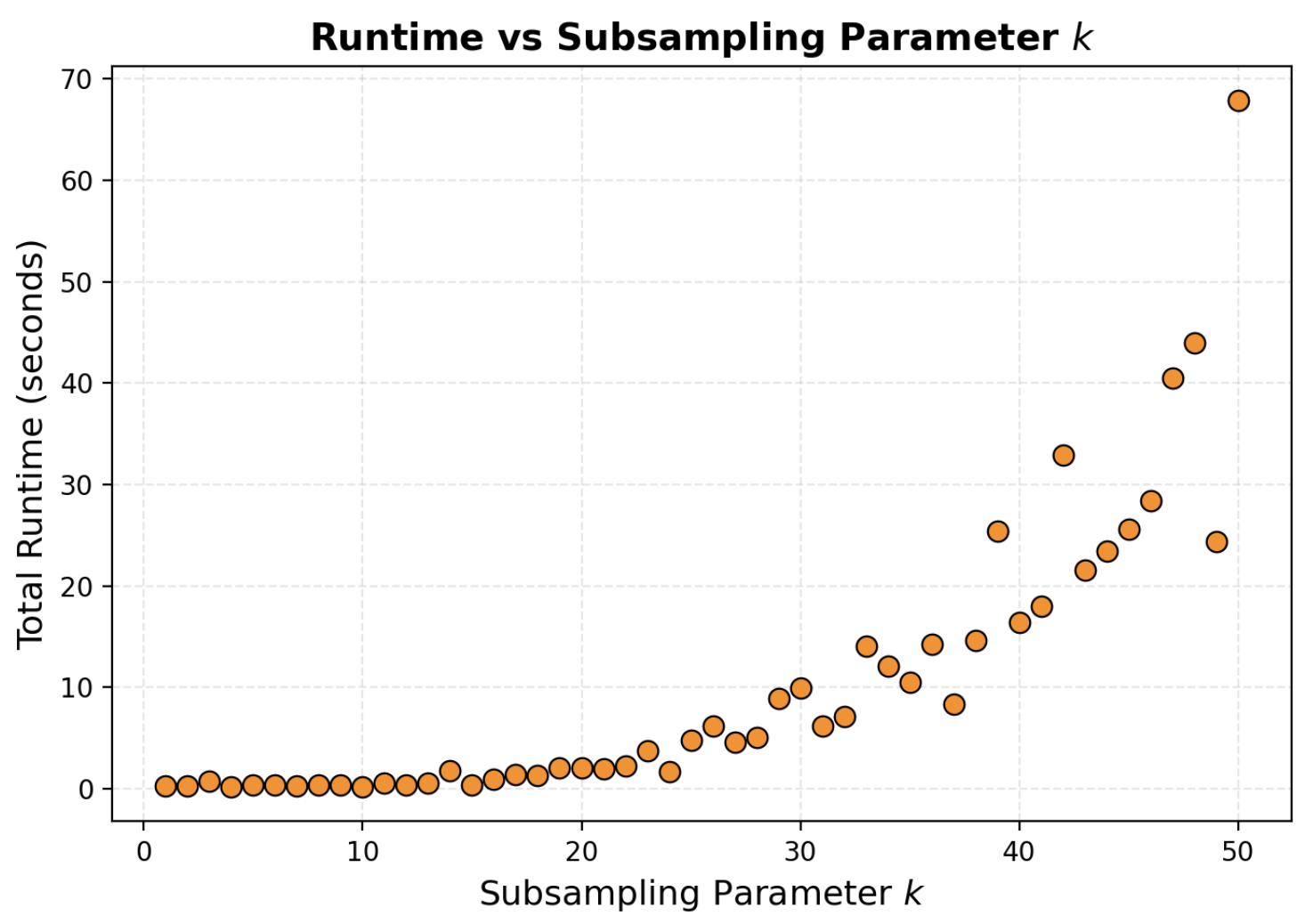}    
         \caption{ a) Discounted cumulative rewards for $k < n = 1000$. As $k$ increases, the rewards generally increase and converge, b) As $k$ increases, the runtime required to converge to a $\tilde{O}(1/\sqrt{k})$-approximate Nash Equilibrium blows up. For some values of $k$, the runtime is significantly shorter since \texttt{ALTERNATING-MARL} learns a policy before $N_{\text{steps}}$ iterations, and terminates early.}
     \label{fig:k increase}
 \end{figure}
\begin{figure}[hbt!]
    \centering
    \includegraphics[width=1\linewidth]{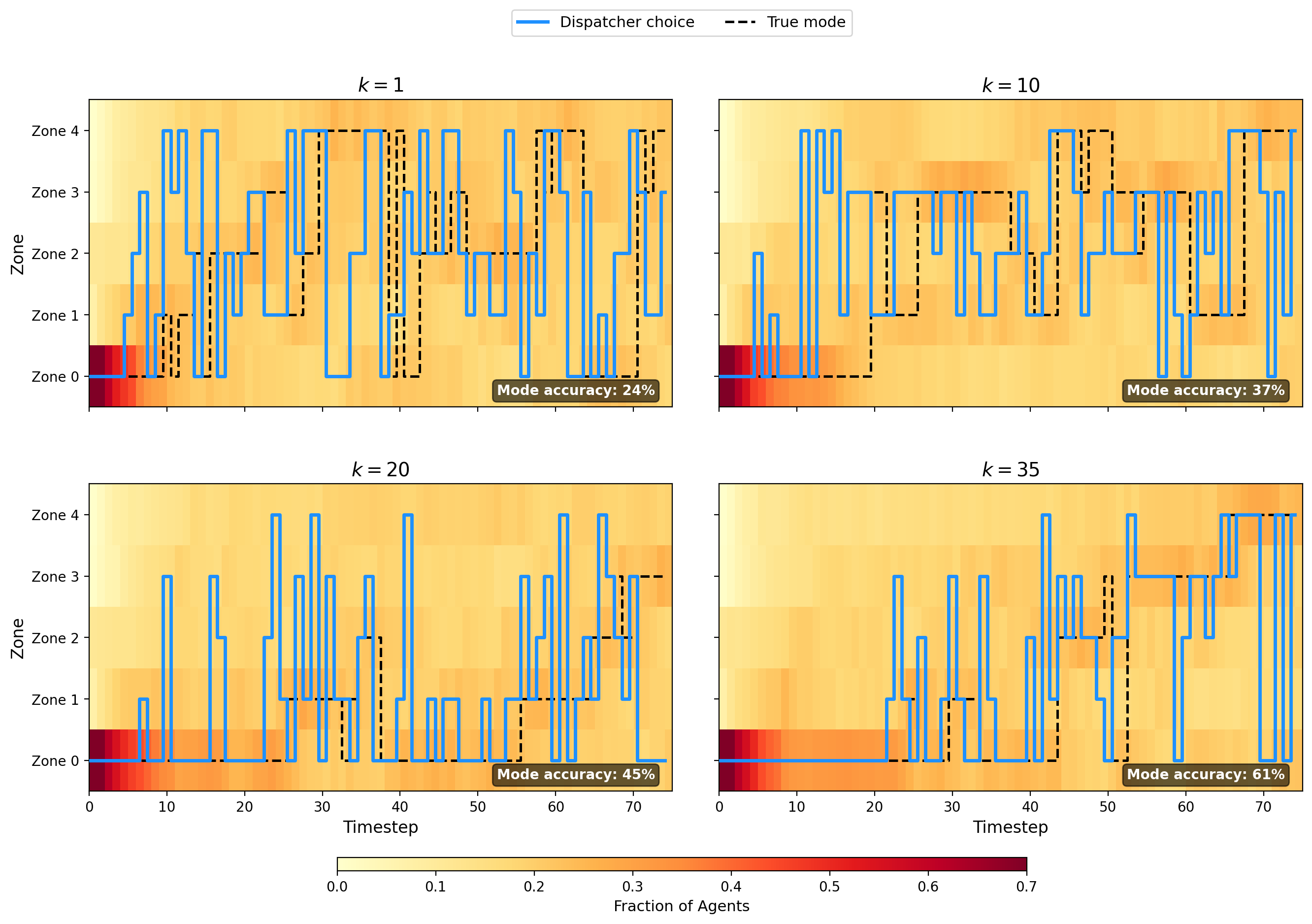}
    \caption{Simulation of the multi-agent warehouse system for subsampling budgets $k = 1, 10, 20, 35$. The heatmap shows the fraction of $n=1000$ robots occupying each of the $5$ zones at each timestep, where darker shades imply a higher concentration. The solid blue line indicates the zone chosen by the dispatcher (global agent), while the dashed black line shows the true population mode (the zone with the most robots). At $k = 1$, the dispatcher's choices differ significantly from the true mode, while at $k = 35$, the dispatcher tracks the mode substantially better, steering resources to the correct zone more than twice as often. The initial agent positions are drawn from $\mathsf{Dirichlet}(0.3)$, creating a concentrated starting configuration, and all panels share the same random seed for comparability.}
    \label{fig:sim result}
\end{figure}

\subsection{Extension to Continuous State Spaces with TRPO}

In this subsection, we provide an extension for \texttt{ALTERNATING-MARL} for the setting where the state spaces of all the agents are continuous (and therefore have uncountably infinite cardinality), and empirically show that the alternating process still converges to an approximate Nash Equilibrium whose value increases monotonically with the subsampling parameter $k$, with decreasing marginal gains, consistent with the result in \cref{fig:k increase}. We begin by describing our experimental setup, providing a formulation of the extension to continuous state spaces, and presenting our experimental results. The practical effectiveness of $\texttt{ALTERNATING-MARL}$ is surprising since there are known theoretical challenges in extending the analysis of the convergence of best-response dynamics in continuous state spaces. Specifically, the cardinality of the policies in continuous state space settings renders the number of interactions, required for the round-robin dynamics to converge, to be infinitely long.

\paragraph{Experimental setup.}  We consider a warehouse coordination setting with a dispatcher (a global agent) who must identify where the majority of robots are clustered and steer resources there, and $n = 500$ robots (local agents) who are spread across a 1D line in two clusters with an (initial) 55/45 majority split. The dispatcher cannot see all robots, but it can only see the states of $k$ randomly sampled robots at each time step. This setting models a continuous state dynamic where the agents can take on state values that lie within a continuous interval.

\paragraph{Dynamics.}  First, the state spaces of the agents are given by $\mathcal{S}_g = [-2,2]$ and $\mathcal{S}_l = [-2, 2]$. Moreover, the action spaces of the agents are given by $\cA_g = \{0,1,2,3,4\}$ and $\cA_l = \{0,1,2\}$. The five actions by the global agent correspond to \emph{bin centers} $\tau_j = -1.6 + 0.8j$ for $j\in \{0,1,\dots,4\}$, and the three local agent actions correspond to $\delta_0 = 0.0, \delta_1 = 0.03,$ and $\delta_2 = -0.03$. The transition for the global agent is given by $s_g' = \mathrm{clip}(s_g + \alpha(\tau_{a_g} - s_g) + \varepsilon_g, -2, 2)$, where $\varepsilon_g \sim \mathcal{N}(0, \sigma_g^2)$, with steering strength $\alpha = 0.6$ and noise $\sigma_g = 0.08$. The transitions for the local agents are given by $s_l' = \mathrm{clip}(s_l + \delta_{a_l} + \varepsilon_l, -2, 2)$ where $\varepsilon_l \sim \mathcal{N}(0, \sigma_l^2)$ with $\sigma_l = 0.04$. 

\paragraph{Rewards.} The global agent's reward $r_g(s_g, a_g)$ rewards the dispatcher for reaching its pointing location, which would enforce that swapping between targets keeps $s_g$ far from every $\tau$ and gives zero rewards, Specifically, $r_g(s_g, a_g) = A_g e^{-{(s_g - \tau_{a_g})^2}/{2\sigma_{r_g}^2}}$, where $A_g = 5$ and $\sigma_{r_g} = 0.3$. Moreover, the local alignment reward is given by $r_l(s_g, s_l, a_l) = 10 + A_l e^{-(s_l - s_g)^2/2\sigma_{r_l}^2}- |a_l|$, which measures a Gaussian alignment with $s_g$ with a movement cost, where $A_l = 10, \sigma_{r_l} = 0.2$, and $c_l = 5$. Together, the joint evaluation reward at each stage is given by $r = r_g(s_g, a_g) + \frac1n \sum_i r_l(s^i_l, s_g, a^i_l)$.

 \begin{figure}
     \centering
     \includegraphics[width=0.7\linewidth]{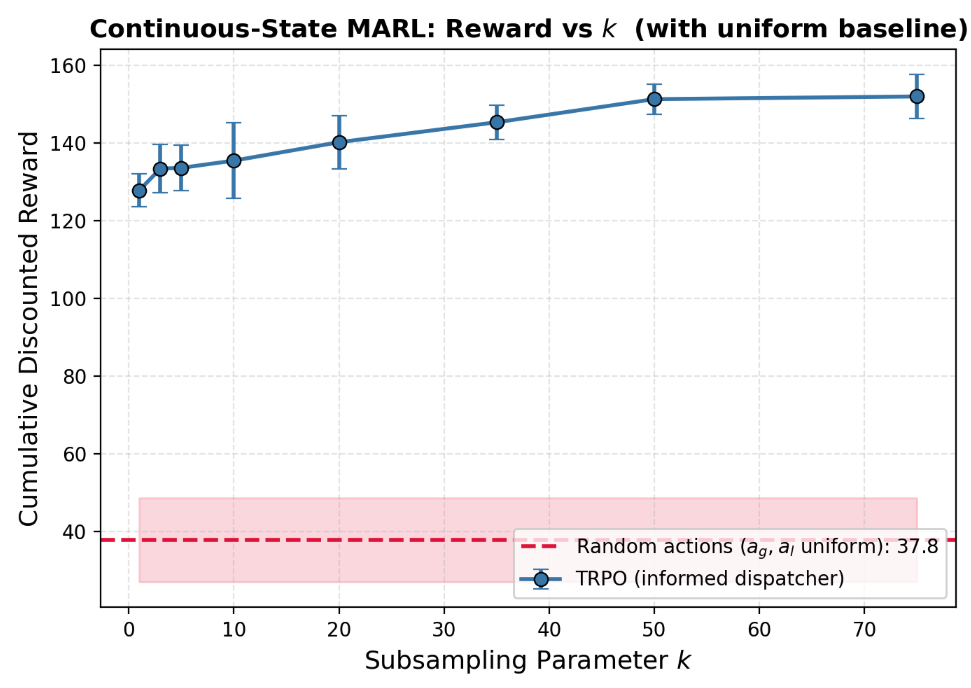}
     \caption{Discounted Cumulative Rewards for $k< n = 500$ for continuous state spaces. We include a baseline of uniformly random actions.}
     \label{fig:reward figures baselines}
 \end{figure}

\paragraph{Algorithm: Alternating Best-Response with TRPO.} We extend the tabular alternating best-response algorithm to continuous state spaces by replacing exact value iteration with neural network–based policy optimization in a two-timescale process. The system comprises a global agent (dispatcher) and $n$ identical local agents (robots), each parameterized by a two-layer feedforward neural network with 64 hidden units and $\tanh$ activations. Training proceeds by alternating between a \texttt{G-LEARN} phase, in which the local policy $\pi_l$ is held fixed and the global policy $\pi_g$ is updated, and an \texttt{L-LEARN} phase, in which the roles are reversed. This alternation is repeated for $K_{\mathrm{outer}}$ iterations. 

For \texttt{G-LEARN}, at each iteration, we generate $M$ independent training samples by resetting the environment and, for each reset, subsampling $k$ agents uniformly at random without replacement from the population of $N$ agents. The positions of the $k$ selected agents are binned into a normalized histogram $h \in \mathbb{R}^B$ over $B$ equally spaced bins covering the state range $[-2, 2]$. The label for each sample is the true population mode $m^* = \arg\max_j h_j^{\mathrm{full}}$, computed from the full population at the time of the initial reset. The global policy network $\pi_g(a \mid h; \theta_g)$, which maps histograms to a distribution over $|\mathcal{A}_g|$ discrete actions, is then trained by minimizing the cross-entropy loss

$$\mathcal{L}_g(\theta_g) = -\frac{1}{M} \sum_{i=1}^{M} \log \pi_g\left(a = m_i^* \mid h_i; \theta_g\right)$$

for $E_g$ steps using the Adam optimizer with learning rate $\eta$ and gradient clipping at norm 1.0. This formulation is a natural function-approximation analog of the model-based value iteration used in the tabular setting: the dispatcher's task at each timestep reduces to a contextual bandit problem where the optimal action depends only on the current histogram observation, not on the history.

For \texttt{L-LEARN}, with the global policy fixed, we collect $H$ episodes of horizon $T$ by simulating the full $N$-agent system, where a representative agent plays $\pi_l$ against the fixed $\pi_g$ and then updating $\pi_l$ with the natural-gradient trust-region step. In each episode, one agent is designated as the representative learner. At each timestep, the representative agent samples an action from $\pi_l(a \mid s_g, s_l;\, \theta_l)$ while all other agents follow the current local policy deterministically. The global agent selects actions according to $\pi_g$ using the subsampled histogram. We record the representative agent's observations, actions, and per-step rewards, and then compute the discounted Monte Carlo returns $G_t = \sum_{\tau>t} \gamma^{\tau - t} r_\tau^l$ and advantage $A_t = G_t - V_\phi(o_t)$ (where $V_\phi$ is the critic), and normalizes them to zero mean and unit variance. The local policy is updated via {TRPO} \cite{pmlr-v37-schulman15,Sutton_McAllester_Singh_Mansour_1999}, where the surrogate objective is given by
\[L(\theta) = \mathbb{E}\left[\tfrac{\pi_\theta(a\mid o)}{\pi_{\theta_{\text{old}}}(a\mid o)} A\right].\]
We compute $g = \nabla_\theta L(\theta)$, and compute the natural gradient. For this, we solve $\mathbf F x = g$, where $\mathbf F$ is the Fisher information matrix (computed as the Hessian of the mean KL-divergence) such that $F\approx \nabla^2 \mathrm{KL}(\pi_{\theta_{\text{old}}}\| \pi_\theta) + \lambda$, where $\lambda$ is a damping coefficient. We use step-sizes of $\beta = \sqrt{2  \delta_{\mathrm{KL}}}$. Finally, we use a backtracking line search where we accept the largest $\theta_{\text{new}} = \theta_{\text{old}} + 0.5^j \cdot \Delta\theta$ for $j=0,1,\dots,9$ that improves $L$ and satisfies $\mathrm{KL}(\pi_{\theta_{\text{old}}}\| \pi_\theta) \leq 1.5 \delta_{\mathrm{KL}}$. This update is performed for $E_l$ epochs. The critic's policy $V_\phi$ is fit separately with MSE regression to $G_t$ over $5$ gradient steps using Adam. Finally, the policy architecture $\pi_l$ is represented by a two-layer feedforward network:
$$\pi_l(a \mid o; \theta_l) = \mathrm{softmax}\Big(W_3 \tanh\big(W_2 \tanh(W_1  o + b_1) + b_2\big) + b_3\Big)$$
with 64 hidden units per layer, taking the 2-dimensional observation $(s_g, s_l)$ as input and producing a categorical distribution over $|\mathcal{A}_l| = 3$ actions.

\newpage

\begin{algorithm}[t]
\begin{algorithmic}
\caption{\texttt{ALTERNATING-MARL} Heuristic with Function Approximation} 
\REQUIRE Transition functions $P_g$ and $P_\ell$, steps $N_{\text{steps}}$,   failure probability $\delta$, algorithms \texttt{G-LEARN} and \texttt{L-LEARN}, and sampling parameters $k$ and $m$. 
\FOR{$t=1,\dots,N_{\text{steps}}$}
\STATE \texttt{\textcolor{blue}{/$\star$ G-LEARN}}
\STATE Sample a new environment $s_0$,
\STATE Subsample $k$ agents, build a histogram $h_\Delta$, record the true mode $m^*$,
\STATE Train the global-agent policy network $\pi_g(a_g|h)$ via cross-entropy loss using Adam
\[\cL_g = -\frac{1}{M}\sum_{i=1}^M \log \pi_g(m_i^*| h_i),\]
\STATE \textcolor{blue}{\texttt{/$\star$ L-LEARN}}
\STATE Collect $H$ episodes and compute $G_t = \sum_{\tau=t}^{T-1} \gamma^{\tau - t} r_\tau^l$,
\STATE Update the local policy network $\pi_l(a_l|s_g, s_l)$ using the \texttt{TRPO} objective \cite{pmlr-v37-schulman15} with loss
\[\mathcal{L}(\theta) = \E\left[\frac{\pi_\theta(a|o)}{\pi_{\theta_{old}}(a|o)} A_t\right],\]
where $A_t = G_t - V_\phi(o_t)$ is the advantage function.
   \ENDFOR
\STATE \textbf{return} $(\pi_g^{N_{\text{steps}}},\pi_\ell^{N_{\text{steps}}})$. 
\label{algorithm: continuous alternating marl}
\end{algorithmic}
\end{algorithm} 
 
\paragraph{Algorithm Configuration.} The hyperparameters used in our implementation are detailed in Table~\ref{tab:hyperparams new}.

\begin{table}[hbt!]
\centering
\begin{tabular}{@{}ll@{}}
\toprule
Hyperparameter & Value \\ \midrule
Population size ($n$) & 500\\
Discount Factor ($\gamma$) & 0.95 \\
Global agent action space size ($|\cA_g|$) & 5 \\
Local agent action space size ($|\cA_l|$) & 3 \\
Steering gain ($\alpha$) & $0.6$ \\
State clipping range & $[-2,2]$ \\
Subsample Sizes ($k$) & $\{1,3, 5, 10, 20, 35, 50, 75\}$ \\
Global transition noise $(\sigma_g)$ & 0.08 \\
Local transition noise $(\sigma_l)$ & 0.04 \\
Global reward amplitude ($A_g$) & $5$ \\
Global Reward Width ($\sigma_{r_g}$) & 0.30 \\
Local reward amplitude ($A_l$) & 10 \\
Local alignment width ($\sigma_{r_l}$) & 0.2 \\
Local action cost & $5$ \\
Monte Carlo Samples ($M$) & 3000 \\
Number of rollouts & 50 per seed \\
Horizon & 100 \\
 Number of alternations ($N_{\text{steps}}$) & 30 \\
Architecture & 2-layer MLP, Tanh activation \\
(Critic and Global) Gradient clip & 1.0 \\
Independent Runs & 15 seeds per $k$ \\
Learning rate & $3\times 10^{-3}$ (Adam) \\
Hidden units & 64 \\
KL trust-region radius ($\delta_{\mathrm{KL}})$ & 0.01 \\
Conjugate-gradient iterations & 10 \\
TRPO damping coefficient $(\lambda)$ & 0.1 \\
Value loss coefficient & 0.5 \\

\bottomrule\\
\end{tabular}
\caption{\texttt{ALTERNATING-MARL} Heuristic Training and Evaluation Configuration}
\label{tab:hyperparams new}
\end{table}

\begin{figure}[hbt!]
    \centering
    \includegraphics[width=0.5\linewidth]{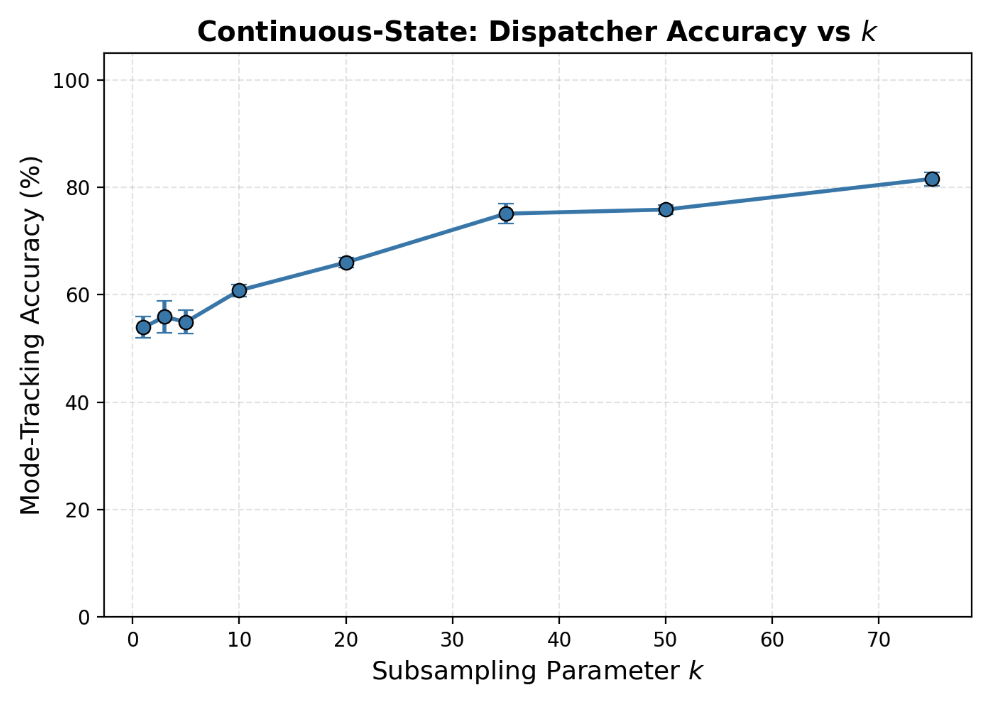}
    \caption{Dispatcher accuracy in the continuous state space setting}
    \label{fig: dispatcher accuracy in the continuous state space setting}
\end{figure}
 \begin{figure}[hbt!]
     \centering
     \includegraphics[width=0.8\linewidth]{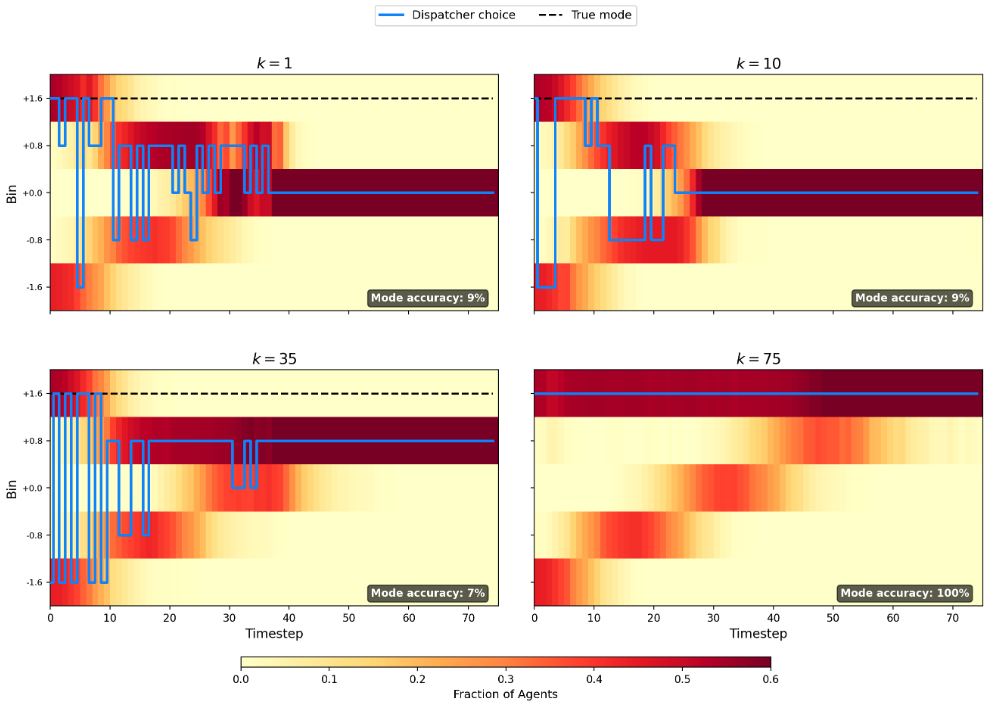}
     \caption{Simulation of the multi-agent warehouse for subsampling budgets $k = 1, 10, 35, 75$, where
the heatmap shows the fraction of $N = 500$ robots occupying the various zones at each timestep
where darker colors imply a higher concentration.}
     \label{fig:placeholder concentration}
 \end{figure}

\newpage

\section{Limitations and Future work}
\label{section: limitations and future work}
We assume that the global and local agents cooperate to optimize a structured reward under a specific dynamic model which was studied in \cite{anand2025meanfield}. A natural direction for future work is to explore the extent to which our algorithms and analysis can be extended to weaker assumptions on the dynamics model or the reward structure. Second, our work is only able to incorporate mild heterogeneity among cooperative agents. It would be interesting to explore the extent to which we can capture a greater degree of heterogeneity \cite{Hu_Wei_Yan_Zhang_2023,cui2022learning,anand2026graphonmeanfieldsubsamplingcooperative,wang2026unsuperviseddecompositionrecombinationdiscriminatordriven}. A third direction would be to prove a theoretical generalization of this work to the setting with infinite/continuous state spaces with general function approximation. This would require more technical algorithms using least-squares value iteration \cite{xu2023tale, ayoub2020model} or linearity assumptions on the underlying MDP \cite{golowich2024linear,  golowich2024the}. Finally, it would be exciting to generalize this result beyond uniform subsampling, which might not be possible in a number of scenarios where it might be easier to subsample closer agents with higher probability, or to subsample among clusters of agents uniformly at random. In this case, we believe our results should extend by using the triangle inequality to bound the distance to the new random distribution over subsampled agents; however, it would be desirable to obtain tighter bounds for such application-specific scenarios. \looseness=-1

% \newpage
% \input{checklist.tex}

\end{document}